\pdfoutput=1

\documentclass[a4paper,11pt]{article}
\usepackage[a4paper,width=150mm,top=25mm,bottom=25mm]{geometry}

\usepackage[cmtip,all]{xy} %commutative diagrams
\usepackage{amsmath}
\usepackage{mathtools}
\usepackage{amsthm}
\usepackage{thmtools}%, thm-restate}
\usepackage{physics}    % better math symbols
\usepackage{amssymb}
\usepackage{amsfonts}
\usepackage{mathdots}   % diagonal dots in other direction
\usepackage{mathrsfs}   % cursive font
\usepackage{nicefrac}   % inline slanted fractions
\usepackage{stmaryrd}
\usepackage{tikz-cd}     % commutative diagrams
\usepackage{xparse}     %allows commands with optional arguments to be defined

\usepackage{amsfonts}
\usepackage{mathtools}
\usepackage{titlesec}

\usepackage{enumitem}
\setlist{noitemsep}
\usepackage{bold-extra}

\usepackage{authblk}

\usepackage{xcolor}

\usepackage{xspace}

\usepackage{tikz}
\usetikzlibrary{decorations.markings}
\usepackage{tikzit}
\usepackage{float}
\usepackage{circuitikz}
\usetikzlibrary{shadows,fadings}
\ctikzset{bipoles/length=.8cm}
\input{styles/symplectic-ZX.tikzdefs}
\newcommand{\arrowangle}{142}
\newcommand{\arrowscaling}{1}

\tikzstyle{graphlabel}=[font={\scriptsize\boldmath}, inner sep=1mm, outer sep=-1.8mm, scale=0.8]
\tikzstyle{graphlines}=[dashed,color=gray]
\tikzstyle{wn}=[font={\scriptsize\boldmath}, inner sep=1mm, outer sep=-1.8mm, scale=0.8, tikzit shape=circle, draw=black, fill=black!01, tikzit fill=white, tikzit draw=black, shape=circle, tikzit category=GLA]
\tikzstyle{bn}=[font={\scriptsize\boldmath}, inner sep=1mm, outer sep=-1.8mm, scale=0.8, tikzit shape=circle, draw=black, fill={rgb,255: red,100; green,100; blue,100}, tikzit draw=black, shape=circle, tikzit category=GLA]
\tikzstyle{gn}=[style=bn, tikzit fill=gray]
\tikzstyle{rn}=[style=wn, tikzit fill=white]
\tikzstyle{grn}=[style=gwn, tikzit fill=white]
\tikzstyle{ggn}=[style=gbn, tikzit fill={zx_grey}]
\tikzstyle{gwn}=[shading=whiteballshading, line width=1pt, inner sep=1mm, outer sep=-1.8mm, scale=0.8, tikzit shape=circle, draw=black, fill=white, tikzit fill=white, tikzit draw=black, shape=circle, tikzit category=GLA]
\tikzstyle{gbn}=[shading=blackballshading, font={\scriptsize\boldmath}, line width=1pt, inner sep=1mm, outer sep=-1.8mm, scale=0.8, tikzit shape=circle, draw=black, fill={rgb,255: red,100; green,100; blue,100}, tikzit draw=black, shape=circle, tikzit category=GLA]
\tikzstyle{had}=[fill={zx_grey}, draw=black, shape=rectangle, tikzit category=ZX, tikzit draw=black, minimum size=5pt, inner sep=1.5pt, scale=1, font={\scriptsize\boldmath}]
\tikzstyle{ghad}=[shading=hadballshading, fill={zx_grey}, draw=black, shape=rectangle, tikzit category=ZX, tikzit draw=black, thick, minimum size=5pt, inner sep=1.5pt, font={\scriptsize\boldmath}]
\tikzstyle{wphase}=[
  node on layer= frontlayer,%always put on front
  drop shadow={blend mode=color burn, opacity=.015, shadow xshift=-0.40pt, shadow yshift=-0.45pt, shadow scale=1.020,color=black},%, left color=black, shading angle=-45,right color=black!75}, %largest shadow
  drop shadow={blend mode=color burn, opacity=.015, shadow xshift=-0.40pt, shadow yshift=-0.45pt, shadow scale=1.015,color=black},% left color=black, shading angle=-45,right color=black!75}, % middle shadow
  drop shadow={blend mode=color burn, opacity=.015, shadow xshift=-0.40pt, shadow yshift=-0.45pt, shadow scale=1.010,color=black},%, left color=black, shading angle=-45,right color=black!75}, %smallest shadow
  rectangle, rounded corners, ultra thin, draw opacity=1, draw=black!50,  inner sep=2pt, font={\scriptsize\boldmath}, label distance=1mm, text opacity=1,minimum height=.3cm,minimum width=.3cm, fill=white, fill opacity=.85, %code for node itself
  tikzit category=ZX, tikzit fill=white, tikzit draw=white  %tikzit code
]
\tikzstyle{bphase}=[
  node on layer= frontlayer, %always put on front
  drop shadow={blend mode=color burn, opacity=.03, shadow xshift=-0.40pt, shadow yshift=-0.45pt, shadow scale=1.020,color=black},%, left color=black, shading angle=-45,right color=black!75}, %largest shadow
  drop shadow={blend mode=color burn, opacity=.03, shadow xshift=-0.40pt, shadow yshift=-0.45pt, shadow scale=1.015,color=black},%, left color=black, shading angle=-45,right color=black!75}, %middle shadow
  drop shadow={blend mode=color burn, opacity=.03, shadow xshift=-0.40pt, shadow yshift=-0.45pt, shadow scale=1.010,color=black},%, left color=black, shading angle=-45,right color=black!75}, %smallest shadow
  rectangle, rounded corners, ultra thin, draw opacity=1, fill opacity=.7, draw=black!70, inner sep=2pt, font={\scriptsize\boldmath}, label distance=1mm, text opacity=1,minimum height=.3cm,minimum width=.3cm, text=white, fill=black!85, %code for node itself
  tikzit category=ZX, tikzit fill=black, tikzit draw=white %tikzit code
]
\tikzstyle{mphase}=[
  node on layer= frontlayer,%always put on front
  drop shadow={blend mode=color burn, opacity=.025, shadow xshift=-0.40pt, shadow yshift=-0.45pt, shadow scale=1.020,color=black},%, left color=black, shading angle=-45,right color=black!75}, %largest shadow
  drop shadow={blend mode=color burn, opacity=.025, shadow xshift=-0.40pt, shadow yshift=-0.45pt, shadow scale=1.015,color=black},%, left color=black, shading angle=-45,right color=black!75}, %middle shadow
  drop shadow={blend mode=color burn, opacity=.025, shadow xshift=-0.40pt, shadow yshift=-0.45pt, shadow scale=1.010,color=black},%, left color=black, shading angle=-45,right color=black!75}, %smallest shadow
  rectangle, rounded corners, ultra thin, draw opacity=1, draw=black!60, inner sep=2pt, font={\scriptsize\boldmath}, label distance=1mm, fill opacity=.75, text opacity=1,minimum height=.3cm,minimum width=.3cm, fill=black!25,  %code for node itself
  tikzit category=ZX, tikzit fill=gray, tikzit draw=white %tikzit code
]
\tikzstyle{gphase}=[style=bphase,tikzit category=ZX, tikzit fill=black, tikzit draw=white]
\tikzstyle{rphase}=[style=wphase,tikzit category=ZX, tikzit fill=white, tikzit draw=white]
\tikzstyle{scalar}=[mphase]
\tikzstyle{lmat}=[shape=signal, signal to=west, signal from=east, fill={zx_grey}, draw=black, minimum height=6pt, inner sep=1pt, font={\scriptsize\boldmath}, tikzit fill=gray, tikzit category=GLA, anchor=center, outer sep=-.1cm, signal pointer angle=\arrowangle,scale=\arrowscaling]
\tikzstyle{rmat}=[shape=signal, signal to=east, signal from=west, fill={zx_grey}, draw=black, minimum height=6pt, inner sep=1pt, font={\scriptsize\boldmath}, tikzit fill=gray, tikzit category=GLA, anchor=center, outer sep=-.1cm, signal pointer angle=\arrowangle,scale=\arrowscaling]
\tikzstyle{dmat}=[shape=signal, signal to=east, signal from=west, fill={zx_grey}, draw=black, minimum height=6pt, inner sep=1pt, font={\scriptsize\boldmath}, tikzit fill=gray, tikzit category=GLA, rotate=270, anchor=center, outer sep=-.1cm, signal pointer angle=\arrowangle,scale=\arrowscaling]
\tikzstyle{umat}=[shape=signal, signal to=east, signal from=west, fill={zx_grey}, draw=black, minimum height=6pt, inner sep=1pt, font={\scriptsize\boldmath}, tikzit fill=gray, tikzit category=GLA, rotate=90, anchor=center, outer sep=-.1cm, signal pointer angle=\arrowangle,scale=\arrowscaling]
\tikzstyle{lmatt}=[shading=matballshading,shape=signal, signal to=west, signal from=east, fill={zx_grey}, draw=black, minimum height=6pt, inner sep=1pt, font={\scriptsize\boldmath}, tikzit fill=gray, tikzit category=GLA, anchor=center, outer sep=-.1cm, thick, signal pointer angle=\arrowangle,scale=\arrowscaling]
\tikzstyle{rmatt}=[shading=matballshading,shape=signal, signal to=east, signal from=west, fill={zx_grey}, draw=black, minimum height=6pt, inner sep=1pt, font={\scriptsize\boldmath}, tikzit fill=gray, tikzit category=GLA, anchor=center, outer sep=-.1cm, thick, signal pointer angle=\arrowangle,scale=\arrowscaling]
\tikzstyle{dmatt}=[shading=matballshading,shape=signal, signal to=east, signal from=west, fill={zx_grey}, draw=black, minimum height=6pt, inner sep=1pt, font={\scriptsize\boldmath}, tikzit fill=gray, tikzit category=GLA, rotate=270, anchor=center, outer sep=-.1cm, thick, signal pointer angle=\arrowangle,scale=\arrowscaling]
\tikzstyle{umatt}=[shading=matballshading, outer color={zx_grey_thick}, inner color={zx_grey},shape=signal, signal to=east, signal from=west, fill={zx_grey}, draw=black, minimum height=6pt, inner sep=1pt, font={\scriptsize\boldmath}, tikzit fill=gray, tikzit category=GLA, rotate=90, anchor=center, outer sep=-.1cm, thick, signal pointer angle=\arrowangle,scale=\arrowscaling]
\tikzstyle{graph_vertex}=[fill=black, draw=black, shape=circle, tikzit category=mbqc, minimum size=2.4mm, inner sep=.8mm]
\tikzstyle{graph_weight}=[fill=white, draw=none, shape=rectangle, tikzit category=mbqc, inner sep=2pt, scale=.8]
\tikzstyle{graph_state}=[fill=yellow, draw=none, shape=rectangle, tikzit category=ZX, rounded corners=1.3mm, minimum height=1.7cm, minimum width=1.3cm, opacity=.7, text opacity=1]
\tikzstyle{box}=[fill=white, draw=black, shape=rectangle, inner sep=2.5pt]
\tikzstyle{wirelable}=[node on layer= labeltextlayer, font={\scriptsize},scale=.9,text=dark_grey, fill=none, inner sep=1pt]
\tikzstyle{tightwirelable}=[node on layer= labeltextlayer, text=dark_grey, fill=white, inner sep=0pt]
\tikzstyle{gather}=[shading=gatherballshading, outer color={zx_grey_thick}, inner color={zx_grey}, fill={zx_grey}, draw=black, tikzit category=scal, rounded corners=0.8mm, regular polygon, regular polygon sides=3, shape border rotate=-90, inner sep=1.6pt, anchor=center, outer sep=-.1cm, line width=0.75 pt]
\tikzstyle{divide}=[shading=divideballshading, outer color={zx_grey_thick}, inner color={zx_grey}, regular polygon, regular polygon sides=3, shape border rotate=90, draw=black, fill={zx_grey}, inner sep=1.6pt, tikzit category=scal, rounded corners=0.8mm, anchor=center, outer sep=-.1cm,line width=0.75 pt]
\tikzstyle{plus}=[inner sep=2.5pt, draw, circle, path picture={ \draw[black](path picture bounding box.east) -- (path picture bounding box.west) (path picture bounding box.south) -- (path picture bounding box.north);}, fill=white]
\tikzstyle{dot}=[thick, fill=black, circle, scale=1, inner sep=.05cm]
\tikzstyle{origin}=[thick, fill=black, circle, scale=.5, inner sep=.05cm]
\tikzstyle{dash_edge}=[-, dashed]
\tikzstyle{hadamard_edge}=[-, dashed, dash pattern=on 2pt off 1.5pt, thick, draw=blue]
\tikzstyle{brace edge}=[-, tikzit draw=blue, decorate, decoration={brace,amplitude=1mm,raise=-1mm}]
\tikzstyle{ultra thin}=[-, line width=0.03 pt]
\tikzstyle{thin}=[-, line width=0.5 pt]
\tikzstyle{thick}=[-, line width=.8pt, tikzit draw=red]
\tikzstyle{very thick}=[-, line width=1.05pt, tikzit draw=red]
\tikzstyle{multiplexer}=[-, fill={rgb,255: red,179; green,179; blue,179}]
\tikzstyle{gmultiplexer}=[-, line width=1pt, tikzit draw=red, fill={rgb,255: red,179; green,179; blue,179}]
\tikzstyle{background}=[-, fill={rgb,255: red,250; green,250; blue,250}, draw=none, tikzit draw={rgb,255: red,128; green,128; blue,128}, on layer= backlayer]
\tikzstyle{white}=[-, fill=white, draw=none, tikzit fill=white]
\tikzstyle{blue_line}=[-, draw=blue]
\tikzstyle{bbox}=[-, fill=white]
\tikzstyle{shadedpolygon}=[rounded corners=0.2mm, opacity=1, fill=gray, fill opacity=0.5]
\tikzstyle{arrow}=[->]

\input{styles/styles-pbs.tikzdefs}
\tikzstyle{diamant}=[diamond, fill=couleurdefond, draw=black]
\tikzstyle{newe}=[rectangle, fill={gray!15}, draw=black, tikzit shape=rectangle, inner sep=0.2em]
\tikzstyle{cercle}=[circle, fill=couleurdefond, draw=black]
\tikzstyle{scercle}=[circle, fill=couleurdefond, draw=black, tikzit fill=white, inner sep=0.1em]
\tikzstyle{cartouche}=[rounded rectangle, fill=couleurdefond, draw=black]
\tikzstyle{neg}=[rounded rectangle, fill=couleurdefond, draw=black, execute at end node={$\neg$}]
\tikzstyle{sneg}=[rounded rectangle, fill=couleurdefond, draw=black, execute at end node={$\neg$}, scale=0.8]
\tikzstyle{negserie}=[rounded rectangle, fill=couleurdefond, draw=black, execute at end node={\footnotesize$\star\star$}]
\tikzstyle{diagrammevide}=[rectangle, fill=couleurdefond, draw=black, inner sep=1.25em, borddiagrammevide, tikzit shape=rectangle]
\tikzstyle{mdiagrammevide}=[rectangle, fill=couleurdefond, draw=black, inner sep=0.75em, sborddiagrammevide, tikzit shape=rectangle]
\tikzstyle{msdiagrammevide}=[rectangle, fill=couleurdefond, draw=black, inner sep=0.7em, msborddiagrammevide, tikzit shape=rectangle]
\tikzstyle{sdiagrammevide}=[rectangle, fill=couleurdefond, draw=black, inner sep=0.5em, sborddiagrammevide, tikzit shape=rectangle]
\tikzstyle{xsdiagrammevide}=[rectangle, fill=couleurdefond, draw=black, inner sep=0.4em, xsborddiagrammevide, tikzit shape=rectangle]
\tikzstyle{bs}=[shape=beam, fill=couleurdefond, draw, inner sep=0.25em, thick, tikzit fill=white]
\tikzstyle{sbs}=[shape=beam, fill=couleurdefond, draw, inner sep=0.2em, thick, tikzit fill=white]
\tikzstyle{npbs}=[shape=beam, horizontal fill={{npbsmoitiebasse}{npbsmoitiehaute}}, draw, inner sep=0.25em, thick, tikzit fill={rgb,255: red,128; green,128; blue,128}]
\tikzstyle{npbsalenvers}=[shape=beam, horizontal fill={{npbsmoitiehaute}{npbsmoitiebasse}}, draw, inner sep=0.25em, thick, tikzit fill={rgb,255: red,128; green,128; blue,128}]
\tikzstyle{snpbs}=[shape=beam, horizontal fill={{npbsmoitiebasse}{npbsmoitiehaute}}, draw, inner sep=0.2em, thick, tikzit fill={rgb,255: red,128; green,128; blue,128}]
\tikzstyle{snpbsalenvers}=[shape=beam, horizontal fill={{npbsmoitiehaute}{npbsmoitiebasse}}, draw, inner sep=0.2em, thick, tikzit fill={rgb,255: red,128; green,128; blue,128}]
\tikzstyle{cnot}=[shape=circle, draw, path picture={ 
\draw[black](path picture bounding box.north) -- (path picture bounding box.south) (path picture bounding box.west) -- (path picture bounding box.east);
}]
\tikzstyle{thickcnot}=[shape=circle, draw, thick, path picture={ 
\draw[thick,black](path picture bounding box.north) -- (path picture bounding box.south) (path picture bounding box.west) -- (path picture bounding box.east);
}]
\tikzstyle{boite22}=[fill=white, draw=black, shape=rectangle, minimum height=1cm, minimum width=0.5cm]
\tikzstyle{boite15}=[fill=white, draw=black, shape=rectangle, minimum height=0.7cm, minimum width=0.5cm]
\tikzstyle{boite2}=[fill=white, draw=black, shape=rectangle, minimum height=0cm, minimum width=0cm]
\tikzstyle{snegpotentiel}=[fill=couleurdefond, draw=black, shape=rounded rectangle, inner sep=0.25em, tikzit fill={rgb,255: red,191; green,191; blue,191}, execute at end node={\footnotesize$\star$}]
\tikzstyle{negpotentiel}=[fill=couleurdefond, draw=black, shape=rounded rectangle, tikzit fill={rgb,255: red,191; green,191; blue,191}, execute at end node={$\star$}]
\tikzstyle{token}=[fill=black, draw=black, shape=circle, inner sep=0.1em]
\tikzstyle{whitetoken}=[fill=white, draw=black, shape=circle, inner sep=0.1em]
\tikzstyle{boitePBS}=[fill=white, draw=gray, thick, shape=rectangle, rounded corners=3pt, minimum height=0.6cm, inner sep=0.1em, minimum width=0.5cm]
\tikzstyle{boitePBS2}=[fill=white, draw=gray, thick, shape=rectangle, rounded corners=3pt, minimum height=0.55cm, inner sep=0.1em, minimum width=0.5cm]
\tikzstyle{sgene}=[fill={gray!30}, draw=black, shape=rounded rectangle, rounded rectangle east arc=0pt, minimum height=0.5cm, inner sep=0em, minimum width=0cm, scale=0.8]
\tikzstyle{sdetector}=[fill={gray!30}, draw=black, shape=rounded rectangle, rounded rectangle west arc=0pt, minimum height=0.5cm, inner sep=0em, minimum width=0cm, scale=0.8]
\tikzstyle{xsgene}=[fill={gray!30}, draw=black, shape=rounded rectangle, rounded rectangle east arc=0pt, minimum height=0.5cm, inner sep=0em, minimum width=0cm, scale=0.67]
\tikzstyle{xsdetector}=[fill={gray!30}, draw=black, shape=rounded rectangle, rounded rectangle west arc=0pt, minimum height=0.5cm, inner sep=0em, minimum width=0cm, scale=0.67]
\tikzstyle{PolRot}=[fill=white, draw=black, shape=rectangle, minimum height=0.5cm, inner sep=0.1em, minimum width=0.1cm]
\tikzstyle{PolRotrouge}=[fill={red!3}, draw=red, shape=rectangle, minimum height=0.5cm, inner sep=0.1em, minimum width=0.1cm, tikzit fill=red, execute at begin node={\textcolor{red}\bgroup}, execute at end node={\egroup}]
\tikzstyle{PhS}=[fill=white, draw=black, shape=rectangle, minimum height=0.5cm, inner sep=0.1em, minimum width=0.1cm]
\tikzstyle{gene}=[fill=white, draw=black, shape=rounded rectangle, rounded rectangle east arc=0pt, minimum height=0.5cm, inner sep=0em, minimum width=0cm]
\tikzstyle{detector}=[fill=white, draw=black, shape=rounded rectangle, rounded rectangle west arc=0pt, minimum height=0.5cm, inner sep=0em, minimum width=0cm]
\tikzstyle{cartoucherouge}=[rounded rectangle, fill={red!55!white}, draw=black, tikzit fill=red]
\tikzstyle{cartouchebleu}=[rounded rectangle, fill={blue!33!white}, draw=black, tikzit fill=blue]
\tikzstyle{diamantrouge}=[diamond, fill={rgb,255: red,255; green,115; blue,115}, draw=black]
\tikzstyle{diamantbleu}=[diamond, fill={rgb,255: red,171; green,171; blue,255}, draw=black]

\tikzstyle{new}=[-]
\tikzstyle{tirets}=[-, draw=black, dashed]
\tikzstyle{noire}=[-, draw=black]
\tikzstyle{ep}=[-, draw=black]
\tikzstyle{longdashed}=[-, dash pattern=on 5pt off 5pt]
\tikzstyle{pointilles}=[-, draw=black, dotted]
\tikzstyle{grise}=[-, draw={rgb,255: red,191; green,191; blue,191}]
\tikzstyle{rouge}=[-, draw=red]
\tikzstyle{bleue}=[-, draw=bleu, tikzit draw=blue]
\tikzstyle{verte}=[-, draw={rgb,255: red,0; green,230; blue,0}]
\tikzstyle{borddiagrammevide}=[-, dash pattern=on 0.5em off 0.5em on 0.5em off 0.5em on 0.5em off 0em]
\tikzstyle{msborddiagrammevide}=[-, dash pattern=on 0.28em off 0.28em on 0.28em off 0.28em on 0.28em off 0em]
\tikzstyle{sborddiagrammevide}=[-, dash pattern=on 0.2em off 0.2em on 0.2em off 0.2em on 0.2em off 0em]
\tikzstyle{xsborddiagrammevide}=[-, dash pattern=on 0.1em off 0.1em on 0.15em off 0.1em on 0.1em off 0em]
\tikzstyle{mediumdash}=[-, dash pattern=on 2pt off 2pt]
\tikzstyle{rougefonce}=[-, draw={red!50!black}, tikzit draw={rgb,255: red,136; green,0; blue,0}]
\tikzstyle{background}=[-, fill={rgb,255: red,229; green,229; blue,229}, draw=none, tikzit draw={rgb,255: red,128; green,128; blue,128}]

\usepackage[
  style=alphabetic,
  sorting=nyt,
  backend=biber,
  maxnames=6,
  maxcitenames=2,
  url=true,
  doi=true,
  isbn=false,
  eprint=false
]{biblatex}

\usepackage{hyperref}
\hypersetup{
  colorlinks=true,
  linkcolor=blue,
  citecolor=gray
}

\ifx\supresscomments\undefined%
  \newcommand{\robert}[1]{%
    \noindent{\color{purple} \textsf{[RIB: #1]}}
  }
  \newcommand{\cole}[1]{%
    \noindent{\color{blue} \textsf{[CC: #1]}}
  }
  \newcommand{\titouan}[1]{%
    \noindent{\color{red} \textsf{[TC: #1]}}
  }
\else%
  \newcommand{\robert}[1]{}
  \newcommand{\cole}[1]{}
  \newcommand{\titouan}[1]{}
\fi

\newcommand{\Aff}{%
  \mathsf{Aff}
}
\newcommand{\Lin}{%
  \mathsf{Lin}
}
\newcommand{\Lag}{%
  \mathsf{Lag}
}
\newcommand{\Co}{%
  \mathsf{Co}
}
\newcommand{\Isot}{%
  \mathsf{Isot}
}

\newcommand{\zx}{%
  \mathsf{GSA}
}
\newcommand{\gla}{%
  \mathsf{GLA}
}
\newcommand{\gaa}{%
  \mathsf{GAA}
}
\newcommand{\stab}{%
  \mathsf{Stab}
}

\newcommand{\im}{%
  \operatorname{im}
}

\newcommand{\interp}[1]{%
  \left\llbracket #1 \right\rrbracket
}
\newcommand{\trans}{%
  \mathsf{T}
}
\DeclareRobustCommand{\disc}{%
  {\scalebox{.5}{\input{../figures/AffCoisoRel/discard_small.tikz}}}
}

\newcommand{\N}{%
  \mathbb{N}
}

\newcommand{\Zp}{%
  {\mathbb{F}_p}
}
\newcommand{\K}{%
  \mathbb{K}
}
\newcommand{\Q}{%
  \mathbb{Q}
}
\newcommand{\R}{%
  \mathbb{R}
}
\newcommand{\C}{%
  \mathbb{C}
}

\newcommand{\LOv}{%
  \mathsf{LOv}
}
\newcommand{\ECirc}{%
  \mathsf{ECirc}
}

\NewDocumentCommand{\Rel}{O{X}}{%
  \mathsf{Rel}_{#1}
}
\NewDocumentCommand{\Symp}{O{\K}}{%
  \mathsf{Symp}_{#1}
}
\NewDocumentCommand{\ASymp}{O{\K}}{%
  {\Aff}\Symp[#1]
}
\NewDocumentCommand{\AR}{O{\K}}{%
  {\Aff}\Rel[#1]
}
\NewDocumentCommand{\lR}{O{\K}}{%
  {\Lin}\Rel[#1]
}
\NewDocumentCommand{\LR}{O{\K}}{%
  {\Lag}\Rel[#1]
}
\NewDocumentCommand{\IR}{O{\K}}{%
  {\Isot}\Rel[#1]
}
\NewDocumentCommand{\CR}{O{\K}}{%
  {\Co}\IR[#1]
}
\NewDocumentCommand{\ALR}{O{\K}}{%
  {\Aff}\LR[#1]
}
\NewDocumentCommand{\AIR}{O{\K}}{%
  {\Aff}\IR[#1]
}
\NewDocumentCommand{\ACR}{O{\K}}{%
  {\Aff}\CR[#1]
}
\NewDocumentCommand{\ZX}{O{\K}}{%
  \zx_{#1}
}
\NewDocumentCommand{\ZXdisc}{O{\K}}{%
  \ZX[#1]^\disc
}
\NewDocumentCommand{\GLA}{O{\K}}{%
  \gla_{#1}  
}
\NewDocumentCommand{\GAA}{O{\K}}{%
  \gaa_{#1}
}
\NewDocumentCommand{\Stab}{O{p}}{%
  \stab_{#1}
}

\newcommand{\RX}{%
  \Rel[X]
}
\newcommand{\RK}{%
  \Rel[\K]
}

\newcommand{\abbrar}{%
  \mathsf{AR}
}

\newcommand{\abbralagr}{%
  \mathsf{ALR}
}

\NewDocumentCommand{\Matrices}{ O{m} O{n} O{\K} }{%
  \operatorname{Mat}_{#3}(#2,#1)
}
\NewDocumentCommand{\Sym}{O{n} O{\K} }{%
  \operatorname{Sym}_{#1}(#2)
}

\newcommand{\bvdots}{%
  \tikz[baseline, every node/.style={inner sep=0}]{ \node at (0,0){.}; \node at (0,-6pt){.}; \node at (0,6pt){.}; }
}

\newcommand{\diagram}{%
  \(\ZX\)-diagram\xspace
}

\newlength\oversetwidth
\newlength\underwidth
\newcommand\alignedoverset[2]{
  \settowidth\oversetwidth{$\overset{#1}{#2}$}
  \settowidth\underwidth{$#2$}
  \setlength\oversetwidth{\oversetwidth-\underwidth}
  \hspace{.5\oversetwidth}
  &
  \settowidth\oversetwidth{$\overset{#1}{#2}$}
  \settowidth\underwidth{$#2$}
  \setlength\oversetwidth{\oversetwidth-\underwidth}
  \hspace{-.5\oversetwidth}
  \overset{#1}{#2}
}

\newcommand{\stackedrefs}[1]{%
  \begingroup\renewcommand*{\arraystretch}{.5}\begin{matrix}#1\end{matrix}\endgroup
}

\newcommand{\stackeqmid}[1]{%
  \stackrel{\stackedrefs{#1}}{=}
}
\newcommand{\stackeq}[1]{%
  \alignedoverset{\stackedrefs{#1}}{=}
}

\renewcommand{\phi}{%
  \varphi
}
\renewcommand{\epsilon}{%
  \varepsilon
}

\renewcommand{\leq}{%
  \leqslant
}
\renewcommand{\geq}{%
  \geqslant
}
\newtheorem{theorem}{Theorem}
\newtheorem{proposition}[theorem]{Proposition}
\newtheorem{lemma}[theorem]{Lemma}
\newtheorem{corollary}[theorem]{Corollary}

\theoremstyle{definition}
\newtheorem{example}[theorem]{Example}
\newtheorem{remark}[theorem]{Remark}
\newtheorem{definition}[theorem]{Definition}

\newcommand{\minipropref}[1]{\text{\tikzpropref{#1}}}
\newcommand{\minilemref}[1]{\text{\tikzlemref{#1}}}
\newcommand{\minieqref}[1]{\text{\tikzeqref{#1}}}
\newcommand{\mini}[1]{\tikzrefsize{#1}}

\declaretheoremstyle[
headfont=\normalfont\bfseries\color{gray},
bodyfont=\normalfont,
notefont=\normalfont\bfseries\color{gray},
notebraces={}{},
headpunct={.},
qed=\qedsymbol,
mdframed={
  linewidth=1.5,
  linecolor=gray,
  hidealllines=true,
  leftline=true,
  skipabove=0,
  innerleftmargin=2mm,
  innerrightmargin=0,
  innertopmargin=0,
  innerbottommargin=.7mm
}
]{line_proof}

\begin{document}

\title{\bfseries Graphical Symplectic Algebra}
\date{}

\author[1,2]{Robert I. Booth}
\author[3]{Titouan Carette}
\author[4,5]{Cole Comfort}
\affil[1]{
  University of Edinburgh, United Kingdom
}
\affil[2]{
  University of Bristol, United Kingdom
}
\affil[3]{
  LIX, CNRS, École polytechnique, Institut Polytechnique de Paris, 91120 Palaiseau, France
}
\affil[4]{
  Department of Computer Science, University of Oxford, United Kingdom
}
\affil[5]{
  Université de Lorraine, CNRS, Inria, LORIA, F 54000 Nancy, France
}

\maketitle    
 
\begin{abstract}
We give complete presentations for the \dag-compact props of affine Lagrangian and coisotropic relations over an arbitrary field. This provides a unified family of graphical languages for both affinely constrained classical mechanical systems, as well as odd-prime-dimensional stabiliser quantum circuits. To this end, we present affine Lagrangian relations by a particular class of undirected coloured graphs. In order to reason about composite systems, \emph{we introduce a powerful scalable notation} where the vertices of these graphs are themselves coloured by graphs. In the setting of stabiliser quantum mechanics, this scalable notation gives an extremely concise description of graph states, which can be composed via ``phased spider fusion.''  Likewise, in the classical mechanical setting of electrical circuits, we show that impedance matrices for reciprocal networks are presented in essentially the same way.
  \vspace{1mm}
\end{abstract}

\section{Introduction}
\label{sec:intro}
Symmetric monoidal categories formalize  process theories with well-behaved notions of sequential and parallel composition. Just as groups can often be presented in terms of generators and equations, so can symmetric monoidal categories. However, in the symmetric monoidal setting, the generators are interpreted as {\em boxes} with wires dangling out of them, and the equations are instead between {\em string diagrams} which are built by wiring these boxes together.  Therefore,  unlike for a group, a presentation for a symmetric monoidal category allows one to reason about the underlying process theory using only diagrammatic rewrite rules: appealing to human beings' innate topological intuitions.

\paragraph{Graphical algebra} \  \newline
The research programme broadly construed as ``graphical algebra'' has produced presentations for categories of relations equipped with algebraic structure, where the monoidal product is given by the disjoint union/direct sum.  For example, the symmetric monoidal  categories of linear relations \cite{gla}, affine relations \cite{gaa}, piecewise linear relations \cite{dpla}, and polyhedral relations \cite{dpa} have all been given presentations. These aforementioned presentations have notably been used as a syntax for signal flow diagrams in control theory \cite{Bonchi2019,control,Bonchi2017,Bonchi2014, coya,erbele}, systems with bounded resources \cite{dpa}, and electrical circuits \cite{gaa,Boisseau2022,network,passive,dpla,coya}.

\paragraph{Categorical quantum mechanics} \  \newline
On the other hand, consider  the process theory of quantum circuits.
Formally, quantum circuits are string diagrams in the symmetric monoidal category of Hilbert spaces with the bilinear tensor as the monoidal product.
The research programme of ``categorical quantum mechanics'' has produced numerous presentations for fragments of quantum circuits. Examples of which include the ZX-calculus \cite{zx,zxcompletea,zxcompleteb,backens_zx-calculus_2014,qutrit}, ZW-calculus \cite{zw}, ZH-calculus \cite{zh}, ZX\&-calculus \cite{zxa} and their various incarnations.   Such presentations have been used for for the design of new quantum circuit optimization techniques \cite{Kissinger2020,deBeaudrap2020,duncan_graph-theoretic_2020,Cowtan2020} as well as quantum circuit verification \cite{Kissinger2020,Lemonnier2021}.  There are many more applications, which are oulined in the review article of van de Wetering \cite{wetering}.

\paragraph{Graphical algebra meets categorical quantum mechanics} \  \newline
Structural similarities between the research programmes of graphical algebra and categorical quantum mechanics have been noticed for quite some time. Specifically, the symmetric monoidal category of linear relations over \(\mathbb{F}_2\) was observed by many people to be isomorphic to the phase-free fragment of the qubit ZX-calculus modulo scalars. This is written down explicitly in the PhD thesis of  \textcite[page~64]{ih} as well as in the work of \textcite[page~12]{control}. This connection has been exploited to reason about \(\mathbb{F}_2\)-linear subspaces in the qubit ZX-calculus, for example in the work of \textcite{carette_szx-calculus_2019}.  A deeper connection between both research programmes was conjectured in the PhD thesis of \textcite[page~146]{ih}. However, some structures of the ZX-calculus, such as the Hadamard gate/Fourier transform and phases were believed by some to be too inherently quantum to be captured in this way.

Zanasi's conjecture was later shown to be true by \textcite{comfort_graphical_2021} who proved that the symmetric monoidal category of affine Lagrangian relations over \(\Zp\) is isomorphic to the odd-prime-dimensional qudit stabiliser ZX-calculus, modulo scalars. Affine Lagrangian relations can be seen to generalize linear relations, and it can accommodate for more notions in quantum circuits such as the Fourier transform and Clifford phases.   Later Comfort showed that one can take this correspondence even further by extending the relational semantics to affine coisotropic relations: in the quantum semantics, this corresponds to adding quantum discarding \cite{comfort_algebra_2023}. However, neither paper provides presentations for these symmetric monoidal categories.

\paragraph{Quantum mechanics meets classical mechanics} \  \newline
Before these developments of Comfort et al. in the relational semantics of quantum circuits, Baez et al., had already interpreted electrical circuits are in the symmetric monoidal category of affine/linear Lagrangian relations over the fields of real and real rational functions \cite{passive,network}.  This is also discussed in the PhD thesis of \textcite{coya} who was also a coauthor in \cite{network}. In this research programme they refer to this interpretation as a {\em``black-box''} which determines the behaviour of the  electrical circuit, but which one can not look inside of: imparting an air of mystery on the compositional structure of affine Lagrangian relations.

In fact, categories of Lagrangian relations have a rich history in the context of the geometric quantisation of classical mechanics 
\cite{guillemin_problems_1979, weinstein_symplectic_2009, benenti_linear_1983, lawruk_special_1975, tulczyjew_category_1984, dold_symplectic_1982, benenti_category_1982, urbanski_structure_1985, kijowski_symplectic_1979}. 
So in retrospect, this connection between quantum and classical mechanics should have not been surprising. However, the rich research programme of geometric quantisation is generally not concerned with string diagrams, therefore, they also do not provide graphical presentations for these symmetric monoidal categories.

\paragraph{Our contribution} \  \newline
In this article  we give presentations for the symmetric monoidal categories of linear Lagrangian relations, affine Lagrangian relations, linear coisotropic relations and affine coisotropic relations over arbitrary fields: the prop of affine coisotropic relations being the most general one. These are symmetric monoidal subcategories of linear and affine relations; however they are quite special. They can be constructed by first doubling the dimension of the objects in linear and affine relations and then adding  ``shearing operations'' in both dimensions. As mentioned earlier, these symplectic categories of relations give semantics for the phase-space picture of both quantum and classical mechanics. Therefore, the presentations we give here yield a unifying treatment for classical mechanical circuits as well as a fragment of quantum circuits.

Therefore, from a conceptual point of view, in this paper we {\em ``open up the black box''} of Baez et al. which was hitherto closed.  By giving presentations for these symplectic categories of relations, we reveal their underlying compositional structure: exposing some deep  connections between quantum and classical mechanical circuits. Moreover, by exploring the structure of the strictification of these categories, we obtain powerful scalable notation for composite systems. This allows us to give elegant and concise descriptions of the normal forms for stabiliser circuits, as well succinct descriptions of notions like graph states and impedance matrices for reciprocal networks.

\paragraph{Structure of paper} \  \newline
In section~\ref{sec:strings} we review string diagrams for symmetric monoidal categories and their close relatives. 

In section~\ref{sec:GAA} we give a presentation for the symmetric monoidal category of affine relations over an arbitrary field.  The presentation which we gives differs slightly from that of \textcite{gaa}.  This is because we use a different set of generators inspired by the ZX-calculus which allows us to represent circuits in terms of undirected coloured graphs.  Using this presentation, we review how some basic notions in linear and affine algebra can be stated graphically.

In section~\ref{sec:AffLagRel} we prove the core result of the paper: providing a presentation for the symmetric monoidal category of affine Lagrangian relations over an arbitrary field.  We begin this journey in subsection~\ref{ssec:symplectic_linear_algebra} by defining the category of Lagrangian relations and giving a brief review of linear/affine symplectic geometry. Following this, in subsection~\ref{ssec:symplectic_ZX}, we give generators and equations for Lagrangian relations, without yet proving that this gives a complete presentation.  In subsection~\ref{ssec:scalable_symplectic}, we introduce scalable notation for affine Lagrangian relations where our generators are generalized to composite systems so that the graphs themselves become coloured by graphs.  In subsection~\ref{ssec:completeness} we use this scalable notation to establish a normal form for affine Lagrangian relations, from which it follows that the generators and equations form a complete presentation. 

In section~\ref{sec:AffCoisoRel} we extend the results of section~\ref{sec:AffLagRel} to the more general setting of affine coisotropic relations. Semantically, this is given by adding the relation which discards the plane; syntactically, this is given by adding a generator which discards isometries.  This connects coisotropic reduction in symplectic geometry to the essential uniqueness of purification in quantum mechanics.

In section~\ref{sec:Applications} we conclude the paper by giving  various applications of our calculus to concrete physical systems, both in the classical and quantum realm.

In appendix~\ref{app:axioms}, we give axiom tables summarizing all of the presentations which we develop in this paper.

We put the gory technical details of the proofs, as well as technical lemmas in appendix~\ref{app:proofs}.

\newpage
\tableofcontents
\newpage

\section{Graphical languages}
\label{sec:strings}
As mentioned in the previous section, symmetric monoidal categories give mathematical semantics for processes with a certain notion of sequential and parallel composition.  A reference for symmetric monoidal categories can be found in Mac Lane   \cite[\S XI]{maclane}.  
Throughout this section, let \((\mathsf{C},\otimes,I)\) denote a symmetric monoidal category with tensor product $(-)\otimes(=):\mathsf{C}\times\mathsf{C}\to\mathsf{C}$ and tensor unit $I \in \operatorname{Ob}(\mathsf{C})$.  Denote the structure maps of the symmetric monoidal category  as follows:
\begin{itemize}
  \item \textbf{associators}: \hfil \(\alpha_{X,Y,Z} : X \otimes (Y \otimes Z) \cong (X \otimes Y) \otimes Z\)
  \item \textbf{unitors}: \hfil \(\lambda_X : I \otimes X \cong X\) \quad\quad {\it and} \quad\quad \(\rho_X: X \otimes I \cong X\)
  \item \textbf{symmetries}: \hfil \(\sigma_{X,Y} : X \otimes Y \cong Y \otimes X\)
\end{itemize}

\subsection{String diagrams}
A symmetric monoidal category is {\em strict} when the unitors and associators are identities.  Strict symmetric monoidal categories are very convenient to work with because they enjoy a graphical calculus in terms of \emph{string
diagrams}.

An arrow \(f : X_1 \otimes \cdots \otimes X_m \to Y_1 \otimes
\dots \otimes Y_n\) is drawn as a box with \(m\) input wires on the left,
and \(n\) output wires on the right, where each wire is labeled by the corresponding object:
\begin{equation}
  \input{figures/strings/morphism.tikz}
\end{equation}
 We will often drop the
labels on wires when the objects can be inferred from context. The identity on the monoidal
unit is drawn as an empty diagram, the identity on all other objects as a plain wire, and the components of the symmetry are drawn as wires crossing:
\begin{equation}
  \input{figures/strings/monoidal_unit.tikz}
  \quad\quad\quad
  \input{figures/strings/identity.tikz}
  \quad\quad\quad
  \input{figures/strings/symmetry.tikz}
\end{equation}
The composition of  arrows is drawn by ``plugging wires'' and the monoidal product of
arrows as ``vertical juxtaposition'':
\begin{equation}
  \input{figures/strings/composition.tikz}
  \quad\quad\quad
  \input{figures/strings/monoidal_product.tikz}
\end{equation}
String diagrams for symmetric monoidal categories can be manipulated as one might expect:
\begin{equation}
  \input{figures/strings/symmetry_equations.tikz}
\end{equation}

Arrows of type \(I\to X\) are called {\em states}, arrows of type \(X\to I\) are called {\em effects}, and endomorphisms \(I\to I\) are called {\em scalars}.

A \emph{prop} is a strict symmetric monoidal category which is generated by a single object under the tensor product.
In other words, the monoid of objects under the tensor product is isomorphic the the natural numbers \(\N\).
Exploiting this isomorphism, we can always drop the wire labels because the types of arrows are uniquely identified by the number of input and output wires. For example,
an arrow \(f : 3 \to 2\) in a prop is represented as:
\begin{equation}
  \input{figures/strings/prop_morphism.tikz}
\end{equation}

\subsection{Strictification and scalable notation}

Without loss of generality, using the following well-known result, it is always possible to use string diagrams to
reason about (not necessarily strict) symmetric monoidal categories:
\begin{theorem}
  Every symmetric monoidal category is symmetric monoidally equivalent to a strict symmetric monoidal category.
\end{theorem}
This process is called \emph{strictification}. There is an explicit strictification
procedure for \emph{monoidal} categories described by \textcite{wilson} presented in terms of generators and relations.
Similar results exist in the context of multiplicative linear logic: for example by \textcite{robin}. Although the
specific symmetric monoidal strictification procedure we use here has not explicitly appeared in the literature,
it follows straightforwardly from the aforementioned results:
\begin{definition}
  \label{def:strict}
  Given a symmetric monoidal category \((\mathsf{C},\otimes,I)\), its \emph{strictification} \(([\mathsf{C}],+,[])\)  has:
  \begin{itemize}
    \item \textbf{Objects:}  in \(\operatorname{List}(\mathrm{Ob}(\mathsf{C}))\).
    \item \textbf{Arrows:} For every arrow \(f:X\to Y\) in \(\mathsf{C}\) there is an induced arrow \([f]:[X]\to[Y]\) in \([\mathsf{C}]\).  Moreover, for all objects \( X,Y\) in \(\mathsf C\) we have the following arrows in  \([\mathsf{C}]\):

      \begin{tabular}{rlrl}
        \emph{Dividers:} &   \input{figures/strings/tensor_dividers.tikz} &
        \emph{Gatherers:}  & \input{figures/strings/tensor_gatherers.tikz}
      \end{tabular}\\~

      We impose the following equations between arrows:
      \begin{equation}
        \label{eq:functor}
         [f]\circ [g] = [f\circ g], \ \ [1_X]=1_{[X]}
      \end{equation}
      \begin{equation}
        \label{eq:strings_tensor_equations}
        \input{figures/strings/tensor_equations.tikz}
      \end{equation}

    \item \textbf{Strict monoidal structure:} The tensor unit is the empty list \([\ ]\).  The tensor product \(+\) is list concatenation.
    \item \textbf{Symmetries:} given by permutations of list elements.
  \end{itemize}

  The \emph{strictification functor}  \([-] : (\mathsf{C},\otimes,I) \simeq ([\mathsf{C}],+,[\ ])\) takes arrows \(f:X\to Y\) to \([f]:[X]\to [Y]\).  For this reason we overload notation and take \([f]=f\). Equations \eqref{eq:functor} witness the functoriality and  equations  \eqref{eq:strings_tensor_equations} witness the preservation of the symmetric monoidal structure. The dividers and gatherers are the components of the monoidal laxators for the strictification functor.
\end{definition}

\begin{remark}
  The components of the associator, unitors and symmetry in \(\mathsf{C}\) are in are internalized to the following string diagrams in \([\mathsf{C}]\):
  \begin{equation}
    \input{figures/strings/tensor_structure_maps.tikz}
  \end{equation}
  where their inverses are given by flipping the string diagrams around horizontally.
\end{remark}

\paragraph{Scalable notation} \  \newline
In this paper, we only work with props.  Even though props are already strict, it will still be useful to work with their strictifications.
This allows for inductive arguments to be made using only string diagrams.  To this end, first remark that in the strictification of a strict symmetric monoidal category, the laxators and oplaxators obey the following equations:
\begin{equation}
  \label{eq:scalablemonoid}
  \input{figures/strings/monoid.tikz}
    \quad
  \input{figures/strings/comonoid.tikz}
\end{equation}
In order to make inductive arguments for props using only pictures, following Carette et al. \cite{carette_szx-calculus_2019, carette_note_2020, carette2022large-scale_2022}, as a matter of convention denote the object \([1]\) by a thin wire, and the objects \([n]\) for \(n \in \N\) by thick wires (often labelled by their type \(n\)).
For example, we use this notation to describe block matrices string digrammatically in subsection~\ref{ssec:scalable_gaa}.

Call this collection of graphical notation {\em scalable notation}.

\subsection{Duality and flexsymmetry}
In this subsection, we extend symmetric monoidal categories with a notion of feedback:
\begin{definition}
  A symmetric monoidal category is {\em compact-closed} when  for every object \(X\), there exists an object \(X^*\) together with a map \(\eta_X : I \to X \otimes X^*\) called the {\em cup} and a map  \(\epsilon_X : X^* \otimes X \to I\) called the {\em cap}, drawn as follows:
  \begin{equation}
    \input{figures/strings/cup.tikz}
    \quad\qand\quad
    \input{figures/strings/cap.tikz}
  \end{equation}
  We require that \((X\otimes Y)^*=X^*\otimes Y^*\) and \(I^*=I\).  We also require the cups and caps to satisfy the following equations:
  \begin{equation}
    \input{figures/strings/string_pulling.tikz}
  \end{equation}

  A {\em compact prop} is a compact-closed prop for which \(1^*=1\) and where cup and cap are symmetric, meaning that they satisfy $\epsilon_1 \circ \sigma_{1,1} = \epsilon_1 $ and $\sigma_{1,1} \circ \eta_1  = \eta_1 $.
\end{definition}

\begin{definition}
 \label{def:flexsym}
  Given a compact prop \((\mathsf{C},+ ,0)\), an arrow \(f : m \to n\) is \emph{flexsymmetric} if for any permutation \(\varsigma\) of \(m+n\) wires:
  \begin{equation}
    \input{figures/strings/flexsymmetry.tikz}
  \end{equation}
\end{definition}

Notice that the cup and the cap are flexsymmetric. In some cases, compact props admit flexsymmetric presentations, that is, a presentation where all generators are flexsymmetric. In other words, the string diagrams can be treated as labeled graphs, and we no longer need to worry about the orientation of edges and arrows: only the connectivity of the diagram matters! This is very desirable from a computer science perspective, because it greatly simplifies the combinatorics of representing string diagrams with a data structure. More details on flexsymmetric presentations can be found in the PhD thesis of Carette \cite{carettethesis}.

\paragraph{Dagger structure} \  \newline
In many settings, a process \(X\to Y\) is canonically associated to a process of the ``reversed type'' \(Y\to X \).  However, this map need not be the inverse.  This idea can be captured as follows:

\begin{definition}
  \label{def:dag}
  A {\em \dag-compact prop} is a compact prop \((\mathsf{C}, \otimes, I)\) equipped with a functor of props \( (-)^\dag : \mathsf{C}^\mathsf{op}\to  \mathsf{C} \) such that:
  \begin{itemize}
    \item For all objects \( X \) in \(\mathsf{C}\), \(X^\dag = X\).
    \item For all maps \(f:X\to Y\) in \(\mathsf{C}\),  \( (f^\dag)^\dag=f  \).
    \item \(\sigma_{1,1}^\dag= \sigma_{1,1}^{-1}\) and \(\eta_1^\dag=\epsilon_1\).
  \end{itemize}
  In such a setting, an arrow \(f:X\to Y\) is an:
  \begin{itemize}
    \item {\em isometry} when \(f^\dag  f = 1_X\).
    \item {\em coisometry} when \(f f^\dag = 1_Y\).
    \item {\em unitary} when it is both an isometry and coisometry (so that \(f^\dag = f^{-1}\)).
  \end{itemize}
\end{definition}

\section{Graphical linear and affine algebra}
\label{sec:GAA}
In this section we review string diagrammatic presentations of the props of affine and linear subspaces by Bonchi et al. \cite{gla,gaa}, which we will respectively refer to as Graphical Linear Algebra ({\sf GLA})  and Graphical Affine Algebra ({\sf GAA}). {\sf GAA} can be seen as a strict extension of GLA, thus we will first give its presentation and then restrict it to {\sf GLA}. Note that the presentations which we give follow directly but differ slightly from those contained in the work of Bonchi et al. \cite{gla,gaa}.

\subsection{Props of linear and affine relations}
\label{subsec:linrelaffrel}
Given some fixed set \(X\), an \(X\)-relation \(R:n\to m\) is a subset of the product \(R\subseteq X^n \times X^m\).  These form a compact prop:

\begin{definition}
  \label{def:xrel}
  Given a set \(X\), the compact prop of {\em\(X\)-relations}, \(\RX\), has:
  \begin{itemize}
    \item \textbf{Arrows} are relations between \(X^n\) and \(X^m\).
    \item \textbf{Composition} of relations \(R:n\to m\) and \(S:m\to k\) is defined as:
    \[ S\circ R \coloneqq  \{ (x,z) \in X^{n}\times X^k \ | \ \exists y \in X^m: (x,y) \in R \wedge (y,z) \in S  \}:n\to k \]
    \item \textbf{Identity} on \(n\) is given by the diagonal relation \(\{(x,x) \in  X^n\times X^n \}:n\to n \)
    \item \textbf{Monoidal structure} is given by the object \(0\) and the Cartesian product:
      \[ S\times T \coloneqq \left\{\left(\begin{bmatrix}x\\y \end{bmatrix},\begin{bmatrix}z\\w\end{bmatrix}\right)  \,\middle|\,   (x,z) \in S\wedge  (y,w)\in T \right\} \]
    \item \textbf{Symmetric structure} is generated by the relation: 
      \[\sigma_{1,1}\coloneqq \left\{\left(\begin{bmatrix}x\\y \end{bmatrix},\begin{bmatrix}y\\x\end{bmatrix}\right)  \,\middle|\,  x,y\in X\right\}:2\to 2 \]
    \item \textbf{Compact structure} is generated by the relations: \[\eta_1\coloneqq \left\{\left(\bullet,\begin{bmatrix}x\\x\end{bmatrix}\right)  \,\middle|\,  x\in X\right\}:0\to 2\quad\text{and}\quad\epsilon_1\coloneqq \left\{\left(\begin{bmatrix}x\\x\end{bmatrix}, \bullet\right)  \,\middle|\,  x\in X\right\}:2\to 0\]
  \end{itemize}
\end{definition}

Given a function \(f:X^n\to X^m\) then its graph is a relation \(\Gamma_f\coloneqq \{(x,f(x))\ |\ x\in X \}:n\to m\). This assignment is functorial, so that the composition of functions corresponds to the relational composition of their graphs.

We restrict this prop to the linear and affine settings. To this end, first recall some basic linear and affine algebra:

\begin{definition}
  Given a field \(\K\) and some \(n\in\N\), a {\em \(\K\)-affine combination} of vectors \(\vec v, \vec w \in \K^m\) is a vector \(a \cdot \vec v + (1-a)\cdot \vec w \) for some \(a\in\K\). More generally, {\em \(\K\)-linear combination} of vectors \(\vec v, \vec w \in \K^m\) is a vector \(a \cdot \vec v + b\cdot \vec w \) for some \(a,b\in\K\).

  A {\em \(\K\)-affine subspace} of \(\K^n\) is subset \(V\subseteq \K^n\) which is closed under affine combinations. Similarly, a {\em \(\K\)-linear subspace} of \(\K^n\) is nonempty subset \(V\subseteq \K^n\) which is closed under linear combinations.

\end{definition}

Given any nonempty affine subspace \(V\subseteq \K^n\), there exists a  vector \(\vec a\in\K^n\)  such that the coset \( V+ \vec a \coloneqq \{\vec v+\vec a \ | \ \forall \vec v \in V \}\subseteq \K^n \) is a linear subspace.  Call this vector, the {\em affine shift} associated to \(V+\vec a\).
Conversely, every linear subspace is an affine subspace, as affine combinations are particular linear combinations.  Conceptually, affine subspaces are linear subspaces without a fixed origin. This justifies why we don't require affine subspaces to be inhabited: since there is no fixed origin, it need not contain the zero vector.

Both of these notions of subspaces induce a compact prop:

\begin{definition}
  Given a field \(\K\), the category, \(\AR\) (resp. \(\lR\)), of {\em \(\K\)-affine (resp. linear) relations} is the sub-prop of  \(\RK\) whose arrows \(n\to m\) are \(\K\)-affine (resp. linear) subspaces of the vector space \(\K^n\times\K^m \cong \K^{n}\oplus \K^m\).
\end{definition}

This is a category because affine (resp. linear) subspaces are closed under relational composition, and that the diagonal subset is a linear subspace. This is a compact prop because the cup, the cap and symmetry are also linear subspaces.

By construction we have \(\lR \hookrightarrow \AR \hookrightarrow \RK\). In \(\lR\) and \(\AR\), the arrows \(0 \to n\) are exactly linear and affine subspaces of \(\K^n\).  In \(\lR\) there is a unique scalar \(1_0:0\to 0\), corresponding to the linear subspace \(\K^0 \subset \K^0\). But in \(\AR\) we have two scalars, the identity and the empty affine subspace \(\emptyset\subset \K^1\).

\subsection{Graphical affine algebra and its interpretation}
In this subsection, we review the graphical presentations for the props of affine and linear relations of Bonchi et al \cite{gaa,gla}.

Let \(\K\) be an arbitrary field, and  \(\K^*\) is the group of invertible elements of \(\K\) under multiplication.

We start with affine relations, and then restrict  to linear relations:
\begin{definition}
\label{def:gaa}
  Let \(\GAA\) be the compact prop where the arrows are generated by  {\em grey and white spiders}, as well as {\em scalar multiplication} for every \(a \in \K\):

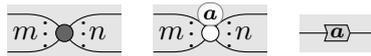
  \[
    \input{figures/AffRel/b-spider.tikz}\quad
    \input{figures/AffRel/w-spider.tikz}\quad
    \input{figures/AffRel/scalar.tikz}\quad
  \]
  Modulo the equations making the grey and white spiders flexsymmetric as well as the equations, for all \(a,b\in \K\) and \(c \in \K^*\):
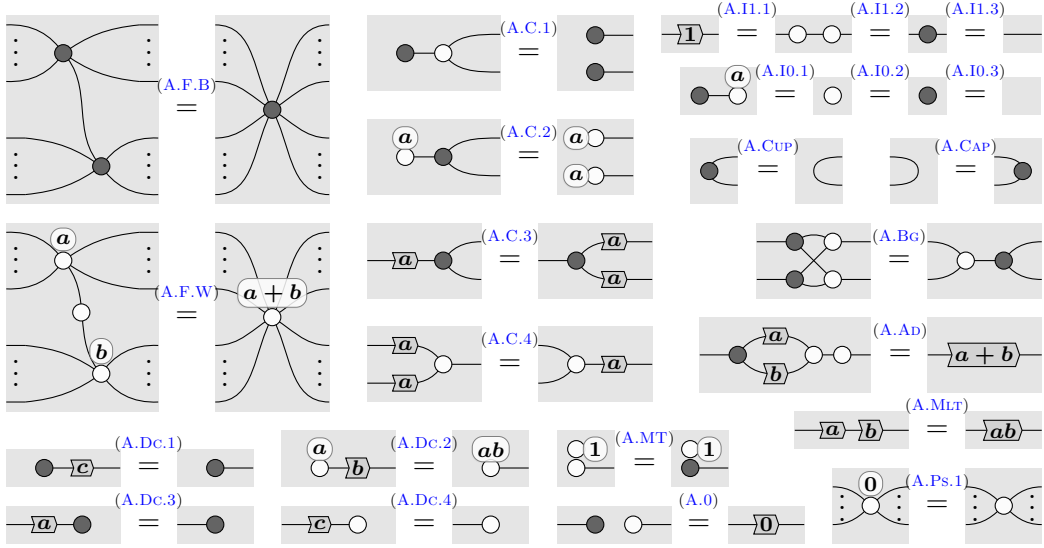
  \[\input{figures/AffRel/axioms.tikz}\]

  With the derived generators, for all \(a\in\K\):
  \[
    \input{figures/AffRel/scalar-daggered.tikz}
    \coloneqq 
    \input{figures/AffRel/scalar-twisted.tikz}
  \]

\end{definition}

\begin{theorem}
  There is an isomorphism of compact props  \(\interp{-}^{\gaa}_{\abbrar}:\GAA \cong \AR\)  defined on generators, for all \(n,m\in\N\) and \(a\in \K\):
  \begin{align*}
    \interp{\input{figures/AffRel/b-spider.tikz}}^{\gaa}_{\abbrar}
      &\coloneqq
         \left\{
           \left(
             \begingroup
               \renewcommand*{\arraystretch}{.7}
               \begin{bmatrix}
                 a\vspace*{-.2cm}\\ \vdots \\ a
               \end{bmatrix}
               ,
               \begin{bmatrix}
                 a\vspace*{-.2cm}\\ \vdots \\ a
               \end{bmatrix}
             \endgroup
           \right) \in \K^{n}\oplus \K^m
           \ \middle|\
           a \in \K 
         \right\} \\
    \interp{ \input{figures/AffRel/w-spider.tikz}}^{\gaa}_{\abbrar}
      &\coloneqq 
        \left\{ 
          (\vec b, \vec c)
          \in \K^n \oplus \K^m
          \ \middle|\
          \sum_{j=0}^{n-1} b_j + \sum_{k=0}^{m-1} c_k  = a 
        \right\} \\
    \interp{ \input{figures/AffRel/mult.tikz}}^{\gaa}_{\abbrar}
    &\coloneqq \left\{ \left(b,ab \right) \ \middle|\ b \in \K \right\}
  \end{align*}
\end{theorem}

The notions of white and grey spiders labelled by a phase-group will occur throughout the rest of this paper in different incarnations:

\begin{definition}
\label{def:phasegroup}
  Given a group \(\mathcal{G}=(G,+,0)\), the compact prop of {\em \(\mathcal{G}\)-labelled white spiders} is generated by, for each \(a \in G\) and \(n,m \in \N \):
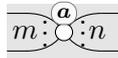
  \[\input{figures/AffRel/labelled-w-spider.tikz}\]
  modulo flexsymmetry (see definition~\ref{def:flexsym}) and the following equations for all \(a,b\in G \):
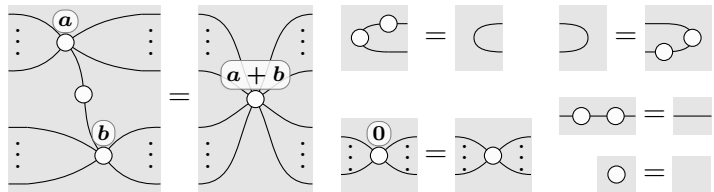
  \[\input{figures/AffRel/labelled-w-spider-axioms.tikz}\]

  Similarly, the compact prop of {\em \(\mathcal{G}\)-labelled grey spiders} is generated by, for each \(a \in G\):
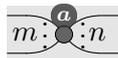
  \[\input{figures/AffRel/labelled-b-spider.tikz}\]
  modulo flexsymmetry and the following equations for all \(a,b\in G \) and \(n,m \in \N \):
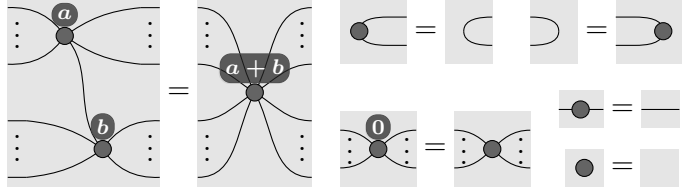
  \[\input{figures/AffRel/labelled-b-spider-axioms.tikz}\]
  In both cases, say that \(\mathcal{G}\) is the {\em phase group} of the white and grey spiders.
\end{definition}

Therefore, we see that the white phase group of \(\GAA\) is the additive group of \(\K\); whereas the grey phase group is the trivial group.  As opposed to \(\AR\), we find that \(\lR\) has a more symmetric presentation:

\begin{corollary}
  By restricting the phase group of the  grey spider in \(\GAA\) to be the trivial group, and removing the rules involving non-zero phases, we obtain a presentation \(\GLA\) for \(\lR\).
\end{corollary}

Note that our presentations  \(\GAA\) of \(\AR\) and \(\GLA\) of \(\lR\) differ from their original presentations by Bonchi et al. \cite{gla,gaa}. The original presentation of affine relations is recovered from ours via the translation:
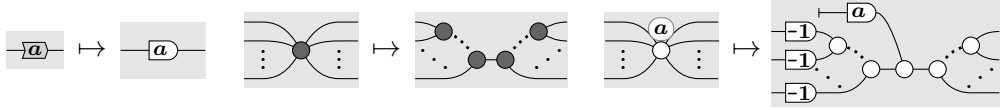
\[ \input{figures/AffRel/traduction.tikz} \]

Where the original presentation of  \(\lR\) can be recoverd by restricting \(a=0\).

\begin{remark}
\label{rem:notflexsym}
In our presentations \(\GAA\) and \(\GLA\) both grey and white spiders are flexsymmetric generators. In particular, this means that affine relations over \(\Q\) and \(\Zp\) for prime \(p\) admit flexsymmetric presentations because all scalars are generated by the scalars \(1\) and \(0\) under addition, multiplication and division.
\end{remark}

This flexsymmetric presentation makes typesetting diagrams much easier, as one doesn't have to worry about the direction which wires are connected to generators; however, it will be useful to keep both presentations in mind, as the original non-flexsymmetric presentations are more natural from an algebraic perspective. For example, the concrete semantics of the constructions in the following subsection are more easily interpreted in the non-flexsymmetric presentation.

\subsection{Scalable notation for graphical linear and affine algebra}
\label{ssec:scalable_gaa}
In this subsection we give notations for the strictification of \(\GAA\). This allows for many concepts
in linear and affine algebra to be expressed concisely using pictures.

\begin{definition}
  Scalable spiders are defined by induction on the colour of wire \(k\in\N\).
  The base case is given by the identity on the tensor unit:
  \begin{equation}
    \input{figures/AffRel/scaledspidef_basecase.tikz}
  \end{equation}
  The inductive step is given by:
  \begin{equation}
    \label{eq:GAAthickspiders}
    \input{figures/AffRel/scaledspidef.tikz}
  \end{equation}
\end{definition}

These scalable spiders satisfy analagous axioms as before, except now the white \([k]\)-coloured spiders are indexed by the additive group of \(\K^k\).

As a matter of notation, for any \(n,m\in\N\) take \(\Matrices\) to be the set of \(m \times n\) matrices over \(\K\).
\begin{definition}
  \label{def:matrix_arrow}
Block matrices are defined by induction on the domain and codomain.
  The base cases are given by the unique matrices \(0_{0,0}\in \Matrices[0][0]\), \(0_{0,1}:\in \Matrices[0][1]\) and \(0_{1,0}\in \Matrices[1][0]\):
  \begin{equation}
    \input{figures/AffRel/matrixarrow_basecase.tikz}
  \end{equation}
  For the first inductive step given a vector \(\vec v \in \K^m\), and an element \(a\in\K\) we can build the matrix arrow  for \( [a  \ {\vec v}^\trans]^\trans \in \K^{m+1}\):
  \begin{equation}
    \label{eq:blockmatrixinductivestep0}
    \input{figures/AffRel/matarrowdef.tikz}
  \end{equation}

  For the second inductive step, given a matrix \(A\in \Matrices[m][n]\) and a vector  \(\vec v\in \K^m \), we can build the matrix arrow for  \([\vec v \ A ] \in \Matrices[m][n+1]\):
  \begin{equation}
    \label{eq:blockmatrixinductivestep1}
    \input{figures/AffRel/matarrowdef1.tikz}
  \end{equation}
\end{definition}

Unrolling this inductive definition, a block matrix divided once both horizontally and vertically is represented by:
\begin{equation}
  \label{eq:arrow_decomposition}
  \input{figures/AffRel/blockmatex.tikz}
\end{equation}

So that any affine transformation \((A,\vec v);\ x\mapsto Ax + \vec v \) is represented by the following diagram:

\begin{equation}
  \label{eq:affinematrix}
  \input{figures/AffRel/affinematex.tikz}
\end{equation}

Note that the extra white nodes on the right hand sides of equations~\eqref{eq:blockmatrixinductivestep1}, \eqref{eq:arrow_decomposition} and \eqref{eq:affinematrix} are artifacts of our flexsymmetric presentation. From this representation the following rules directly follow:

\begin{equation}
	\label{matrules}\input{figures/AffRel/matrules.tikz}
\end{equation}

As a matter of notation, just like the converse of scalar multiplication, the converse of a matrix \(A\) is given by:

\begin{equation}
  \input{figures/AffRel/matrix-daggered.tikz}
  \coloneqq 
  \input{figures/AffRel/matrix-twisted.tikz}
\end{equation}

Using this notation, we can depict the image and the kernel of a matrix \(A\) using pictures:

\begin{equation}
  \interp{\input{figures/AffRel/matarrowspace1.tikz}}^{\gaa}_{\abbrar}=\ker(A)\quad\interp{\input{figures/AffRel/matarrowspace2.tikz}}^{\gaa}_{\abbrar}=\im(A)
\end{equation}

And more generally the affine subspace of solutions for \(\vec v\) in the  equation \(A\vec v=\vec w\):

\begin{equation}
  \interp{\input{figures/AffRel/matarrowspace3.tikz}}^{\gaa}_{\abbrar}=\left\{\vec w \,\middle|\,  A \vec w=\vec v\right\}
\end{equation}

Note that in this relational setting, it is perfectly fine to ask for the the image or kernel of an affine {\em relation} by replacing the matrix arrows with arbitrary maps.

We end this section by recalling diagrammatic reformulation of injectivity and surjectivity:

\begin{proposition}\label{prop:matprop}
  For any matrix $A$ we have the following two separate columns of equivalent propositions:
  
  \begin{center}
    \begin{tabular}{rlcrl}
      \((1)\) & \(A\) is injective & \quad\quad\quad & \((1)\) & \(A\) is surjective \\ \\
      \((2)\) & \input{figures/AffRel/inj3.tikz} & & \((2)\) & \input{figures/AffRel/surj3.tikz}\\ \\
      \((3)\) & \input{figures/AffRel/inj2.tikz} & & \((3)\) & \input{figures/AffRel/surj2.tikz}\\ \\
      \((4)\) & \input{figures/AffRel/inj1.tikz}& & \((4)\) & \input{figures/AffRel/surj1.tikz}
    \end{tabular}
  \end{center}

\end{proposition}

\section{Graphical linear and affine Lagrangian algebra}
\label{sec:AffLagRel}
In this section, we refine graphical treatments of linear and affine algebra to the symplectic setting. To this end, we first review some of the basic notions in linear symplectic geometry:

\subsection{Symplectic linear and affine algebra}
\label{ssec:symplectic_linear_algebra}

\begin{definition}
  A \emph{symplectic \(\K\)-vector space} is a \(\K\)-vector space \(V\) endowed
  with a bilinear map \(\omega_V: V \oplus V \to \K\) called the  \emph{symplectic form} which is:
  \begin{itemize}
    \item \emph{Alternating:} For all \(\vec v \in V\), \(\omega_V(\vec v,\vec v)=0\).
    \item \emph{Nondegenerate:} Given some \(\vec v \in V\): if for all \(\vec u\in V\) \(\omega_V(\vec v,\vec u)=0\), then \(\vec v=\vec 0\).
  \end{itemize}
  A \emph{symplectic \(\K\)-linear map}
  between symplectic spaces \((V,\omega_V) \to (W, \omega_W)\) is a \(\K\)-linear
  map \(S : V \to W\) that preserves the symplectic form, i.e. \(\omega_W(S\vec v,S\vec u)
  = \omega_V(\vec v,\vec u)\) for all \(\vec v, \vec u\in V\). A
  \emph{symplectomorphism} is a symplectic isomorphism.
  Similarly, a \emph{symplectic \(\K\)-affine map} is a \(\K\)-affine map whose linear component is symplectic; and an affine symplectomorphism is an  symplectic \(\K\)-affine isomorphism.
\end{definition}

The vector space \(\K^0\) has trivial symplectic structure.
The next most simple example of a symplectic vector space is
\(\K^2\) equipped with the symplectic form 

\[
\omega_{1}\!\left(\begin{bmatrix}z_0\\ x_0\end{bmatrix}, \begin{bmatrix}z_1\\ x_1\end{bmatrix}\right)
  \coloneqq  z_0x_1-x_0z_1
\]

Call the first component of \(\K^2\cong \K\oplus\K\) the \(Z\)-grading, and the second the \(X\)-grading.

The vector space \(\K^{2n}\) is also a symplectic vector space with the symplectic form:

\[\omega_n\!\left( \bigoplus_{j=0}^{n-1} \begin{bmatrix}z_j\\x_j\end{bmatrix}, \bigoplus_{j=0}^{n-1} \begin{bmatrix}z_j'\\x_j'\end{bmatrix} \right) \coloneqq  \sum_{j=0}^{n-1} \omega_1\!\left(\begin{bmatrix}z_j\\x_j\end{bmatrix} ,\begin{bmatrix}z_j'\\x_j'\end{bmatrix} \right)\]

\begin{theorem}[Linear Darboux theorem]
  \label{thm:darboux}
  Every finite-dimensional symplectic vector space \((V,\omega_V)\) is symplectomorphic to the symplectic vector space \((\K^{2n} , \omega_n)\).
\end{theorem}

There is an intuitive geometric interpretation of \(\omega_n\): it refines the measure of oriented volumes in \(\K^{2n}\), as \(\bigwedge_{j=1}^n \omega_1\) is a volume form. For \(n=1\), \(\omega_1\) captures the area of a parallelogram given by points \((z_0,x_0),(z_1,x_1)\in\K^2\). This is visualized as follows, where the vertical axis denotes the \(X\)-component, and the horizontal axis denotes the \(Z\)-component:

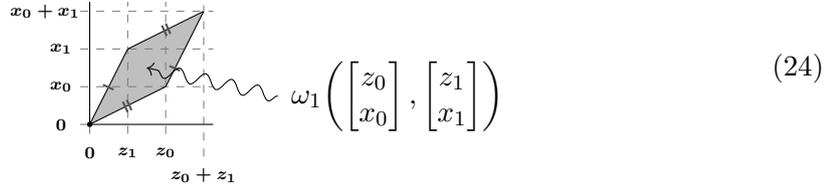
\begin{equation}
  \label{eq:volume}
  \input{figures/AffLagRel/symplectic_volume/volume.tikz}
\end{equation}

Therefore the  linear and affine symplectomorphisms are the linear and affine transformations which preserve oriented volume between the \(Z\) and \(X\)-gradings:

\begin{example}
  \label{ex:volume}
  Using affine symplectomorphisms, the parallelogram from equation~\eqref{eq:volume} can be:
  \begin{align*} 
    &\bullet\ \text{\em Squeezed:} 
      \quad
      \xmapsto
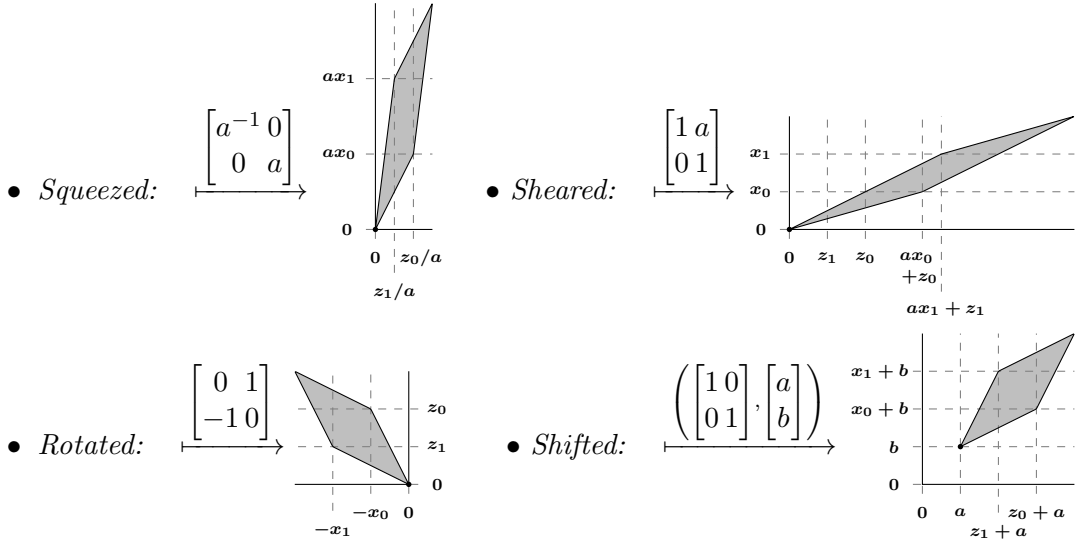
{
        \begin{bmatrix}
          a^{-1} & 0\\
          0    & a
        \end{bmatrix}
      }
      \input{figures/AffLagRel/symplectic_volume/squeezing.tikz}
      \hfil
    &\bullet\ \text{\em Sheared:} 
      \quad
      \xmapsto{
        \begin{bmatrix}
          1 & a\\
          0 & 1
        \end{bmatrix}
      }
      \input{figures/AffLagRel/symplectic_volume/sheared.tikz}
      \\
    &\bullet\ \text{\em Rotated:} 
      \quad
      \xmapsto{
        \begin{bmatrix}
          0  & 1\\
          -1 & 0
        \end{bmatrix}
      }
      \input{figures/AffLagRel/symplectic_volume/rotating.tikz}
      \hfil
    &\bullet\ \text{\em Shifted:} 
      \quad
      \xmapsto{
        \left(
        \begin{bmatrix}
          1 & 0\\
          0 & 1
        \end{bmatrix}
        ,
        \begin{bmatrix}
           a\\
           b
        \end{bmatrix}
        \right)
      }
      \input{figures/AffLagRel/symplectic_volume/shifting.tikz}
  \end{align*}
\end{example}

Notice that for $n=1$ the symplectic form is a volume form and hence the correspondence is perfect, however, for $n\geq 2$, symplectomorphisms are still volume preserving but there are volume preserving maps that are not symplectomorphisms.

Given a linear subspace \(S\) of a symplectic vector space \((V,\omega)\), its \emph{symplectic
complement} is the linear subspace: 
\[S^\omega \coloneqq \left\{\vec v \in V \,\middle|\,  \forall \vec s \in S: \omega(\vec v,\vec s)
= 0\right\}\subseteq V\]

A linear subspace \(S\) of a symplectic vector space \((V,\omega)\) is:
\begin{itemize}
  \item  \emph{isotropic} if \(S \subseteq S^\omega\) (so that for all \(\vec s, \vec t \in S\), \(\omega(\vec{t}, \vec{s})= 0\));
  \item  \emph{coisotropic} if \(S^\omega \subseteq S\);
  \item  \emph{Lagrangian} if it is both isotropic and coisotropic (so that \(S = S^\omega\)).
\end{itemize}

Similarly, an affine subspace \((S,\vec a)\) of  a symplectic vector space \((V,\omega)\) is isotropic / coisotropic / Lagrangian when either \(S\) is isotropic / coisotropic / Lagrangian, or \(S\) is empty.

\begin{example}
  Any line in the affine plane is given by an affine Lagrangian subspace of \((\K^2,\omega_1)\).  These take of one of the two following forms, for some \(a,b\in \K\):
  \[
    \left\{\begin{bmatrix}z\\x\end{bmatrix} \in \K^2\ \middle| \ x=bz+a \right\} 
    \quad\text{visualized as}\quad
    \input{figures/AffLagRel/symplectic_volume/line0.tikz}
  \]
  
  \[
    \left\{ \begin{bmatrix}z\\x\end{bmatrix} \in \K^2\ \middle| \ z=bx+a \right\} 
    \quad\text{visualized as}\quad
    \input{figures/AffLagRel/symplectic_volume/line1.tikz}
  \]

  Notice that the intersection of both classes of lines are the lines which pass through the origin and which are not completely vertical or horizontal.

  The entire plane, however, is the coisotropic subspace 
  \(\K^2 \subseteq \K^2\)
  whereas its orthogonal complement, the origin, is an isotropic  subspace
  \(\K^0 \subset \K^2\).  An arbitrary single point on the plane is an affine coisotropic subspace.

  From a geometric point of view, we see that the points and lines are affine isotropic subspaces of  \(\K^2\) because they have no volume; the line is affine Lagrangian subspace of \(\K^2\) because it has the largest possible dimension without having any volume; and the entire plane is a coisotropic subspace because it is dual to the isotropic subspace given by the origin.
\end{example}

An \emph{isotropic} / \emph{coisotropic} / \emph{Lagrangian}, \emph{linear} / \emph{affine relation} \(m\to n\)
is an isotropic / coisotropic / Lagrangian, linear / affine subspace of:
\[(\K^{2m} \oplus \K^{2n},  (\vec v, \vec w) \mapsto \omega_{m}(\vec v)-\omega_{n}(\vec w))\]

We must introduce this twisted symplectic form, as opposed to \(\omega_{n+m}\), so that these notions of ``relations'' are really stable under relational composition as described in definition~\ref{def:xrel}:

\begin{definition}
  The \(\dagger\)-compact props of isotropic / coisotropic / Lagrangian, linear / affine relations, \(\IR\), \(\CR\),  \(\LR\),  \(\AIR\), \(\ACR\), \(\ALR\) have:
  \begin{itemize}
    \item \textbf{Arrows} are given by isotropic / coisotropic / Lagrangian, linear / affine relations.
    \item \textbf{Composition} is given by relational composition.
  
    \item \textbf{Identity} on \(1\) is given by the diagonal on \(\K^{1}\).
    \item \textbf{The monoidal structure} is given by the direct sum and \(0\). 
    \item \textbf{The symmetry} on \(1\) is given by:
    \begin{equation*}
      \left\{\left( \begin{bmatrix} \vec v \\ \vec w \end{bmatrix}, \begin{bmatrix} \vec w \\ \vec v \end{bmatrix} \right) \,\middle|\,  \vec v, \vec w \in \K^2 \right\}:2\to 2
    \end{equation*}
    \item \textbf{The compact structure} on \(1\) is given by\\
  
    \hspace*{-.5cm} \mbox{\(
    \eta_1 \coloneqq \left\{\left(\bullet, \begin{bmatrix} z \\ x \\ z \\ -x \end{bmatrix}\right) \,\middle|\,  z,x\in \K\right\}: 0\to 2 \quad\quad
    \epsilon_1 \coloneqq \left\{\left(\begin{bmatrix} z \\ x \\ z \\ -x \end{bmatrix}, \bullet\right) \,\middle|\,  z,x\in \K\right\}: 2\to 0
  \)}\\
  
    \item \textbf{The dagger} of an arrow \(R:n\to m\) is given by:
    \begin{equation*}
       R^\dag \coloneqq 
      \left\{
        \left(\bigoplus_{j=0}^{m-1} \begin{bmatrix} z_j \\ -x_j \end{bmatrix}, \bigoplus_{k=0}^{n-1} \begin{bmatrix} z_k' \\ -x_k' \end{bmatrix} \right)
       \,\middle|\, 
        \left(\bigoplus_{k=0}^{n-1}\begin{bmatrix} z_k' \\ x_k'\end{bmatrix}, \bigoplus_{j=0}^{m-1} \begin{bmatrix} z_j \\ x_j \end{bmatrix}  \right)\in R
      \right\} :m\to n
    \end{equation*}
  \end{itemize}
\end{definition}

Our definition differs slightly from the way it is often stated in the literature, for example in the work of \textcite{weinstein_symplectic_2009} and \textcite{network}.  This is because we have decided to work in the skeletal case, where all isomorphic objects are equal. However, both definitions are equivalent, so this difference is inconsequential.

%and there is just as well a \dag-compact closed functor \( \interp{-}_{\abbrar}^{\abbralagr}:\ALR\to\AR\) which negates the \(Z\)-grading appropriately.  Conversely, there is also a \dag-compact closed functor  \( \interp{-}_{\abbralagr}^{\abbrar}:\AR\to \ALR\) which doubles things appropriately.

We must be careful as the relationship between the props \(\lR\) / \(\AR\) and their symplectic counterparts is subtle. First, for example, the symmetric monoidal functor \(\interp{-}_{\abbralagr}^{\abbrar}:\ALR \to \AR \) which forgets the symplectic structure is not a prop morphism, as it sends \(n\mapsto 2n\). Furthermore the compact structure of \(\ALR\) is different from the compact structure we defined on \(\AR\) because the latter is not Lagrangian.

Remark that any (affine) symplectic map \(R: (\K^{2n},\omega_n) \to (\K^{2m},\omega_m)\) induces an (affine) Lagrangian relation via its graph
 \[ \Gamma_R\coloneqq \{(\vec v, R \vec v) \ | \ \forall \vec  v \in \K^{2n} \}:n\to m\]
which is functorial, so these categories of symplectic relations strictly generalize (affine) symplectic maps. 

It is helpful to get a geometric intuition for composition in these symplectic categories of relations:
\begin{example}
  Consider the affine Lagrangian relation \(R:0\to 1\) relating the point to the line \(x=bz+a\); and consider another affine Lagrangian relation \(S:1\to 0\) relating the line \(x=b'z+a'\) to the point.

  \begin{itemize}
    \item
    If \(b=b'\), so that both lines are parallel, but \(a\neq a'\) so that they never meet, then \(S\circ R\) is the empty set.

    \item
    On the other hand if \(b=b'\) and \(a=a'\), so that these are both the same lines then  \(S\circ R\)  is the identity.

    \item
    Even if these two lines intersect only once, then \(S\circ R\) is still the identity!
  \end{itemize}
\end{example}

\subsection{Generators and equations for affine Lagrangian relations}
\label{ssec:symplectic_ZX}

In order to give a graphical axiomatisation of \(\ALR\), we extend the phase group of \(\GAA\) (recalling definition~\ref{def:phasegroup}):

\begin{definition}
  \label{def:GSA}
  Let  \(\ZX\) be the \dag-compact prop given by extending the phase groups of \(\GAA\); the grey spiders phase group extended along the group homorphism \(\K^0\hookrightarrow \K^2\) and the white spider along the group homorphism \(\K^1\hookrightarrow \K^2\):
  \begin{equation}
    \label{eq:symplecticspiders}
    \input{figures/AffLagRel/lhsgrey.tikz} \longmapsto \input{figures/AffLagRel/rhsgrey.tikz}\quad\quad
    \input{figures/AffLagRel/lhswhite.tikz} \longmapsto \input{figures/AffLagRel/rhswhite.tikz}
  \end{equation}
  Modulo the following equations for all \(a,b,c \in\K\) and \(d\in\K^*\):
  \begin{equation}
    \label{eq:ALR_axioms}
    \input{figures/AffLagRel/axioms.tikz}
  \end{equation}
  Where there are derived generators, for all \(a \in\K \):
  \begin{equation}
    \label{eq:box-def}
    \input{figures/AffLagRel/box-def.tikz}
  \end{equation}
  The \dag-structure is given by:
  \begin{equation*}
      \input{figures/AffLagRel/b-spider.tikz} \mapsto \input{figures/AffLagRel/b-spider-daggered.tikz}  \quad
      \input{figures/AffLagRel/w-spider.tikz} \mapsto \input{figures/AffLagRel/w-spider-daggered.tikz}  \quad
      \input{figures/AffLagRel/scalar.tikz} \mapsto \input{figures/AffLagRel/scalar_daggered.tikz} \quad
      \input{figures/AffLagRel/box_undaggered.tikz} \mapsto \input{figures/AffLagRel/box_daggered.tikz} 
    \end{equation*}
  Call the first component of the phase groups for the grey and white  spiders in equation~\eqref{eq:symplecticspiders} the  \emph{affine} phase, and the second component the
  \emph{symplectic} phase.
  
  The derived generators in~\ref{eq:box-def} are called {\em boxes}.  The disinguished boxes with labels \(1\) and \(-1\) are respectively the {\em symplectic Fourier transform} and {\em inverse symplectic Fourier transform}.
\end{definition}

Recalling the geometric intuition from example~\ref{ex:volume}, lifting scalar multiplication along the embedding \(\GAA\hookrightarrow\ZX\) corresponds to squeezing.  Similarly, the affine phase for the white spider corresponds to shifting the \(X\)-gradient.  On the other hand, the affine phase for the grey spider corresponds to shifting the \(Z\)-gradient.  The symplectic phases for the white and grey spiders corresponding to shearing in the \(X\) and \(Z\) direction.  The symplectic Fourier transform corresponds to a \(\pi/2\)-rotation in the affine plane.

Following equation~\eqref{eq:ALR_axioms}, we can actually define multipliers
from boxes:
\[\input{figures/AffLagRel/multiplier_from_box.tikz}\]
From now on,
we will take the boxes as primary, rather than multiplication by scalars.
This is because:
\begin{proposition}
  The presentation  \(\ZX\)  of \(\ALR\) with generators
  \begin{equation*}
    \input{figures/AffLagRel/b-spider.tikz} : m \to n  \quad\quad\quad
    \input{figures/AffLagRel/w-spider.tikz} : m \to n  \quad\quad\quad
    \input{figures/AffLagRel/box.tikz} : 1 \to 1 
  \end{equation*}
  is flexsymmetric.
\end{proposition}
See appendix~\ref{ssec:zx-axioms} for an explicit presentation of \(\ZX\).

Note that the symplectic setting is quite special: as discussed in remark~\ref{rem:notflexsym} \(\AR\) and \(\lR\) do not obviously admit flexsymmetric presentations. 
\begin{definition}
  \label{def:gsa_to_alr}
  There is an interpretation \(\interp{-}_\abbralagr^\zx: \ZX \to \ALR\) given by:
  \begin{align*}
  &\interp{\input{figures/AffLagRel/b-spider.tikz}}_\abbralagr^\zx
    \coloneqq \left\{ \left(
    \bigoplus_{j=0}^{m-1} \begin{bmatrix}
      z_j \\ x
    \end{bmatrix} ,
    \bigoplus_{k=0}^{n-1} \begin{bmatrix}
      z_k' \\ x
    \end{bmatrix}
    \right) \,\middle|\, \begin{aligned} &\vec z\in\K^m, \vec{z'} ; \in \K^n, x \in \K \qq{such that} \\ &\sum_{j=0}^{m-1} z_j - \sum_{k=0}^{n-1} z_k' + bx = a \end{aligned} \right\} \\
  &\interp{\input{figures/AffLagRel/w-spider.tikz}}_\abbralagr^\zx
    \coloneqq
    \left\{ \left(
      \bigoplus_{j=0}^{m-1}
      \begin{bmatrix}
        z \\ x_j
      \end{bmatrix} ,
      \bigoplus_{k=0}^{n-1}
      \begin{bmatrix}
        -z \\ x_k'
      \end{bmatrix}
    \right) \,\middle|\, \begin{aligned} & \vec x \in \K^m, \vec{x'}\in \K^n, z\in \K \qq{such that} \\ &\sum_{j=0}^{m-1} x_j + \sum_{k=0}^{n-1} x_k' - bz = a \end{aligned} \right\}
    \end{align*}
    \begin{align*}
      \interp{\input{figures/AffLagRel/box_zero.tikz}}_\abbralagr^\zx
      \coloneqq \left\{ \left(
        \begin{bmatrix} 0 \\ x \end{bmatrix}, \begin{bmatrix} 0 \\ x' \end{bmatrix}
      \right) \,\middle|\, x,x' \in \K \right\} & &
      \interp{\input{figures/AffLagRel/box_invertible.tikz}}_\abbralagr^\zx
      \coloneqq \left\{ \left(
        \begin{bmatrix}
          z\\x
        \end{bmatrix},
        \begin{bmatrix}
          -cx\\
          z/c
        \end{bmatrix}
      \right) \,\middle|\, z,x\in\K \right\}
  \end{align*}
  where \(m,n \in \N\), \(a,b \in \K\) and \(c \in \K^*\).
\end{definition}

Translating this to string diagrams, following \textcite{comfort_graphical_2021}, we have a factorization of \(\interp{-}_\abbralagr^\zx: \ZX \to \ALR\) through $\GAA$ :

\begin{definition}
  There is an interpretation
  \(\interp{-}_\gaa^\zx : \ZX \to \GAA\), defined on objects
  by \(\interp{m}_\mathsf{GAA}^\zx= 2m\) and on morphisms by:
  \[
    \interp{\input{figures/AffLagRel/b-spider-ind.tikz}}_\gaa^\zx 
    \hspace*{-.25cm}=
    \input{figures/AffLagRel/GAA_b-spider.tikz}
    \quad
    \interp{\input{figures/AffLagRel/w-spider-ind.tikz}}_\gaa^\zx
    \hspace*{-.25cm}=
    \input{figures/AffLagRel/GAA_w-spider.tikz}
  \]
  \[
    \interp{\input{figures/AffLagRel/box_invertible_whydidyouusethesamediagramwithalocalvariableinsideofit.tikz}}_\gaa^\zx =\input{figures/AffLagRel/GAA_box_invertible.tikz}
  \]
\end{definition}

\begin{proposition}
  \label{thm:soundness}
  The interpretations \(\interp{-}_\abbralagr^\zx
  : \ZX \to \ALR\) and \(\interp{-}_{\gaa}^\zx : \ZX \to
  \GAA\)  are \dag-symmetric monoidal functors, making the following diagram of \dag-compact closed categories commute:

  \[
    \xymatrix{%
      \ZX \ar[rr]^{\interp{-}_{\gaa}^\zx} \ar[d]_{\interp{-}_\abbralagr^\zx}
        && \GAA \ar[d]^{\interp{-}_{\abbrar}^{\gaa}}\\
      \ALR \ar[rr]_{\interp{-}_{\abbrar}^{\abbralagr}}
        && \AR
    }
  \]

Where the vertical arrows are morphisms of \dag-compact props.
%  And the vertical arrows are moreover morphisms of \dag-compact props.
\end{proposition}

\subsection{Scalable notation for graphical symplectic algebra}
\label{ssec:scalable_symplectic}
In this subsection we extend the scalable notation for affine relations introduced in
section~\ref{ssec:scalable_gaa} to the setting of affine Lagrangian relations.
First, note that the the matrix arrows in  \(\GAA\) lift along the embedding
 \(\GAA\hookrightarrow\ZX\) given in definition~\ref{def:GSA}.  However, now the  \([n]\)-coloured grey and white spiders both have phase groups \(\K^n \times
\Sym[n]\) (recalling definition~\ref{def:phasegroup}), where  \(\Sym\) denotes the Abelian group of symmetric \(n \times n\) matrices over \(\K\) under addition.

\begin{definition}
  \label{def:scalable_spider}
  Define scalable grey and white spiders in \(\ZX\) by induction on the colour of wires \(k\in\N\).
  The base case is trivial.  For the inductive step, take \(n,m\in \N\), \(a,b\in \K\), \(\vec v, \vec w \in \K^k\) and \(A\in\Sym[k]\), we define  \([k+1]\)-coloured spiders as:
  \begin{equation}
    \label{eq:thick_spider_def}
    \input{figures/AffLagRel/scalable/b-spider-def.tikz}
    \quad\quad\quad\quad
    \input{figures/AffLagRel/scalable/w-spider-def.tikz}
  \end{equation}
  where the scaled box is also inductively defined.  The base case is trivial, and the inductive step is given by:
  \begin{equation}
    \label{eq:scalable_box}
    \input{figures/AffLagRel/scalable/box-def.tikz}
  \end{equation}
\end{definition}

A \(k\)-coloured grey spider with \(n\) inputs and \(m\) outputs parametrizes an undirected open graph with edges coloured by \(\K\), vertices coloured by \(\K^2\), with \(n\) distinguished inputs and \(m\) distinguished outputs.

\begin{example}
  \label{ex:scalable_graph_example}
  On the left hand side we have a graph {\em state} with \(n = 0\), \(m=1\) and \(k=3\).  On the right hand side, we have a spider with  \(m = 3\),  \(n = 2\) and \(k =
  3\):
  \begin{equation}
    \input{figures/AffLagRel/scalable/b-spider-explicit-simple.tikz}
    \quad\quad\quad
    \input{figures/AffLagRel/scalable/b-spider-explicit.tikz}
  \end{equation}
\end{example}

The axioms of \(\ZX\) admit very natural scalable generalizations which we have included in appendix~\ref{ssec:axioms_scalable}.  This notation is very powerful, and makes the proof of completeness much more understandable and elegant.

\subsection{Completeness and normal form}
\label{ssec:completeness}
In this subsection, we show that the interpretation  \(\interp{-}^\zx_\abbralagr:\ZX \cong \ALR\)  is an isomorphism of \dag-compact props.  We have already proven that it is a \dag-compact functor in proposition~\ref{thm:soundness}.  Fullness follows immediately from the result of Comfort and Kissinger  \cite[Lemma 4.6]{comfort_graphical_2021}.
Therefore it remains to show that it is faithful.  To prove this, we establish a normal form.

If a diagram  has empty semantics, then a simple induction argument yields the following unique normal form:

\begin{proposition}
  \label{prop:zero_normal_forms}
  For any \(m,n \in \N\) and \(D \in \ZX(m,n)\) and \(a \in \K^*\): 
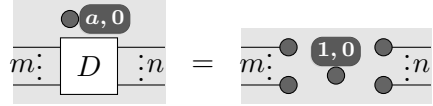
  \begin{equation*}
    \input{figures/AffLagRel/zero_normal_forms.tikz}
  \end{equation*}
\end{proposition}

The normal form for diagrams with nonempty semantics takes more work.
To this end, we further expose the graph theoretical properties of \(\ZX\).
Consider a subclass of diagrams, which is inspired by that of \textcite{duncan_graph-theoretic_2020} in the context of the qubit ZX-calculus:

\begin{definition}
  A \(\ZX\)-diagram is \emph{graph-like} when:
  \begin{itemize}
    \item All spiders are grey spiders.
    \item Grey spiders are only connected via edges containing a box.
    \item There are no self-loops.
    \item Every input or output is connected to a grey spider.
    \item Every grey spider is connected to at most one input or output.
  \end{itemize}

Graph like diagrams induce undirected graph with edges coloured by \(\K\),  vertices coloured by \(\K^2\), with distinguished vertices corresponding to spiders which are connected to specific inputs and outputs.  Given a graph-like diagram:
\begin{itemize}
  \item To each \((a,b)\)-phased grey spider, associate an \((a,b)\)-labelled vertex.
  \item To each wire between two spiders mediated by an \(a\)-labelled box, associate an \(a\)-labelled edge between the corresponding internal vertices.
  \item To each spider connected connected to an input/output, distinguish the corresponding edge as being associated to the specific input/output.  Call such a vertex a {\em boundary vertex}.  Call vertices not associated to inputs/outpus {\em internal}.
\end{itemize}
\end{definition}

\begin{example}
The following \(\ZX\)-diagram is graph-like:
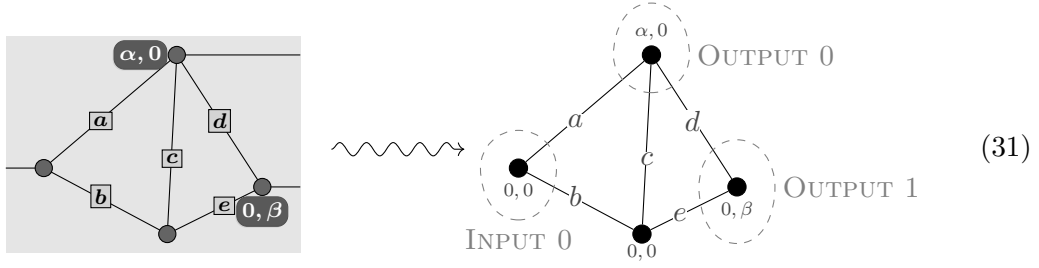
\begin{equation}
  \label{eq:graph-like-example}
  \input{figures/AffLagRel/graph-like-example.tikz}\ \ \
  \begin{tikzpicture}\draw [->, decorate, decoration=snake] (0,0) to (3.5,0);\end{tikzpicture}
  \input{figures/AffLagRel/graph-like-example-graph.tikz}
\end{equation}
where the top three vertices are boundary vertices, and the bottom one is
internal.
\end{example}

\begin{proposition}
  \label{prop:graph_form}
  Every \(\ZX\)-diagram can be rewritten to graph-like form.
\end{proposition}

Because \(\ZX\) is compact closed, we can bend the input wires of diagrams into output
wires. As long as we remember which outputs originally came from bending
inputs, we only need to work with {\em states} \(0
\to n\). Graph-like states take the following simple form in the scalable notation, for some \(m,n \in \N\), \(\vec x \in \K^{m+n}\) and \(X \in \Sym[m+n][\K]\):
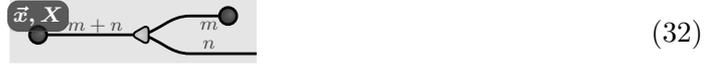
\begin{equation}
  \label{eq:scalable_graph-like}
  \input{figures/AffLagRel/scalable/graph-like.tikz}
\end{equation}
This graph-like state has \(m\) internal vertices, and \(n\) outputs. As
before, \(X\)  encodes the adjacency matrix of the underlying
weighted graph with \(m+n\) vertices (compare with example~\ref{ex:scalable_graph_example} and equations~\eqref{eq:graph-like-example} and \eqref{eq:scalable_graph-like}).  It is
possible to simplify graph-like diagrams by reducing the number of internal
vertices using the following result.

\begin{proposition}
  \label{prop:symplectic_elimination}
  For any \emph{invertible} \(X \in \Sym[m][\K]\), \(Y \in
  \Sym[n][\K]\), \(E \in \Matrices[m][n]\), \(\vec x \in \K^m\), and \(\vec y \in \K^n\):
  \begin{equation*}
    \input{figures/AffLagRel/scalable/symplectic_elimination.tikz}
  \end{equation*}
\end{proposition}

This proposition encompasses both \emph{local complementation} and \emph{pivoting}
of graphs in \textcite{bouchet_circle_1994}, generalised to coloured graphs by \textcite{kante_rank-width_2013}.  This has been used in the context of the stabiliser fragment of the ZX-calculus, for example by 
\textcite{backens_zx-calculus_2014, duncan_graph-theoretic_2020, booth_complete_2022, poor_qupit_2023}:
\begin{proposition}[Local complementation]
  \label{prop:local_complementation}
  For any \(z \in \K^*\), \(a \in \K\), \(\vec{z}
  \in \K^k\), \(E \in \Matrices[k][1]\), \(Z \in \Sym[k]\) where
  \begin{equation*}
    \input{figures/AffLagRel/local_complementation.tikz}
  \end{equation*}
  Where\hfil \(\gamma_i \coloneqq z_i + E_i a z^{\minu 1},\quad \delta_i \coloneqq Z_{i,i} - z^{\minu 1} E_i^2,\quad{\it and}\quad g_{i,j} \coloneqq Z_{i, j} - z^{\minu 1} E_i E_j.\)
\end{proposition}
The new edge weights \(g_{i,j}\) in the right-hand side are obtained from
the edge weights \(E_{i_j}\) of the left-hand side by performing a local
complementation (in the sense of \textcite{kante_rank-width_2013}) around the
vertex \(1\), which is then deleted.

\begin{proposition}[Pivoting]
  \label{prop:pivot}
  For any \(\epsilon \in \K^*\), \(a,b \in \K\), \(\vec{z}
  \in \K^k\), \(E \in \Matrices[n][2]\), \(Z \in \Sym[k]\) where
  \begin{equation*}
    \input{figures/AffLagRel/pivot.tikz}
  \end{equation*}
  Where 
\[\gamma_i \coloneqq z_i + \epsilon^{\minu 1} (a E_{2,i} + b E_{1,i}),\  \delta_i \coloneqq Z_{i,i} - 2 \epsilon^{\minu 1} E_{1,i} E_{2,i},\  {\it and}\  g_{i,j} \coloneqq - \epsilon^{\minu 1} (E_{1,i} E_{2,j} + E_{1,j} E_{2,i}).\]
\end{proposition}
The new edge weights \(g_{i,j}\) in the right-hand side are obtained from
the edge weights \(E_{i_j}\) of the left-hand side by performing a pivot
(in the sense of \textcite{kante_rank-width_2013}) around the edge weighted
by \(\epsilon\). The two vertices connected by this edge are then deleted.

Using proposition~\ref{prop:local_complementation}, it is clear that we can
eliminate any internal spider with phase \((x,z)\) where \(z \neq 0\). Having
done this, we then use proposition~\ref{prop:pivot} to eliminate any
remaining internal spiders that are connected. Therefore, we can refine
 the notion of graph-like diagram:

\begin{definition}
A graph-like diagram is in \emph{Affine with Phases form}
  (AP-form) when:
  \begin{itemize}
    \item There are no inputs;
    \item The internal spiders have phases of the form \((x,0)\) for \(x \in \K\);
    \item Internal spiders are only connected to boundary spiders.
  \end{itemize}
\end{definition}

\begin{proposition}
  \label{prop:ap_form}
  Every graph-like diagram \(0 \to n\) can be put into AP-form.
\end{proposition}

\begin{remark}
In the scalable notation, the AP-form  can be seen to be the following refinement of the graph-like form which was presented in equation~\eqref{eq:scalable_graph-like}:
\begin{equation}
  \label{eq:scalable-AP-form}
  \input{figures/AffLagRel/scalable/AP-form.tikz}
\end{equation}
for \(m,n \in \N\), \(\vec{x} \in \K^m\), \(\vec{y} \in \K^n\),
\(E \in\Matrices[m][n][\K]\), and \(Y \in \Sym\). Here
there are \(m\) internal spiders with phases \((x_k,0)\), connected to \(n\)
boundary spiders via the matrix \(E\).
\end{remark}

The AP-form also appeared in the completeness result of odd-prime-dimensional stabiliser ZX-calculus of  \textcite{poor_qupit_2023}. Our scalable presentation of the AP-formal form can be easily rewritten to more closely resemble their presentation:
\begin{equation}
  \label{eq:scalable-AP-form-derivation}
  \input{figures/AffLagRel/scalable/AP-form-derivation.tikz}
\end{equation}
The history of this form is rather complicated, and we will discuss the attribution and history at the end of this section.

The AP-form, can be row-reduced, which leads to the following refinement:
\begin{definition}
  \label{def:reduced_ap_form}
  A diagram in AP-form defined by \((E,Y,\vec{x}, \vec{y})\) is in
  \emph{reduced AP-form} if its semantics are empty, or it satisfies the following:
  \begin{itemize}
    \item \(E\) is in reduced row echelon form;
    \item \(E\) is surjective (none of the rows of  \(E\) are zero);
    \item  \(\im (E^\trans) \subseteq \ker (Y) \) (if the \(j^\text{th}\)  column of \(E\) is a pivot, then the \(j^\text{th}\)  column of \(Y\) is \(0_n\));
    \item  \(\im (E^\trans) \subseteq \ker (\vec{y}^\trans)\) (if the \(j^\text{th}\)  column of \(E\) is a pivot, then the \(j^\text{th}\)  entry of \(\vec y\) is \(0\)).
  \end{itemize}
\end{definition}

\begin{remark}
  \label{rem:scalable_reduced_AP-form}
  In the scalable notation, the reduced AP-form can be stated graphically, for some \(m \leqslant n \in \N\), \(\vec{x} \in \K^m\), \(\vec{s} \in
  \K^{n-m}\), \(F \in \Matrices[m][n-m]\) and \(S \in \Sym[n-m]\), up to a permutation matrix \(\varsigma\in \Matrices[n][n]\):
  \begin{equation}
    \label{eq:scalable-reduced-AP-form}
    \input{figures/AffLagRel/scalable/reduced-AP-form.tikz}
  \end{equation}
Note that the permutation \(\varsigma\) allows for the pivots \(1_m\) to be separated into a different block than \(F\), so that \(E=\begin{bmatrix} 1_m & F \end{bmatrix}\varsigma\), and similarly for \(\vec s\) and \(0_m\), so that \(y=\varsigma^\trans \begin{bmatrix} 0_m^\trans & \vec s^\trans \end{bmatrix}^\trans \). Therefore, pulling the permutation out to the right makes  the normal form easier to state diagrammatically.
\end{remark}

\begin{proposition}
  \label{prop:reduced_ap_form}
  Every \(\ZX\)-diagram can be put into reduced AP-form.
\end{proposition}

The reduced AP-form is unique:
\begin{proposition}
  \label{prop:ap-unique}
  For any nonempty affine Lagrangian state there is exactly one equivalent diagram in reduced-AP form.
\end{proposition}

Recall that \(\interp{-}^\zx_\abbralagr:\ZX \to \ALR\) is full.  Moreover, it faithful on diagrams with empty semantics by proposition~\ref{prop:zero_normal_forms}.  Because, in addition,  every nonempty \(\ZX\) diagram can be rewritten to the reduced-AP form by proposition~\ref{prop:reduced_ap_form}, which is unique by proposition~\ref{prop:ap-unique}, we have that:

\begin{theorem}
  \label{thm:equivalence}
  The  interpretation \(\interp{-}^\zx_\abbralagr:\ZX \cong \ALR\) is a \dag-compact prop isomorphism.
\end{theorem}

This also gives a complete representation for Lagrangian relations as an immediate corollary:

\begin{corollary}
There is an isomorphism of \dag-compact props \(\interp{-}^{\zx^0}_\abbralagr:\ZX^0 \cong \LR\), where \(\ZX^0\) obtained by restricting the affine phases of the white and grey spiders in \(\ZX\) to be 0 and dropping the axioms involving nonzero affine phase.
\end{corollary}

See appendix~\ref{ssec:zx-axioms-linear} for an explicit presentation of \(\ZX^0\).

\paragraph{A note on the attribution of the (reduced) AP-form}\ \newline
What we have called the ``reduced AP-form'' has a complicated history, originating from within the field of quantum information theory, which we review here in order to give correct attribution. A closely related form, called the ``standard form'' for qubit stabiliser codes seems to first appear in the work of \textcite{Cleve1997}.  When the stabiliser code is taken to have maximal dimension, the standard form for the stabiliser tableau is essentially the reduced AP-form for a linear Lagrangian subspace over the field \(\mathbb{F}_2\). The connection of this form for qubit stabiliser states to graphs was developed by \textcite{Elliott2008}; they decorate the graphs to recover the ``phase-information'' which is lost in the representation of a stabiliser state via its standard form.  \textcite{hu_improved_2022} refined the results of \textcite{Elliott2008}, providing a unique graph-like representation of qubit stabiliser states which also accounted for phase information, in terms of its affine support and a phase-polynomial.  Note that the work of \textcite{Elliott2008} was itself is directly inspired by that of \textcite{van_den_nest_graphical_2004}, where they showed that, modulo scalars, qubit stabiliser states can be represented, non-uniquely by ``quantum graph states'' up to the action of the local Clifford operators. 

In the ZX-calculus community, \textcite{backens_zx-calculus_2014} used Van den Nest et al.'s pseudonormal form of quantum graph states in order to give the first complete presentation of the qubit stabiliser ZX-calculus.
Later on, McElvaney and Backens reformulated  the unique normal form of  \textcite{hu_improved_2022} in the context of the qubit ZX-calculus \cite{mcelvanney_complete_2022}. Because of the connection which they draw to phase-polynomials, they named this incarnation of the ``canonical form'' with phase information the  ``phase-polynomial form'' for a qubit stabiliser ZX-diagram, analagous to our AP-from. They further refine this form into the unique  ``canonical phase-polynomial form'' for qubit stabiliser ZX-diagrams.  This canonical phase-polynomial form is analagous to our reduced AP-form. 

To make the connection between the ``canonical phase-polynomial form'' of  \textcite{mcelvanney_complete_2022} and the AP-form as described in our paper, we must take a detour through odd-prime-dimensional quantum circuits. In the setting of odd-prime-dimensional quantum systems, Bahramgiri and Beigi \cite{bahramgiri_graph_2007, bahramgiri_efficient_2007} had already noted that local complementation of \(\mathbb{F}_p\)-coloured graphs, as defined in the work of \textcite{kante_rank-width_2013}, provides a suitable generalisation of the results of \textcite{van_den_nest_graphical_2004}.  \textcite{booth_complete_2022} used this to prove completeness of a ZX-calculus for odd-prime-dimensional stabiliser quantum systems, using a proof method strongly inspired by \cite{backens_zx-calculus_2014}. On the other hand, inspired by the phase-polynomial form and the unique canonical  phase-polynomial form for qubit stabiliser states of \textcite{mcelvanney_complete_2022}, \textcite{poor_qupit_2023} adapted these forms to  the odd-prime-dimensional qudit stabiliser fragment of the ZX-calculus.  They dubbed the qudit analogue of the canonical phase-polynomial form the ``affine with phases form'' (AP-form), and the analogue of the canonical phase-polynomial form the ``reduced affine with phases form'' (reduced AP-form). By contrast to the quantum graph-state approach, the uniqueness of the reduced AP-form provided a cleaner completeness proof for the odd-prime-dimensional qudit stabiliser ZX-calculus.
 Since \textcite{comfort_graphical_2021} proved that \(\ALR[\Zp]\)  is a projective representation for the odd-prime-dimensional qudit stabiliser ZX-calculus the year before (see section~\ref{sssec:stabiliser}), our generalization of this normal form to \(\ALR\) was a natural next-step. Although, annoyingly. instead of representing qubit stabiliser circuits, \(\ALR[\mathbb{F}_2]\) is isomorphic to the prop of  ``pure circuits in Spekken's toy model'' which was axiomatized separately by \textcite{spekzx}; hence the difficulty in giving correct attribution.

In fact, \textcite{cockett_categories_2022} had already remarked that any affine Lagrangian relation can be represented by a generator matrix in a way which is  analogous to the canonical form for a stabiliser state. When translated into string diagrams this coincides with the reduced AP-form we have given here;
however their statement is purely semantic, and does not concern the completeness of a graphical language.

\section{Graphical linear and affine coisotropic algebra}
\label{sec:AffCoisoRel}
In this section, we extend the presentation \(\ZX\) of \(\ALR\) to one for  \(\ACR\).  We also obtain a presentation of \(\CR\cong\IR\) as an immediate corollary. 
%To this end, we need to understand the structure of \(\ACR\) and \(\CR\).  

First, recall from subsection~\ref{ssec:scalable_gaa} that we can not only take the images of matrices, but we can also take images of affine (resp. linear) relations.  Because affine (resp. linear) Lagrangian relations are structure preserving affine (resp. linear) relations, we can also ask for their images:

\begin{definition}
  Given an affine Lagrangian relation \(L:k\to n\), the {\em image of \(L\)} is the image of the underlying affine relation \(2k\to 2n\):

  \[\input{figures/AffCoisoRel/image.tikz}\]

  And likewise, for the image of linear Lagrangian relations.
\end{definition}

Recall the following result of Comfort \cite[Theorem 4.2]{comfort_algebra_2023}:

\begin{lemma}
  \label{lem:coisotropicimage}
  The set of nonempty affine (resp. linear) coisotropic subspaces of a symplectic vector space \((\K^{2n},\omega_n)\) of dimension \(n+k\) are in bijection with the set of the  images of isometries of type \(k\to n\) in \(\ALR\) (resp. \(\LR\)).
\end{lemma}

%Given an arrow \(L:k\to n\) in \(\ZX\), recall from subsection~ that the image is given by the following diagram in \(\GAA\):

%Therefore, we immediately have that, to obtain all affine coisotropic relations it suffices to add a single affine coisotropic relation which relates the point to the affine plane:
%
%%\begin{corollary}
%%Adding the following generator to  \(\ALR\) generates \(\ACR\) under the monoidal product and composition:
%%
%%\[\left\{\left(\bullet,\begin{bmatrix}z\\x\end{bmatrix}\right) \, \middle| \, z,x\in \K\right\}:1\to 0\]
%%
%%Similarly,  adding this relation to  \(\LR\) and generates \(\CR\).
%%\end{corollary}

Therefore, we only need to add a single generator of type \(0\to 1\) to  \(\ZX\)  (resp. \(\ZX^0\)) to get a presentation for \(\ACR\) (resp. \(\CR\)).
Geometrically, this corresponds to the linear relation which relates the point to affine plane.  In fact, it will be slightly more convenient, but equivalent, to instead add its converse: the relation which discards the entire plane.

We  now must find which additional equations this generator satisfies.  To this end, we investigate more of the the structure of \(\ALR\) and \(\LR\).
First, we observe that there is an alternative presentation of affine/linear coisotropic subspaces due to \textcite[Theorem 4.5]{comfort_algebra_2023}:

\begin{lemma}
\label{lem:sameimagesameidempotent}
Given  \(L:k\to n\) and \(K:m\to n\) in \(\ALR\) (or in \(\LR\)), then:

\[ \im(L)= \im(K) \iff L L^\dag     = K K^\dag  \]
\end{lemma}

Second, observe that linear and affine Lagrangian relations have the following bipartite decomposition:

\begin{lemma}
\label{lem:decomposition}
  Any affine (resp. linear) Lagrangian relation \(L:n\to m\) can be decomposed in the form:
  \[\input{figures/AffCoisoRel/stdform.tikz}\]
  where \(S:n\to n \) and \(T: m \to m\) are affine (resp. linear) symplectomorphisms.
\end{lemma}

Note that affine symplectomorphisms are precisely the isomorphisms which are precisely the unitaries in \(\ALR\), and similarly, for the linear symplectomorphisms in \(\LR\).
It seems that this decomposition is well-known, where it is related to ``coisotropic reduction'' as in \textcite[section~3.1]{weinstein_symplectic_2009}.
However, we arrived at this decomposition from a slightly different angle than it usually stated in the literature. 
For this reason, it is evident that this decomposition is an instance of the ``generalized compact singular value decomposition'' in the context of dagger categories of \textcite[definition~3.5]{MPI}.  So in some sense, this a symplectic version of the singular value decomposition of complex matrices.  Because of this similarity to the singular value decomposition, we have a symplectic version of the ``essential uniqueness of purification'' of mixed quantum circuits:

\begin{lemma}\label{lem:enough_iso}
	Given (affine) Lagrangian relations \(L:n \to a\) and \(K:n \to b\) such that  \(L^\dagger  L = K^\dagger  K\), and \(b\leq a\) then there exist an (affine) symplectomorphism \(H: b+ (a-b)\to a\), such that:
	\[\input{figures/AffCoisoRel/enoughiso.tikz}\]
\end{lemma}

Combining lemmas~\ref{lem:sameimagesameidempotent} and~\ref{lem:enough_iso}, given two (affine) Lagrangian relations \(L\) and \(K\) with \(\im(L^\dag)=\im(K^\dag)\), then it follows that \(L\) and \(K\) differ only up to postcomposition by an isometry!
Recall the following definition from Carette et al. which quotients precisely by this congurence \cite[definition~7]{discard}:

\begin{definition}
Given a \dag-compact prop, \(\mathsf C\), the {\em discard construction} of \(\mathsf{C}^\disc\) is presented by freely adding a generator \(\input{figures/AffCoisoRel/discard.tikz}: 1\to 0\) to \(\mathsf C\) modulo the equations such that for all isometries  \(H:n\to m\) in \(\mathsf C\):

\[\input{figures/AffCoisoRel/discardconstruction.tikz}\]

\end{definition}

In other words, in the language of  \textcite[definition~10]{discard}, it follows immediately from lemma~\ref{lem:enough_iso} that \(\ALR\) and \(\LR\) both have ``enough isometries,'' and thus:

\begin{theorem}
\(\ZX^\disc\cong \ACR\) and \((\ZX^0)^\disc\cong\CR\).
\end{theorem}

Luckily, the generators for (affine) symplectomorphisms, and by immediate consequence, the  isometries are well-known; thus:

\begin{corollary}
\label{cor:disc}
\(\ZX^\disc\cong\ACR\) is presented by the adding the generator \(\input{figures/AffCoisoRel/discard.tikz}: 1\to 0\)  to \(\ZX\) in addition to the equations discarding the generating isometries for all \(a,b\in\K\):

\[\input{figures/AffCoisoRel/eqdisc.tikz}\]

where \((\ZX^0)^\disc\cong\CR\) is presented by adding the same generator to  \(\ZX^0\)  and restricting the grey and white phase groups to \(\K\) along the embeddings \(b \mapsto (0,b)\).
\end{corollary}

\section{Applications}
\label{sec:Applications}
In mechanics, the state of a system is often given by the configurations of positions, and the evolution of the system is described by how the positions change over time.  However, it is often useful to take a different perspective where states are determined by their abstract positions and momenta, ie. where they are elements of the {\em phase space}.  From this perspective, the evolution of the system is described by how the positions and momenta change over time; in other words, the evolution is determined by the {\em Hamiltonian flow} of particles through the system.  

\subsection{Classical mechanics}
\label{subsec:classical}
In classical mechanics, the only valid reversible transformations are those which preserve the oriented local volume between position and momentum: ie the (affine) symplectormorphisms.
In other words, symplectic geometry is the geometry of the phase space picture of classical mechanics.
Recall from theorem~\ref{thm:darboux} that a symplectic vector space can be decomposed into two parts: a \(Z\) and \(X\) grading. The elements of these gradings are interpreted respectively as abstract  configurations of {\em positions} and {\em momenta}. There are different notions of positions and momenta  in different classical mechanical settings (adapted from \textcite[page~23, table~2.1]{bondgraph} and \textcite{thetable}):\\

\hfil
\begin{tabular}{||l|llll||}
\hline
{\it Classical}     		      & 	  &			       &                                         &\\
{\it mechanics}			     & Z        & \({d Z}/{d t}\)   & \(X\)                                  & \({d X}/{d t}\)\\
 \hline
{\bf Translation}                     & position & velocity         & momentum                      & force\\
{\bf Electronic}                      & charge   & current          & flux linkage                      & voltage\\
{\bf Hydraulic}                       & volume  & flow               &  pressure momentum      & pressure\\
{\bf Thermal}                         & entropy & entropy flow  & temperature momentum & temperature\\
\hline
\end{tabular}\\~

The connection between affine Lagrangian relations and electrical circuits is well documented, therefore we focus on this example. Baez et al. use  linear Lagrangian relations to give a semantics for passive linear circuits: an idealized class of electrical circuits with linear behaviour  \cite{passive}.  We will review this work from a string diagrammatic perspective.
To this end, recall the physical laws respectively due to Ohm \cite{ohm} and  Kirchhoff \cite{Kirchhoff1845}:
\begin{description}
\item[Ohm's law:]  The voltage around the node in a circuit is equal to the current multiplied by the resistance.
\item[Kirchhoff's current law:] The sum of currents flowing into a node is equal to the sum of currents flowing out of the node.
\end{description}

Given a resistor with resistance \(r\in \R^{+}\), by Kirchhoff's current  law the incoming current is equal to the outgoing current; and by By Ohm's law, the outgoing potential is equal to the current multiplied by \(r\) plus the incoming potential.  Therefore we can interpret electrical circuits composed of resistors in \(\zx_\R\):
\[
\interp{
\begin{circuitikz}
	\node [style=none] (3) at (0, 1) {};
	\node [style=none] (4) at (0, -1) {};
	\node [style=none] (5) at (2, -1) {};
	\node [style=none] (6) at (2, 1) {};
	\draw [style=background] (4.center)
		 to (3.center)
		 to (6.center)
		 to (5.center)
		 to cycle;
	\draw (0,0) to [resistor, R=$r$] (2,0);
\end{circuitikz}
}_{\zx_\R}^\ECirc
=
\input{figures/Applications/electrical/resistor_interpretation.tikz}
\]

White spider fusion entails that composing resistors in sequence adds resistance:
\begin{equation}
  \input{figures/Applications/electrical/serie.tikz}
\end{equation}

Consider an ideal junction of \(n\) incoming wires connected to \(m\) outgoing wires with no resitance.  By Kirchhoff's current law, the sum of the currents of the incoming wires is equal to the sum of the currents of the outgoing wires.  Moreover, by Ohm's law, because resistance is zero, the   incoming and outgoing potentials are all made to be equal.  Therefore an ideal junction of \(n\) incoming wires and \(m\) outgoing wires is interpreted as the following Lagrangian relation:
\[
\interp{
\input{figures/Applications/electrical/junction_symbol.tikz}
}_{\zx_\R}^\ECirc
=
\input{figures/Applications/electrical/junction_interpretation.tikz}
\]

A circuit of resistors with resistances \(r_0\) and \(r_1\) composed in parallel can therefore be simplified using pivoting:\\

\begin{align}
\input{figures/Applications/electrical/resistors_in_parallel/0.tikz}
&\stackeqmid{\Fusion}
  \input{figures/Applications/electrical/resistors_in_parallel/0-1.tikz}
\stackeqmid{\minilemref{lem:antipode_spider}}
  \input{figures/Applications/electrical/resistors_in_parallel/1.tikz}
\stackeqmid{\Bigebra}
  \input{figures/Applications/electrical/resistors_in_parallel/2.tikz}\\
\stackeq{\minilemref{lem:symplectic_states}}
  \input{figures/Applications/electrical/resistors_in_parallel/3.tikz}
\stackeqmid{\Fusion}
  \input{figures/Applications/electrical/resistors_in_parallel/4.tikz}
\stackeqmid{\minilemref{lem:symplectic_states}}
  \input{figures/Applications/electrical/resistors_in_parallel/5.tikz}
\stackeqmid{\Fusion}
  \input{figures/Applications/electrical/resistors_in_parallel/6.tikz}
\end{align}

This fragment of eletrical circuits composed of circuits of resistors with junctions of wires has recently been given a complete equational theory by  \textcite{amolakresistors}; their completeness result relies on the ``star-mesh'' transformations, which are linear versions of the pivoting rule which we discussed in proposition~\ref{prop:pivot}.
However, we are working in a larger semantic domain,  therefore more generators can be expressed using our string diagrams.
For example, by working in affine Lagrangian relations, now we can nudge the position of momentum by a constant factor. Using results from  the subsequent work of Baez et al., voltage and current sources with voltage \(v \in \R^{+}\) and current \(\iota \in \R^{+}\) have the following interpretations \cite{network}:
\[
\interp{
\begin{circuitikz}[/tikz/circuitikz/bipoles/length=1cm]
	\node [style=none] (3) at (0, 1) {};
	\node [style=none] (4) at (0, -1) {};
	\node [style=none] (5) at (2, -1) {};
	\node [style=none] (6) at (2, 1) {};
	\draw [style=background] (4.center)
		 to (3.center)
		 to (6.center)
		 to (5.center)
		 to cycle;
	\draw (2,0) to  (0,0);
	\node [style=vsourceAMshape, fill=white, xscale=-1] at (1,0) {};
	\node [style=none] (0) at (1, 1) {$v$};
\end{circuitikz}
}_{\zx_\R}^\ECirc
=
\input{figures/Applications/electrical/voltage_source_interpretation.tikz}
\qquad\qquad
\interp{
\begin{circuitikz}
	\node [style=none] (3) at (0, 1) {};
	\node [style=none] (4) at (0, -1) {};
	\node [style=none] (5) at (2, -1) {};
	\node [style=none] (6) at (2, 1) {};
	\draw [style=background] (4.center)
		 to (3.center)
		 to (6.center)
		 to (5.center)
		 to cycle;
	\draw (2,0) to  (0,0);
	\node [style=isourceAMshape, fill=white, xscale=-1] at (1,0) {};
	\node [style=none] (0) at (1, 1) {$\iota$};
\end{circuitikz}
}_{\zx_\R}^\ECirc
=
\input{figures/Applications/electrical/current_source_interpretation.tikz}
\]

We obtain an elegant proof that composing voltage sources in parallel with different voltages \(v_0\neq v_1\) leads to contradiction:

\begin{align}
\input{figures/Applications/electrical/voltage_sources_in_parallel/0.tikz}
&\stackeqmid{\Fusion}
  \input{figures/Applications/electrical/voltage_sources_in_parallel/0-1.tikz}
\stackeqmid{\minilemref{lem:antipode_spider}}
  \input{figures/Applications/electrical/voltage_sources_in_parallel/1.tikz}
\stackeqmid{\Bigebra}
  \input{figures/Applications/electrical/voltage_sources_in_parallel/2.tikz}\\
\stackeq{\Copy}
  \input{figures/Applications/electrical/voltage_sources_in_parallel/3.tikz}
\stackeqmid{\Fusion}
  \input{figures/Applications/electrical/voltage_sources_in_parallel/4.tikz}
\stackeqmid{\minipropref{prop:zero_normal_forms}}
  \input{figures/Applications/electrical/voltage_sources_in_parallel/6.tikz}
\end{align}

Our graphical simplifications of the parallel composition of resistors and voltage sources should be compared with the ``doubled'' proofs of Bonchi et al., where they represent the voltage and current gradings as separate wires  in \(\GAA[\R]\) \cite{gaa}.

Baez et al. showed how discretized time-dependent electical circuits can be interepreted in affine Lagrangian relations over the real rational functions \(\R(s)\) \cite{network}.  In this setting, the \(j^\text{th}\) element of the stream is encoded as the coefficient of the \(j^\text{th}\) power of the polynomial indeterminate \(s\).  Taking the derivative and integrating with respect to time are therefore represented by multiplication and division by the polynomial indeterminate \(s\).
Applying  Kirchhoff's current law and Ohm's law,  (linear) inductors with inductance  \(\ell \in \R^+\)  and  (linear) capacitors with capacitance \(c \in \R^+\)  can therefore be interpreted as follows:
\[
\interp{
\begin{circuitikz}
	\node [style=none] (3) at (0, 1) {};
	\node [style=none] (4) at (0, -1) {};
	\node [style=none] (5) at (2, -1) {};
	\node [style=none] (6) at (2, 1) {};
	\draw [style=background] (4.center)
		 to (3.center)
		 to (6.center)
		 to (5.center)
		 to cycle;
	\draw (0,0) to [inductor, L=$\ell$] (2,0);
\end{circuitikz}
}_{\zx_{\R(s)}}^\ECirc
=
\input{figures/Applications/electrical/inductor_interpretation.tikz}
\qquad
\interp{
\begin{circuitikz}
	\node [style=none] (3) at (0, 1) {};
	\node [style=none] (4) at (0, -1) {};
	\node [style=none] (5) at (2, -1) {};
	\node [style=none] (6) at (2, 1) {};
	\draw [style=background] (4.center)
		 to (3.center)
		 to (6.center)
		 to (5.center)
		 to cycle;
	\draw (0,0) to [capacitor, C=$c$] (2,0);
\end{circuitikz}
}_{\zx_{\R(s)}}^\ECirc
=
\input{figures/Applications/electrical/capacitor_interpretation.tikz}
\]

Unlike the previous approaches, the scalable notation for \(\ZX\) allows us to compactly represent electrical circuits.  For example, recall that an impedance matrix is an \(n\times n\) real matrix which describes the impedance between ports in an electrical network (see the textbook of \textcite[section~4.2]{Pozar} for reference).  An electrical network is said to be reciprocal, when the relation of impedance between nodes is symmetric: that is to say, when the  corresponding impedance matrix is an element of \(\Sym[n][\R]\).  In the scalable notation for \(\ZX[\R]\), an impedance matrix  is represented by a thickened, scalable version of the resistor, which carries a symmetric matrix \(R\in \Sym[n][\R]\) as a phase, rather than a scalar \(r\in \R\):

\begin{equation}
  \input{figures/Applications/electrical/impendance_interpretation.tikz}
\end{equation}

This raises the question if the constructions from Boisseau et al.'s  impedance calculus \cite{Boisseau2022} such as ``impedance boxes'' can be represented succinctly in the scalable notation for \(\ZX\).

It is  important to reiterate that the electrical circuits semantics is only one example.  We could have just as easily chosen any row from the table as our semantics.  For example, in fluid mechanics we would replace idealized wires with idealized pipes through which fluid flows, resistors with narrowing pipes and so on.  Or more conceptually, some linear/affine network through which tokens are flowing throughout.  For example, consider an idealized network of cars driving along roads, where the flow of cars is assumed to be linear/affine.  Introducing a bottleneck to the flow of traffic, such as reducing the number of lanes, would be analagous to a resistor.

\subsection{Quantum mechanics}
Classical mechanics can be quantized in two different, but connected ways.
\subsubsection{Stabiliser quantum mechanics in odd-prime dimensions}
\label{sssec:stabiliser}
Given an odd prime \(p\),  Comfort and Kissinger showed that the \dag-compact prop of affine Lagrangian relations over \(\Zp\) is isomorphic to the \dag-compact prop \(\Stab\) of odd-prime-dimensional pure qudit stabiliser circuits, modulo scalars \cite{comfort_graphical_2021}.  If one doesn't allow symplectic-phases, then \(\ZX[\mathbb{F}_2]\) restricts to a semantics for qubit ``CSS-circuits.''  
Here, the \(Z\) and \(X\) gradings correspond respectively to the Pauli \(Z\) and Pauli \(X\) observables, and the symplectic form captures the commutation of Pauli operators. Therefore, we can add another row to our table:\\

\hfil
\begin{tabular}{||l|llll||}
\hline
 {\it Quantum mechanics}                     & $Z$          & ${d Z}/{d t}$               & $X$                              & ${d X}/{d t}$\\ \hline
{\bf Stabiliser circuits}                             & Pauli $Z$ & Pauli $Z$ flow         & Pauli $X$                      & Pauli $X$ flow\\
\hline
\end{tabular}\\~

Denote the induced prop isomorphism by \(\interp{-}_{\zx_\Zp}^{\Stab}:\Stab\to\zx_\Zp\).
Following the convention of Gottesman \cite{gottesman_fault-tolerant_1999}, the  \emph{controlled-\(\mathcal X\)}, \emph{Fourier},
 \emph{phase}, and \emph{Pauli \(\mathcal X\)} gates correspond to the generators of the prop for the affine symplectic groupoid, or ``general affine symplectic group'':
\begin{align}
  \interp{\input{figures/Applications/stabiliser/cnot_circuit.tikz}}_{\zx_\Zp}^{\Stab} = \input{figures/Applications/stabiliser/cnot_ZX.tikz}
    \qquad\qquad
  &\interp{\input{figures/Applications/stabiliser/hadamard_circuit.tikz}}_{\zx_\Zp}^{\Stab} = \input{figures/Applications/stabiliser/hadamard_ZX.tikz}\\
  \interp{\input{figures/Applications/stabiliser/phase_circuit.tikz}}_{\zx_\Zp}^{\Stab} = \input{figures/Applications/stabiliser/phase_ZX.tikz}
    \qquad\qquad
  &\interp{\input{figures/Applications/stabiliser/X_circuit.tikz}}_{\zx_\Zp}^{\Stab} = \input{figures/Applications/stabiliser/X_ZX.tikz}
\end{align}
Moreover, the ground state and effect have the following interpretations:
\begin{equation}
\interp{\input{figures/Applications/stabiliser/ketzero_circuit.tikz}}_{\zx_\Zp}^{\Stab} = \input{figures/Applications/stabiliser/ketzero_ZX.tikz}
    \qquad\qquad
\interp{\input{figures/Applications/stabiliser/brazero_circuit.tikz}}_{\zx_\Zp}^{\Stab} = \input{figures/Applications/stabiliser/brazero_ZX.tikz}
\end{equation}

The dagger of \(\ZX[\Zp]\) corresponds to the Hermitian adjoint in \(\Stab\).
Therefore, we find that \(\ZX[\Zp]\) is essentially the same presentation given by Po\'or et al. for odd-prime-dimensional pure qudit stabiliser circuits, where we instead quotient by scalars \cite{poor_qupit_2023}.  Note that our white spiders correspond to the red spiders of Po\'or et al. and our grey spiders correspond to their green spiders.

However, the symplectic perspective on stabiliser circuits which we take here gives a unique notational advantage over the traditional Hilbert spaces semantics which is used by \textcite{poor_qupit_2023} .
 Given
the adjacency matrix \(G\) of an \(\Zp\)-edge-weighted
graph on \(n\) vertices, the associated \emph{quantum graph state}
\cite{zhou_quantum_2003, hall_cluster_2007, brell_generalized_2015} is
precisely
 \(\interp{\input{figures/Applications/stabiliser/graph_state.tikz}}\!{\substack{\Stab\\ \zx_\Zp}}\)
(up to a normalisation scalar). Therefore, it is very natural to work with graph states in \(\ZX\).  In our presentation one can simultanously exploit  the geometric nature of string diagrams as well as the algebraic representation of  adjacency matrices: here both the string diagrams are graphs, and graphs are string diagrams! Such a natural and integrated string diagrammatic representation of graph states is very important because of their widespread usage in the quantum computing literature \cite{keet_quantum_2010, marin_access_2013,
proctor_ancilla-driven_2017, reimer_high-dimensional_2019, joo_logical_2019,
booth_outcome_2023}.

Comfort showed that, for odd prime \(p\),  the prop of affine coisotropic relations over  \(\Zp\) is isomorphic to odd-prime-dimensional qudit {\em mixed} stabiliser circuits modulo scalars, \(\stab_p^\disc\) \cite{comfort_algebra_2023}.  Here, the discard relation which we added to \(\ZX\) to obtain \(\ZXdisc\) corresponds to the quantum discarding map:
\begin{equation}
\interp{\input{figures/Applications/stabiliser/discard.tikz}}_{\zx_\Zp^\disc}^{\stab_p^\disc} = \input{figures/Applications/stabiliser/discard.tikz}
\end{equation}

The states in mixed stabiliser circuits correspond to {\em stabiliser codes}.  Stabiliser codes are used extensively in quantum error correction, as witnessed by the success of Gottesman's PhD thesis on the subject \cite{gottesman}.  Stabiliser codes are essentially stabiliser states which have been perturbed by measurement, which gives them extra redundancy to protect information against errors.
% Consider, for example, the standard encoder for the Steane code represented in \(\ZX[\mathbb{F}_2]\):
%\begin{equation}
%\tikzfig{figures/Applications/stabiliser/steane_encoder_i}
%\end{equation}
%The code itself is given by composing such an ecoder with the maximally mixed state:
%\begin{equation}
%\tikzfig{figures/Applications/stabiliser/steane_encoder_ii}
%\end{equation}
%Therefore, we can derive other encoders using the equational theory of  \(\ZX[\mathbb{F}_2]^\disc\).

Combining our work with that of \textcite{comfort_algebra_2023},  one obtains a presentation of the 2-coloured prop of stabiliser circuits with both quantum and classical types mediated by preparation and measurement maps. In fact, this has nothing to do with prime fields, and can be constructed by the coproduct of \(\ZX\) and \(\GAA\) along an isometry witnessing the embedding of \dag-compact props \(\GAA\hookrightarrow\ZX\) exhibited in definition~\ref{def:GSA}.  In the stabiliser circuit semantics, this isometry corresponds to the ``state preparation map,'' and its adjoint to the ``measurement map.'' But more generally, this provides a syntax for mechanical systems with affine control.

\subsubsection{Picturing quantum optics}

The \(\LOv\)-calculus is a complete diagrammatic language for reasoning about
\emph{passive linear} quantum optics \cite{clement_lov-calculus_2022}. It is
a prop defined on the following set of generators, for all \(\theta, \varphi\in\R\):  
\begin{itemize}
  \item \textbf{phase shifter},  \textbf{wave plate}, \textbf{beamsplitter}, and \textbf{polarising beamsplitter}:
  \begin{equation*}
    \input{figures/Applications/lov/convtp-phase-shift.tikz}:1 \to 1
    \qquad \input{figures/Applications/lov/pol-rot.tikz}:1 \to 1
    \qquad \input{figures/Applications/lov/beamsplitter.tikz}:2 \to 2
    \qquad \input{figures/Applications/lov/bs.tikz}:2 \to 2
  \end{equation*}
  \item \textbf{vacuum state} and \textbf{effect}:
  \begin{equation*}
    \input{figures/Applications/lov/gene-0.tikz}:0 \to 1 \qquad \input{figures/Applications/lov/detector-0.tikz}:1 \to 0
  \end{equation*}
\end{itemize}

We  can interpret a subset of the generators of
\(\LOv\) in \(\GAA[\R]\cong\AR[\R]\).  Let \(\LOv'\) denote the fragment \(\LOv\) generated only by the phase shifters, wave plates, beamsplitters and polarizing beamsplitters. 
Identifing the generators
of \(\LOv'\) with elements of the affine orthogonal group
%\cite{fabre_modes_2020, ferraro_gaussian_2005}. \cole{I don't see how this follows at all}
we obtain a symmetric monoidal functor \(\interp{-}_{\GAA[\R]}^{\LOv'}:{\LOv'}\to\GAA[\R] \).
On objects it is given by doubling \(\interp{n}_{\GAA[\R]}^{\LOv'}=2n\), and on the
generators by:

\[
\begin{array}{rlcrl}
  \label{eq:optics}
  \interp{\input{figures/Applications/lov/pol-rot.tikz}}_{\GAA[\R]}^{\LOv'}
    =&\ \input{figures/Applications/lov/rotation_ZX.tikz} &\quad&
    \interp{\input{figures/Applications/lov/convtp-phase-shift.tikz}}_{\GAA[\R]}^{\LOv'}
     =&\ \input{figures/Applications/lov/phase_ZX.tikz} \\ \\
    \interp{\input{figures/Applications/lov/bs.tikz}}_{\GAA[\R]}^{\LOv'}
     =&\ \input{figures/Applications/lov/polarising_bs.tikz}&\quad&
    \interp{\input{figures/Applications/lov/beamsplitter.tikz}}_{\GAA[\R]}^{\LOv'}
     =&\ \input{figures/Applications/lov/beamsplitter_ZX.tikz}
\end{array}
\]

Here, there are some terminological clashes.
Passive optical components correspond to \emph{affine} relations whereas the passive
electric circuits discussed in subsection~\ref{subsec:classical} are \emph{linear} relations. Furthermore, the phase-shifter in this context of quantum optics should therefore not be confused with the phase-shift gate in the context of stabiliser circuits which we discussed in subsubsection~\ref{sssec:stabiliser}.

By working in the image of the embedding \(\interp{n}_{\GAA[\R]}^{\LOv'}\) we already have the ability to model 
additional quantum-optical components that are not present in
\(\mathsf{LOv}\). For instance, active squeezing of the vertical polarisation
mode is given, for \(\eta \in \R\), by:
\begin{equation}
  \input{figures/Applications/lov/squeezing.tikz}
\end{equation}
Furthermore \(\GAA[\R]\) is compact closed, as opposed to the usual infinite dimensional Hilbert space semantics,  so by contrast we have much more freedom for rewriting in this setting.

Extending along the embedding \(\GAA[\R] \to\ZX[\R]\),
we can now also describe more components of Gaussian quantum optics,
since these correspond to elements of the affine \emph{symplectic} group
\cite{weedbrook_gaussian_2012}. %  Therefore, we can add one more row to our table:
%
%\begin{tabular}{||l|llll||}
%\hline
% {\it Quantum mechanics}                     & $Z$          & ${d Z}/{d t}$               & $X$                              & ${d X}/{d t}$\\ \hline
%{\bf Quantum optics}                             & wavelength & frequency         & amplitude                      & amplitude difference \\
%\hline
%\end{tabular}\\~
In this setting, shearing between the position and momentum
operators takes the form:
\begin{equation}
  \input{figures/Applications/lov/shear.tikz}
\end{equation}

However, one can not interpret the vacuum state in
\(\ALR[\R]\), because this only admits ``infinitely-squeezed'' states. In our forthcoming paper
\cite{booth_complete_2023}, we will also show that there is a symmetric monoidal embedding \(\interp{-}^\LOv_{\ZX[\C]} : \LOv
\to \ZX[\C]\). This interpretation extends the one above and is given
on the vacuum state and effect by:
\begin{align}
  \interp{\input{figures/Applications/lov/gene-0.tikz}}^\LOv_{\ZX[\C]}
  = \input{figures/Applications/lov/vacuum_state.tikz}
  \qquad\qquad
 \interp{\input{figures/Applications/lov/detector-0.tikz}}^\LOv_{\ZX[\C]}
 = \input{figures/Applications/lov/vacuum_effect.tikz}
\end{align}
lifting the diagrammatic characterisation of quopit stabiliser quantum
mechanics given in \cite{poor_qupit_2023} to the setting of Gaussian quantum
mechanics \cite{weedbrook_gaussian_2012}.

\section*{Acknowledgement}
CC is supported by the project NEASQC funded from the
European Union’s Horizon 2020 research and innovation programme (grant
agreement No 951821).

\addcontentsline{toc}{section}{References}
{\raggedright\printbibliography}

\appendix

\section{Axiom tables}
\label{app:axioms}
%
%\subsection{Axioms for \texorpdfstring{\(\GLA\)}{GLA}}
%\robert{...}
%
%\subsection{Axioms for \texorpdfstring{\(\GAA\)}{GAA}}
%\robert{...}
\subsection{Axioms for \texorpdfstring{\(\ZX\)}{GSA}}
\label{ssec:zx-axioms}

\begin{definition}
  \label{def:zx}
  Let \(\ZX\) denote the \dag-compact prop with:
  \begin{itemize}
  \item \textbf{Objects} natural numbers.
  \item \textbf{Morphisms} are presented by the following generators indexed by \(n,m \in \N\) and \(a,b\in \K\): 
  \begin{equation*}
    \input{figures/AffLagRel/b-spider.tikz} : m \to n  \quad\quad\quad
    \input{figures/AffLagRel/w-spider.tikz} : m \to n  \quad\quad\quad
    \input{figures/AffLagRel/box.tikz} : 1 \to 1 
  \end{equation*}
  Modulo the equations imposing flexsymmetry, in addition to the following equations for all \(a,b,c,d \in \mathbb{K}\) and \(z
  \in \mathbb{K}^*\):
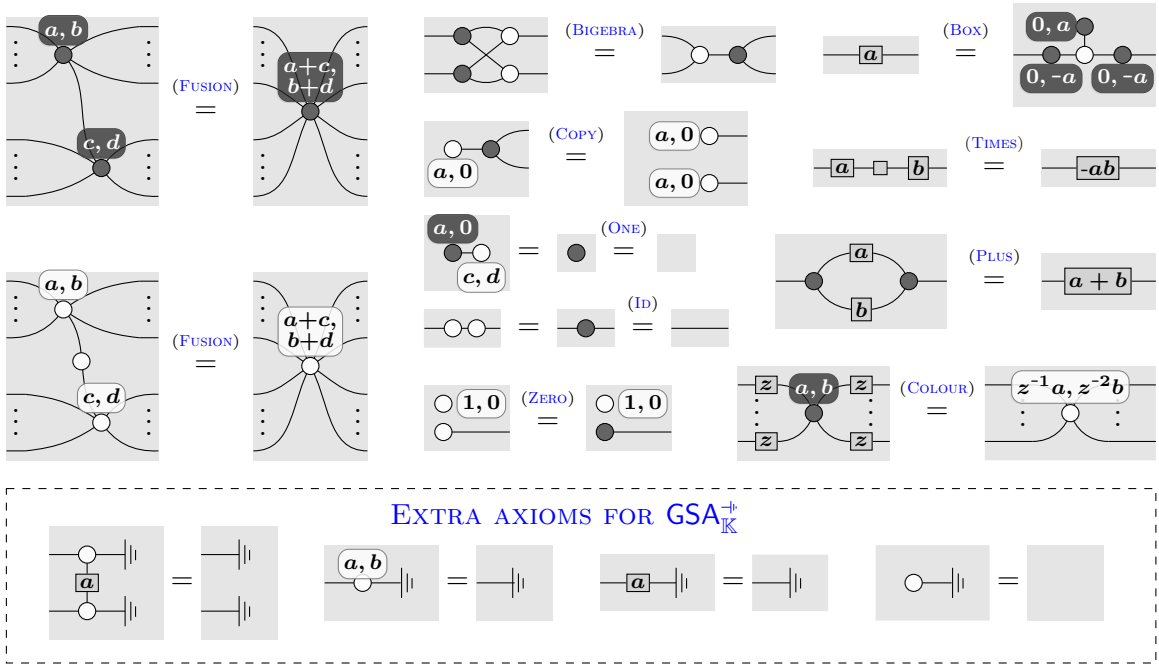
  \begin{figure}[H]
    \makebox[\textwidth][c]{\input{figures/axioms/axioms_zx.tikz}}
  \end{figure}

  \item \textbf{Composition} is given by horizontally connecting output wires
  to input wires.
  \item \textbf{The monoidal product} is given on objects by addition,
  and on morphisms by vertical juxtaposition.
  \item \textbf{The compact structure} is obtained from the grey Frobenius algebra:
  \begin{equation*}
   \input{figures/AffLagRel/cup.tikz} =  \input{figures/AffLagRel/cup_explicit.tikz}
   \qand
   \input{figures/AffLagRel/cap.tikz} =  \input{figures/AffLagRel/cap_explicit.tikz}
  \end{equation*}
  \item \textbf{The dagger} is given by:
  \begin{equation*}
    \input{figures/AffLagRel/b-spider.tikz} \longmapsto \input{figures/AffLagRel/b-spider-daggered.tikz}  \quad\quad
    \input{figures/AffLagRel/w-spider.tikz} \longmapsto \input{figures/AffLagRel/w-spider-daggered.tikz}  \quad\quad
    \input{figures/AffLagRel/box.tikz} \longmapsto
\input{figures/AffLagRel/minusbox.tikz} 
  \end{equation*}
  \end{itemize}
\end{definition}

\subsection{Axioms for \texorpdfstring{\(\ZX^0\) (linear \(\ZX\))}{linear GSA}}
\label{ssec:zx-axioms-linear}

\begin{definition}
  Let \(\ZX^0\) denote the \dag-compact prop with:
  \begin{itemize}
  \item \textbf{Objects} natural numbers.
  \item \textbf{Morphisms} are presented by the following generators indexed by \(n,m \in \N\) and \(a\in \K\):
  \begin{equation*}
    \begin{tikzpicture}
	\begin{pgfonlayer}{nodelayer}
		\node [style=none] (19) at (-1, 0.75) {};
		\node [style=none] (20) at (2, 0.75) {};
		\node [style=none] (21) at (2, -0.75) {};
		\node [style=none] (22) at (-1, -0.75) {};
		\node [style=bn] (23) at (0.5, 0) {};
		\node [style=none] (24) at (1.5, 0.5) {};
		\node [style=none] (25) at (1.5, -0.5) {};
		\node [style=none] (26) at (-0.5, 0.5) {};
		\node [style=none] (27) at (-0.5, -0.5) {};
		\node [style=none] (28) at (-0.525, 0) {$m$};
		\node [style=none] (29) at (2, 0.5) {};
		\node [style=none] (30) at (2, -0.5) {};
		\node [style=none] (31) at (-0.5, 0.5) {};
		\node [style=none] (32) at (-0.5, -0.5) {};
		\node [style=none] (33) at (-0.5, 0.5) {};
		\node [style=none] (34) at (-0.5, -0.5) {};
		\node [style=none] (35) at (0, 0.2) {$\vdots$};
		\node [style=none] (36) at (1.375, 0) {$n$};
		\node [style=none] (37) at (1, 0.2) {$\vdots$};
		\node [style=none] (38) at (-1, 0.5) {};
		\node [style=none] (39) at (-1, -0.5) {};
		\node [style=bphase] (18) at (0.5, 0.5) {$a$};
	\end{pgfonlayer}
	\begin{pgfonlayer}{edgelayer}
		\draw [style=background] (21.center)
			 to (20.center)
			 to (19.center)
			 to (22.center)
			 to cycle;
		\draw [in=180, out=60] (23) to (24.center);
		\draw [in=180, out=-60] (23) to (25.center);
		\draw [in=360, out=120] (23) to (26.center);
		\draw [in=0, out=-120] (23) to (27.center);
		\draw (29.center) to (24.center);
		\draw (30.center) to (25.center);
		\draw [style=multiplexer] (31.center) to (33.center);
		\draw [style=multiplexer] (32.center) to (34.center);
		\draw (33.center) to (38.center);
		\draw (34.center) to (39.center);
	\end{pgfonlayer}
\end{tikzpicture}
 : m \to n  \quad\quad\quad
    \begin{tikzpicture}
	\begin{pgfonlayer}{nodelayer}
		\node [style=none] (19) at (-1, 0.75) {};
		\node [style=none] (20) at (2, 0.75) {};
		\node [style=none] (21) at (2, -0.75) {};
		\node [style=none] (22) at (-1, -0.75) {};
		\node [style=wn] (23) at (0.5, 0) {};
		\node [style=none] (24) at (1.5, 0.5) {};
		\node [style=none] (25) at (1.5, -0.5) {};
		\node [style=none] (26) at (-0.5, 0.5) {};
		\node [style=none] (27) at (-0.5, -0.5) {};
		\node [style=none] (28) at (-0.525, 0) {$m$};
		\node [style=none] (29) at (2, 0.5) {};
		\node [style=none] (30) at (2, -0.5) {};
		\node [style=none] (31) at (-0.5, 0.5) {};
		\node [style=none] (32) at (-0.5, -0.5) {};
		\node [style=none] (33) at (-0.5, 0.5) {};
		\node [style=none] (34) at (-0.5, -0.5) {};
		\node [style=none] (35) at (0, 0.2) {$\vdots$};
		\node [style=none] (36) at (1.375, 0) {$n$};
		\node [style=none] (37) at (1, 0.2) {$\vdots$};
		\node [style=none] (38) at (-1, 0.5) {};
		\node [style=none] (39) at (-1, -0.5) {};
		\node [style=wphase] (18) at (0.5, 0.5) {$a$};
	\end{pgfonlayer}
	\begin{pgfonlayer}{edgelayer}
		\draw [style=background] (21.center)
			 to (20.center)
			 to (19.center)
			 to (22.center)
			 to cycle;
		\draw [in=180, out=60] (23) to (24.center);
		\draw [in=180, out=-60] (23) to (25.center);
		\draw [in=360, out=120] (23) to (26.center);
		\draw [in=0, out=-120] (23) to (27.center);
		\draw (29.center) to (24.center);
		\draw (30.center) to (25.center);
		\draw [style=multiplexer] (31.center) to (33.center);
		\draw [style=multiplexer] (32.center) to (34.center);
		\draw (33.center) to (38.center);
		\draw (34.center) to (39.center);
	\end{pgfonlayer}
\end{tikzpicture}
 : m \to n  \quad\quad\quad
    \input{figures/AffLagRel/box.tikz} : 1 \to 1 
  \end{equation*}
  Modulo the equations imposing flexsymmetry, in addition to the following equations for all \(a,b \in \mathbb{K}\) and \(z
  \in \mathbb{K}^*\):
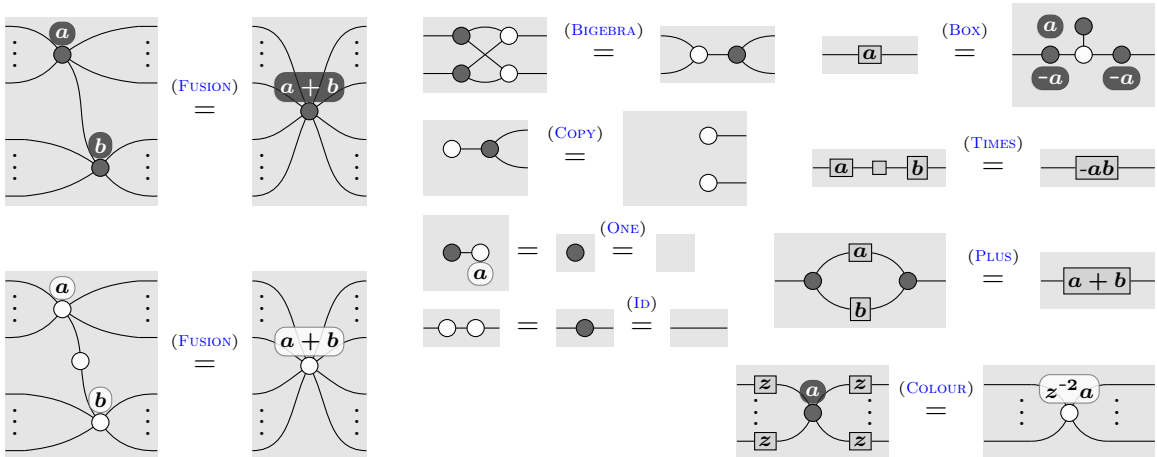
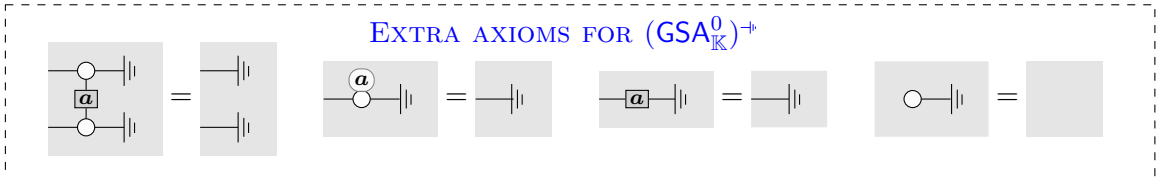
  \begin{figure}[H]
    \makebox[\textwidth][c]{\input{figures/axioms/axioms_zx_linear.tikz}}
  \end{figure}
  \item \textbf{Composition} is given by horizontally connecting output wires
  to input wires.
  \item \textbf{The monoidal product} is given on objects by addition,
  and on morphisms by vertical juxtaposition.
  \item \textbf{The compact structure} is obtained from the grey Frobenius algebra:
  \begin{equation*}
   \input{figures/AffLagRel/cup.tikz} =  \input{figures/AffLagRel/cup_explicit.tikz}
   \qand
   \input{figures/AffLagRel/cap.tikz} =  \input{figures/AffLagRel/cap_explicit.tikz}
  \end{equation*}
  \item \textbf{The dagger} is given by:
  \begin{equation*}
    \begin{tikzpicture}
	\begin{pgfonlayer}{nodelayer}
		\node [style=none] (19) at (-1, 0.75) {};
		\node [style=none] (20) at (2, 0.75) {};
		\node [style=none] (21) at (2, -0.75) {};
		\node [style=none] (22) at (-1, -0.75) {};
		\node [style=bn] (23) at (0.5, 0) {};
		\node [style=none] (24) at (1.5, 0.5) {};
		\node [style=none] (25) at (1.5, -0.5) {};
		\node [style=none] (26) at (-0.5, 0.5) {};
		\node [style=none] (27) at (-0.5, -0.5) {};
		\node [style=none] (28) at (-0.525, 0) {$m$};
		\node [style=none] (29) at (2, 0.5) {};
		\node [style=none] (30) at (2, -0.5) {};
		\node [style=none] (31) at (-0.5, 0.5) {};
		\node [style=none] (32) at (-0.5, -0.5) {};
		\node [style=none] (33) at (-0.5, 0.5) {};
		\node [style=none] (34) at (-0.5, -0.5) {};
		\node [style=none] (35) at (0, 0.2) {$\vdots$};
		\node [style=none] (36) at (1.375, 0) {$n$};
		\node [style=none] (37) at (1, 0.2) {$\vdots$};
		\node [style=none] (38) at (-1, 0.5) {};
		\node [style=none] (39) at (-1, -0.5) {};
		\node [style=bphase] (18) at (0.5, 0.5) {$a$};
	\end{pgfonlayer}
	\begin{pgfonlayer}{edgelayer}
		\draw [style=background] (21.center)
			 to (20.center)
			 to (19.center)
			 to (22.center)
			 to cycle;
		\draw [in=180, out=60] (23) to (24.center);
		\draw [in=180, out=-60] (23) to (25.center);
		\draw [in=360, out=120] (23) to (26.center);
		\draw [in=0, out=-120] (23) to (27.center);
		\draw (29.center) to (24.center);
		\draw (30.center) to (25.center);
		\draw [style=multiplexer] (31.center) to (33.center);
		\draw [style=multiplexer] (32.center) to (34.center);
		\draw (33.center) to (38.center);
		\draw (34.center) to (39.center);
	\end{pgfonlayer}
\end{tikzpicture}
 \longmapsto \begin{tikzpicture}
	\begin{pgfonlayer}{nodelayer}
		\node [style=none] (19) at (-1, 0.75) {};
		\node [style=none] (20) at (2, 0.75) {};
		\node [style=none] (21) at (2, -0.75) {};
		\node [style=none] (22) at (-1, -0.75) {};
		\node [style=bphase] (18) at (0.5, 0.5) {$\minu b$};
		\node [style=bn] (23) at (0.5, 0) {};
		\node [style=none] (24) at (1.5, 0.5) {};
		\node [style=none] (25) at (1.5, -0.5) {};
		\node [style=none] (26) at (-0.5, 0.5) {};
		\node [style=none] (27) at (-0.5, -0.5) {};
		\node [style=none] (28) at (1.5, 0) {$m$};
		\node [style=none] (29) at (2, 0.5) {};
		\node [style=none] (30) at (2, -0.5) {};
		\node [style=none] (31) at (-0.5, 0.5) {};
		\node [style=none] (32) at (-0.5, -0.5) {};
		\node [style=none] (33) at (-0.5, 0.5) {};
		\node [style=none] (34) at (-0.5, -0.5) {};
		\node [style=none] (35) at (0.975, 0.2) {$\vdots$};
		\node [style=none] (36) at (-0.4, 0) {$n$};
		\node [style=none] (37) at (-0.025, 0.2) {$\vdots$};
		\node [style=none] (38) at (-1, 0.5) {};
		\node [style=none] (39) at (-1, -0.5) {};
	\end{pgfonlayer}
	\begin{pgfonlayer}{edgelayer}
		\draw [style=background] (21.center)
			 to (20.center)
			 to (19.center)
			 to (22.center)
			 to cycle;
		\draw [in=180, out=60] (23) to (24.center);
		\draw [in=180, out=-60] (23) to (25.center);
		\draw [in=360, out=120] (23) to (26.center);
		\draw [in=0, out=-120] (23) to (27.center);
		\draw (29.center) to (24.center);
		\draw (30.center) to (25.center);
		\draw [style=multiplexer] (31.center) to (33.center);
		\draw [style=multiplexer] (32.center) to (34.center);
		\draw (33.center) to (38.center);
		\draw (34.center) to (39.center);
	\end{pgfonlayer}
\end{tikzpicture}
  \quad\quad
    \begin{tikzpicture}
	\begin{pgfonlayer}{nodelayer}
		\node [style=none] (19) at (-1, 0.75) {};
		\node [style=none] (20) at (2, 0.75) {};
		\node [style=none] (21) at (2, -0.75) {};
		\node [style=none] (22) at (-1, -0.75) {};
		\node [style=wn] (23) at (0.5, 0) {};
		\node [style=none] (24) at (1.5, 0.5) {};
		\node [style=none] (25) at (1.5, -0.5) {};
		\node [style=none] (26) at (-0.5, 0.5) {};
		\node [style=none] (27) at (-0.5, -0.5) {};
		\node [style=none] (28) at (-0.525, 0) {$m$};
		\node [style=none] (29) at (2, 0.5) {};
		\node [style=none] (30) at (2, -0.5) {};
		\node [style=none] (31) at (-0.5, 0.5) {};
		\node [style=none] (32) at (-0.5, -0.5) {};
		\node [style=none] (33) at (-0.5, 0.5) {};
		\node [style=none] (34) at (-0.5, -0.5) {};
		\node [style=none] (35) at (0, 0.2) {$\vdots$};
		\node [style=none] (36) at (1.375, 0) {$n$};
		\node [style=none] (37) at (1, 0.2) {$\vdots$};
		\node [style=none] (38) at (-1, 0.5) {};
		\node [style=none] (39) at (-1, -0.5) {};
		\node [style=wphase] (18) at (0.5, 0.5) {$a$};
	\end{pgfonlayer}
	\begin{pgfonlayer}{edgelayer}
		\draw [style=background] (21.center)
			 to (20.center)
			 to (19.center)
			 to (22.center)
			 to cycle;
		\draw [in=180, out=60] (23) to (24.center);
		\draw [in=180, out=-60] (23) to (25.center);
		\draw [in=360, out=120] (23) to (26.center);
		\draw [in=0, out=-120] (23) to (27.center);
		\draw (29.center) to (24.center);
		\draw (30.center) to (25.center);
		\draw [style=multiplexer] (31.center) to (33.center);
		\draw [style=multiplexer] (32.center) to (34.center);
		\draw (33.center) to (38.center);
		\draw (34.center) to (39.center);
	\end{pgfonlayer}
\end{tikzpicture}
 \longmapsto \begin{tikzpicture}
	\begin{pgfonlayer}{nodelayer}
		\node [style=none] (19) at (-1, 0.75) {};
		\node [style=none] (20) at (2, 0.75) {};
		\node [style=none] (21) at (2, -0.75) {};
		\node [style=none] (22) at (-1, -0.75) {};
		\node [style=wphase] (18) at (0.5, 0.5) {$\minu b$};
		\node [style=wn] (23) at (0.5, 0) {};
		\node [style=none] (24) at (1.5, 0.5) {};
		\node [style=none] (25) at (1.5, -0.5) {};
		\node [style=none] (26) at (-0.5, 0.5) {};
		\node [style=none] (27) at (-0.5, -0.5) {};
		\node [style=none] (28) at (1.5, 0) {$m$};
		\node [style=none] (29) at (2, 0.5) {};
		\node [style=none] (30) at (2, -0.5) {};
		\node [style=none] (31) at (-0.5, 0.5) {};
		\node [style=none] (32) at (-0.5, -0.5) {};
		\node [style=none] (33) at (-0.5, 0.5) {};
		\node [style=none] (34) at (-0.5, -0.5) {};
		\node [style=none] (35) at (0.975, 0.2) {$\vdots$};
		\node [style=none] (36) at (-0.4, 0) {$n$};
		\node [style=none] (37) at (-0.025, 0.2) {$\vdots$};
		\node [style=none] (38) at (-1, 0.5) {};
		\node [style=none] (39) at (-1, -0.5) {};
	\end{pgfonlayer}
	\begin{pgfonlayer}{edgelayer}
		\draw [style=background] (21.center)
			 to (20.center)
			 to (19.center)
			 to (22.center)
			 to cycle;
		\draw [in=180, out=60] (23) to (24.center);
		\draw [in=180, out=-60] (23) to (25.center);
		\draw [in=360, out=120] (23) to (26.center);
		\draw [in=0, out=-120] (23) to (27.center);
		\draw (29.center) to (24.center);
		\draw (30.center) to (25.center);
		\draw [style=multiplexer] (31.center) to (33.center);
		\draw [style=multiplexer] (32.center) to (34.center);
		\draw (33.center) to (38.center);
		\draw (34.center) to (39.center);
	\end{pgfonlayer}
\end{tikzpicture}
  \quad\quad
    \input{figures/AffLagRel/box.tikz} \longmapsto
\input{figures/AffLagRel/minusbox.tikz} 
  \end{equation*}
  \end{itemize}
\end{definition}

\subsection{Scalable  \texorpdfstring{\(\ZX\)}{GSA}}
\label{ssec:axioms_scalable}
\begin{definition}
  The strictification \([\ZX]\) can be explicitly presented by generators
  and relations:
  \begin{itemize}
  \item \textbf{Objects} are elements of  \(\operatorname{List}(\N)\).
  \item \textbf{Morphisms} are presented by the following generators indexed
  by \(n,m,k, \ell \in \N\), \(\vec v \in \N^n\) and \(x,y\in \K\):
  \begin{align*}
    \input{figures/AffLagRel/scalable/b-spider.tikz} &: [k, \overset{m}{\cdots}, k] \to [k, \overset{n}{\cdots}, k]&
    \quad\quad
    \input{figures/AffLagRel/scalable/w-spider.tikz}  &: [k, \overset{m}{\cdots}, k] \to [k, \overset{n}{\cdots}, k] &
    \\
    \input{figures/AffLagRel/scalable/box.tikz} &: [k] \to [k] &
    \quad\quad
    \input{figures/AffLagRel/scalable/arrow.tikz} &: [k] \to [\ell] &
    \\
    \input{figures/AffLagRel/scalable/dividermonoid.tikz} &: \left[\sum_{i=0}^{n-1}v_i\right] \to [v_0,\cdots, v_{n-1}] &
    \quad\quad
    \input{figures/AffLagRel/scalable/gatherermonoid.tikz} &: [v_0,\cdots, v_{n-1}]   \to \left[\sum_{i=0}^{n-1}v_i \right] &
  \end{align*}
  modulo the equations that the phased spiders and boxes are flexsymmetric in addition to the following equations for all \(n,m,k \in \N\),  \(\vec x,\vec y \in \K^n\), \(\vec w \in \K^m\), \(X,Y\in\Sym[n][\K]\), \(W \in \Sym[m][\K]\),  \(A,B,Z\in \Matrices[m][n][\K]\), \(C\in \Matrices[k][n][\K]\), \(D\in \Matrices[m][\ell][\K]\), \(E\in \Matrices[k][\ell][\K]\),  with \(Z\) surjective:\\

  \begin{figure}[H]

    \makebox[\textwidth][c]{\input{figures/axioms/axioms_scalable.tikz}}

  \end{figure}
  \item \textbf{Composition} is given by horizontally connecting output wires
  to input wires.
  \item \textbf{The monoidal product} is given on objects by concatenation of lists,
  and on morphisms by vertical juxtaposition.
  \item \textbf{The compact structure} is obtained from the grey Frobenius algebra:
  \begin{equation*}
   \input{figures/AffLagRel/thick/cup.tikz} =  \input{figures/AffLagRel/thick/cup_explicit.tikz}
   \qand
   \input{figures/AffLagRel/thick/cap.tikz} =  \input{figures/AffLagRel/thick/cap_explicit.tikz}
  \end{equation*}
  \item \textbf{The dagger} is given by:
  \begin{equation*}
    \input{figures/AffLagRel/thick/b-spider.tikz} \mapsto \input{figures/AffLagRel/thick/b-spider-daggered.tikz}  \quad
    \input{figures/AffLagRel/thick/w-spider.tikz} \mapsto \input{figures/AffLagRel/thick/w-spider-daggered.tikz}  \quad
    \input{figures/AffLagRel/thick/box.tikz} \mapsto \input{figures/AffLagRel/thick/box-daggered.tikz} \quad
    \input{figures/AffLagRel/thick/matrix.tikz} \mapsto \input{figures/AffLagRel/thick/matrix-daggered.tikz} 
\end{equation*}
\begin{equation*}
    \input{figures/AffLagRel/scalable/gatherermonoid.tikz} \longmapsto \input{figures/AffLagRel/scalable/dividermonoid.tikz} \quad\quad \input{figures/AffLagRel/scalable/dividermonoid.tikz}\longmapsto  \input{figures/AffLagRel/scalable/gatherermonoid.tikz}
  \end{equation*}
  \end{itemize}
\end{definition}

We prove \TTextFusion in lemma~\ref{lem:scalable_fusion} and \TTextColour in lemma~\ref{lem:scalable_colour_actual}.  All of the other axioms have already been proven, or follow by simple induction arguments.  In the thickened setting, remark by lemma~\ref{lem:box_scalable}, we don't need a thickened version of \TextHBox.

Note that we have not given a flexsymmetric presentation of the strictification of \(\ZX\) because the dividers and gathers can not be made to be symmetric!

\section{Proofs}
\label{app:proofs}
\subsection{Proofs of section~\ref{sec:GAA}}

\begin{proof}[Proof of proposition \ref{prop:matprop}]
  
  We start with the first set of proposition:
  
  \begin{description}
    \item[$(1)\Rightarrow (2)$] Taking the relational semantics, $(2)$ can be reformulated as: for all $x,y,z\in \K$, $Ax=Ay=z$ is equivalent to $x=y$ and $z=Ax=Ay$. Which is true if $A$ is injective.
    \item[$(2)\Rightarrow (3)$] \input{figures/AffRel/inj2to3.tikz}
    \item[$(3)\Rightarrow (4)$]\input{figures/AffRel/inj3to4.tikz}
    \item[$(4)\Rightarrow (1)$] Taking the relational semantics, $(4)$ amounts to $\ker(A) = \{0\}$ which implies that $A$ is injective.
  \end{description}
  
  Then the second set of propositions:
  
  \begin{description}
    \item[$(1)\Rightarrow (2)$] Taking the relational semantics, $(2)$ can be reformulated as: for all $x,y,z\in \K$, $-Ax=y+z$ is equivalent to there exists $a,b\in\K$, such that $-x=a+b$, $y=Aa$ and $z=Ab$. Which is true if $A$ is surjective by taking premimages of $y$ and $z$.
    \item[$(2)\Rightarrow (3)$]\input{figures/AffRel/surj2to3.tikz}
    \item[$(3)\Rightarrow (4)$]\input{figures/AffRel/surj3to4.tikz}
    \item[$(4)\Rightarrow (1)$] Taking the relational semantics, $(4)$ amounts to $\Im(A)$ being equal to the full space, which implies that $A$ is surjective.
  \end{description}
  
\end{proof}

\subsection{Proofs of section~\ref{ssec:symplectic_ZX}}

\begin{lemma}
  \label{lem:antipode}
  \[\input{figures/soundness/lemmas/antipode.tikz}\]
\end{lemma}
\begin{proof}
  This follows immediately from \TextGAAidoneii and the uniqueness of inverses.
\end{proof}

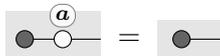
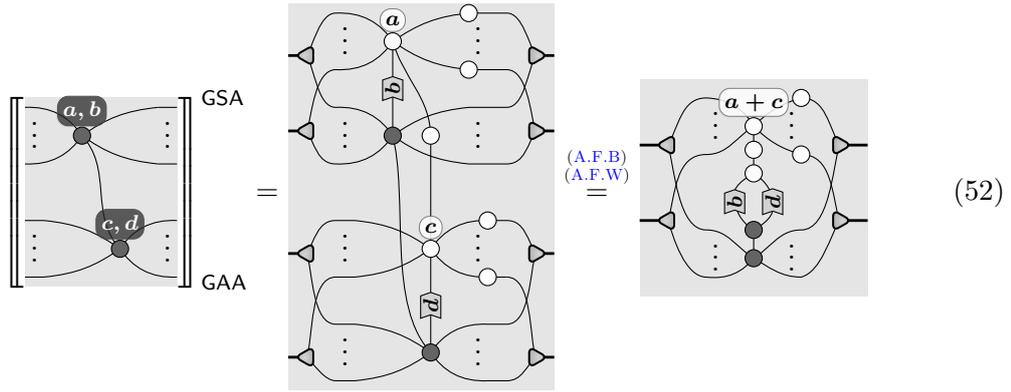
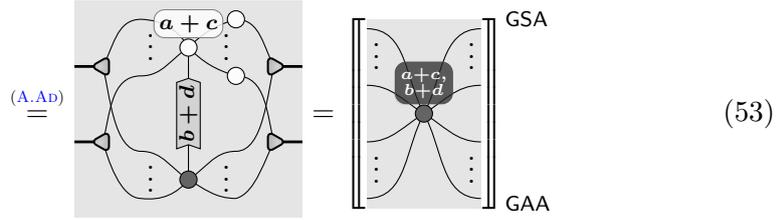
\begin{lemma}
  \label{lem:affinediscard}
  \[\input{figures/soundness/lemmas/affinediscard.tikz}\]
\end{lemma}
\begin{proof}
  \begin{equation}
    \input{figures/soundness/lemmas/affinediscardproof.tikz}
  \end{equation}
\end{proof}

\begin{lemma}
  \label{lem:inversetranspose}
  Given \(a \in \K^*\):
  \[\input{figures/soundness/lemmas/inversetranspose.tikz}\]
\end{lemma}
\begin{proof}
  Firstly:
  \begin{equation}
    \input{figures/soundness/lemmas/inversetransposeproof.tikz}
  \end{equation}

  Secondly:

  \begin{equation}
    \input{figures/soundness/lemmas/inversetransposeproofii.tikz}
  \end{equation}
\end{proof}

\begin{lemma}
  \label{lem:cnotswap}
  Given \(a \in \K^*\):
  \[\input{figures/soundness/lemmas/cnotswap.tikz}\]
\end{lemma}
\begin{proof}
  \[\input{figures/soundness/lemmas/cnotswapproof.tikz}\]
\end{proof}

% \structuralSoundness*
\begin{proof}[Proof of proposition~\ref{thm:soundness}]
The only nontrivial thing to check is the functoriality of \(\interp{-}_{\gaa}^\zx\).  To this end, we show that al of the axioms of \(\ZX\) are sound:

  (\TextFusion)
  \begin{align}\label{lem:fusion}
    \interp{\begin{tikzpicture}
	\begin{pgfonlayer}{nodelayer}
		\node [style=gn, label={[gphase]below:$a,b$}] (0) at (1.5, 1.5) {};
		\node [style=none] (1) at (0, 2.25) {};
		\node [style=none] (2) at (0, 0.75) {};
		\node [style=none] (3) at (3.5, 2.25) {};
		\node [style=none] (4) at (3.5, 0.75) {};
		\node [style=none] (5) at (3.75, 1.75) {$\vdots$};
		\node [style=none] (6) at (0.25, 1.75) {$\vdots$};
		\node [style=gn, label={[gphase]below:$c,d$}] (7) at (2.5, -1.5) {};
		\node [style=none] (8) at (0.5, -0.75) {};
		\node [style=none] (9) at (0.5, -2.25) {};
		\node [style=none] (10) at (4, -0.75) {};
		\node [style=none] (11) at (4, -2.25) {};
		\node [style=none] (12) at (3.75, -1.25) {$\vdots$};
		\node [style=none] (13) at (0.25, -1.25) {$\vdots$};
		\node [style=none] (14) at (0, -0.75) {};
		\node [style=none] (15) at (0, -2.25) {};
		\node [style=none] (16) at (4, 2.25) {};
		\node [style=none] (17) at (4, 0.75) {};
		\node [style=none] (18) at (0, 2.5) {};
		\node [style=none] (19) at (4, 2.5) {};
		\node [style=none] (20) at (4, -2.5) {};
		\node [style=none] (21) at (0, -2.5) {};
	\end{pgfonlayer}
	\begin{pgfonlayer}{edgelayer}
		\draw [style=background] (21.center)
			 to (18.center)
			 to (19.center)
			 to (20.center)
			 to cycle;
		\draw [bend left=15] (0) to (3.center);
		\draw [bend right=15] (0) to (4.center);
		\draw [bend right] (0) to (1.center);
		\draw [bend left] (0) to (2.center);
		\draw [bend left] (7) to (10.center);
		\draw [bend right] (7) to (11.center);
		\draw [bend right=15] (7) to (8.center);
		\draw [bend left=15] (7) to (9.center);
		\draw (14.center) to (8.center);
		\draw (15.center) to (9.center);
		\draw (4.center) to (17.center);
		\draw (3.center) to (16.center);
		\draw [in=135, out=-45] (0) to (7);
	\end{pgfonlayer}
\end{tikzpicture}
}_\gaa^\zx
    &=
      \begin{tikzpicture}
	\begin{pgfonlayer}{nodelayer}
		\node [style=wn, label={[wphase]below:$a$}] (0) at (1.5, 4) {};
		\node [style=none] (1) at (0, 4.75) {};
		\node [style=none] (2) at (0, 3.25) {};
		\node [style=wn] (3) at (3.5, 4.75) {};
		\node [style=wn] (4) at (3.5, 3.25) {};
		\node [style=none] (5) at (3.75, 4.25) {$\vdots$};
		\node [style=none] (6) at (0.25, 4.25) {$\vdots$};
		\node [style=wn, label={[wphase]below:$c$}] (7) at (2.5, -1.5) {};
		\node [style=none] (8) at (0.5, -0.75) {};
		\node [style=none] (9) at (0.5, -2.25) {};
		\node [style=wn] (10) at (4, -0.75) {};
		\node [style=wn] (11) at (4, -2.25) {};
		\node [style=none] (12) at (3.75, -1.25) {$\vdots$};
		\node [style=none] (13) at (0.25, -1.25) {$\vdots$};
		\node [style=none] (14) at (0, -0.75) {};
		\node [style=none] (15) at (0, -2.25) {};
		\node [style=none] (16) at (4.5, 4.75) {};
		\node [style=none] (17) at (4.5, 3.25) {};
		\node [style=bn] (18) at (1.5, 1.5) {};
		\node [style=none] (19) at (0, 2.25) {};
		\node [style=none] (20) at (0, 0.75) {};
		\node [style=none] (21) at (3.5, 2.25) {};
		\node [style=none] (22) at (3.5, 0.75) {};
		\node [style=none] (23) at (3.75, 1.75) {$\vdots$};
		\node [style=none] (24) at (0.25, 1.75) {$\vdots$};
		\node [style=bn] (25) at (2.5, -4.25) {};
		\node [style=none] (26) at (0.5, -3.5) {};
		\node [style=none] (27) at (0.5, -5) {};
		\node [style=none] (28) at (4, -3.5) {};
		\node [style=none] (29) at (4, -5) {};
		\node [style=none] (30) at (3.75, -4) {$\vdots$};
		\node [style=none] (31) at (0.25, -4) {$\vdots$};
		\node [style=none] (32) at (0, -3.5) {};
		\node [style=none] (33) at (0, -5) {};
		\node [style=none] (34) at (4.5, 2.25) {};
		\node [style=none] (35) at (4.5, 0.75) {};
		\node [style=umat] (36) at (1.5, 2.75) {$b$};
		\node [style=umat] (37) at (2.5, -3) {$d$};
		\node [style=wn] (38) at (2.5, 1.5) {};
		\node [style=none] (39) at (4.5, -0.75) {};
		\node [style=none] (40) at (4.5, -2.25) {};
		\node [style=none] (41) at (4.5, -3.5) {};
		\node [style=none] (42) at (4.5, -5) {};
		\node [style=none] (43) at (-1.25, 5) {};
		\node [style=none] (44) at (5.75, 5) {};
		\node [style=none] (45) at (5.75, -5.25) {};
		\node [style=none] (46) at (-1.25, -5.25) {};
		\node [style=gather] (47) at (5.25, 3.625) {};
		\node [style=divide] (48) at (-0.75, 3.625) {};
		\node [style=none] (49) at (5.75, 3.625) {};
		\node [style=none] (50) at (-1.25, 3.625) {};
		\node [style=gather] (51) at (5.25, 1.625) {};
		\node [style=divide] (52) at (-0.75, 1.625) {};
		\node [style=none] (53) at (5.75, 1.625) {};
		\node [style=none] (54) at (-1.25, 1.625) {};
		\node [style=gather] (55) at (5.25, -1.625) {};
		\node [style=divide] (56) at (-0.75, -1.625) {};
		\node [style=none] (57) at (5.75, -1.625) {};
		\node [style=none] (58) at (-1.25, -1.625) {};
		\node [style=gather] (59) at (5.25, -4.375) {};
		\node [style=divide] (60) at (-0.75, -4.375) {};
		\node [style=none] (61) at (5.75, -4.375) {};
		\node [style=none] (62) at (-1.25, -4.375) {};
	\end{pgfonlayer}
	\begin{pgfonlayer}{edgelayer}
		\draw [style=background] (46.center)
			 to (43.center)
			 to (44.center)
			 to (45.center)
			 to cycle;
		\draw [bend left=15] (0) to (3);
		\draw [bend right=15] (0) to (4);
		\draw (0)
			 to [bend right] (1.center)
			 to [in=60, out=180, looseness=0.75] (48);
		\draw (52)
			 to [in=180, out=75, looseness=0.75] (2.center)
			 to [bend right] (0);
		\draw [bend left] (7) to (10);
		\draw [bend right] (7) to (11);
		\draw (7)
			 to [bend right=15] (8.center)
			 to (14.center)
			 to [in=60, out=-180] (56);
		\draw (7)
			 to [bend left=15] (9.center)
			 to (15.center)
			 to [in=75, out=-180, looseness=0.75] (60);
		\draw (4)
			 to (17.center)
			 to [in=105, out=0, looseness=0.75] (51);
		\draw (3)
			 to (16.center)
			 to [in=120, out=0, looseness=0.75] (47);
		\draw (18)
			 to [bend left=15] (21.center)
			 to (34.center)
			 to [in=240, out=0, looseness=0.75] (47);
		\draw (51)
			 to [in=0, out=-120] (35.center)
			 to (22.center)
			 to [bend left=15] (18);
		\draw (18)
			 to [bend right] (19.center)
			 to [in=300, out=180, looseness=0.75] (48);
		\draw (52)
			 to [in=-180, out=-60] (20.center)
			 to [bend right] (18);
		\draw (25)
			 to [bend left] (28.center)
			 to (41.center)
			 to [in=255, out=0, looseness=0.75] (55);
		\draw (25)
			 to [bend right] (29.center)
			 to (42.center)
			 to [in=225, out=0] (59);
		\draw (56)
			 to [in=180, out=285, looseness=0.75] (32.center)
			 to (26.center)
			 to [bend left=15] (25);
		\draw (25)
			 to [bend left=15] (27.center)
			 to (33.center)
			 to [in=-45, out=180, looseness=0.75] (60);
		\draw [in=135, out=-75, looseness=0.75] (18) to (25);
		\draw (0) to (18);
		\draw (7) to (25);
		\draw [in=285, out=90] (38) to (0);
		\draw (38) to (7);
		\draw (10)
			 to (39.center)
			 to [in=120, out=0] (55);
		\draw (11) to (40.center);
		\draw [style=very thick] (48) to (50.center);
		\draw [style=very thick] (49.center) to (47);
		\draw [style=very thick] (52) to (54.center);
		\draw [style=very thick] (53.center) to (51);
		\draw [style=very thick] (56) to (58.center);
		\draw [style=very thick] (57.center) to (55);
		\draw [style=very thick] (60) to (62.center);
		\draw [style=very thick] (61.center) to (59);
		\draw [in=0, out=105, looseness=0.75] (59) to (40.center);
	\end{pgfonlayer}
\end{tikzpicture} 
    \stackeqmid{\GAAFusionB \\ \GAAFusionW}
     \begin{tikzpicture}
	\begin{pgfonlayer}{nodelayer}
		\node [style=wn, label={[wphase]below:$a+c$}] (0) at (2, 1.75) {};
		\node [style=none] (1) at (0.75, 2.5) {};
		\node [style=none] (2) at (0.75, 1) {};
		\node [style=wn] (3) at (3.25, 2.5) {};
		\node [style=wn] (4) at (3.25, 1) {};
		\node [style=none] (5) at (3, 2) {$\vdots$};
		\node [style=none] (6) at (1, 2) {$\vdots$};
		\node [style=none] (7) at (3.25, 2.5) {};
		\node [style=none] (8) at (3.25, 1) {};
		\node [style=bn] (9) at (2, -1.75) {};
		\node [style=none] (10) at (0.75, -1) {};
		\node [style=none] (11) at (0.75, -2.25) {};
		\node [style=none] (14) at (3, -1.5) {$\vdots$};
		\node [style=none] (15) at (1, -1.5) {$\vdots$};
		\node [style=none] (16) at (3.25, -1) {};
		\node [style=none] (17) at (3.25, -2.25) {};
		\node [style=umat] (18) at (1.5, -0.25) {$b$};
		\node [style=umat] (19) at (2.5, -0.25) {$d$};
		\node [style=none] (20) at (-1, 3) {};
		\node [style=none] (21) at (5, 3) {};
		\node [style=none] (22) at (5, -2.75) {};
		\node [style=none] (23) at (-1, -2.75) {};
		\node [style=gather] (24) at (4.25, 1.25) {};
		\node [style=gather] (25) at (4.25, -0.75) {};
		\node [style=divide] (26) at (-0.25, 1.25) {};
		\node [style=divide] (27) at (-0.25, -0.75) {};
		\node [style=none] (28) at (5, 1.25) {};
		\node [style=none] (29) at (5, -0.75) {};
		\node [style=none] (30) at (-1, 1.25) {};
		\node [style=none] (31) at (-1, -0.75) {};
		\node [style=wn] (32) at (2, 0.5) {};
		\node [style=bn] (33) at (2, -1) {};
		\node [style=wn] (34) at (2, 1.125) {};
	\end{pgfonlayer}
	\begin{pgfonlayer}{edgelayer}
		\draw [style=background] (21.center)
			 to (22.center)
			 to (23.center)
			 to (20.center)
			 to cycle;
		\draw [bend left=15] (0) to (3);
		\draw [bend right=15] (0) to (4);
		\draw (26)
			 to [in=180, out=60] (1.center)
			 to [bend left] (0);
		\draw (27)
			 to [in=-180, out=75] (2.center)
			 to [in=-150, out=0] (0);
		\draw [in=0, out=105] (25) to (8.center);
		\draw [in=120, out=0] (7.center) to (24);
		\draw (9)
			 to [in=-180, out=30] (16.center)
			 to [in=-105, out=0, looseness=0.75] (24);
		\draw (9)
			 to [in=-180, out=-30] (17.center)
			 to [in=-120, out=15, looseness=0.75] (25);
		\draw (9)
			 to [in=360, out=150] (10.center)
			 to [in=-75, out=180, looseness=0.75] (26);
		\draw (9)
			 to [bend left] (11.center)
			 to [in=-75, out=180, looseness=0.75] (27);
		\draw [style=very thick] (27) to (31.center);
		\draw [style=very thick] (30.center) to (26);
		\draw [style=very thick] (28.center) to (24);
		\draw [style=very thick] (29.center) to (25);
		\draw (9) to (33);
		\draw [in=-90, out=150] (33) to (18);
		\draw [in=-150, out=90] (18) to (32);
		\draw [in=90, out=-30] (32) to (19);
		\draw [in=30, out=-90] (19) to (33);
		\draw (32) to (34);
		\draw (34) to (0);
	\end{pgfonlayer}
\end{tikzpicture}
\\
    \stackeq{\GAAAdd}
    \begin{tikzpicture}
	\begin{pgfonlayer}{nodelayer}
		\node [style=wn, label={[wphase]below:$a+c$}] (0) at (3.5, 1.75) {};
		\node [style=none] (1) at (2.25, 2.5) {};
		\node [style=none] (2) at (2.25, 1) {};
		\node [style=wn] (3) at (4.75, 2.5) {};
		\node [style=wn] (4) at (4.75, 1) {};
		\node [style=none] (5) at (4.5, 2) {$\vdots$};
		\node [style=none] (6) at (2.5, 2) {$\vdots$};
		\node [style=none] (7) at (4.75, 2.5) {};
		\node [style=none] (8) at (4.75, 1) {};
		\node [style=bn] (9) at (3.5, -1.75) {};
		\node [style=none] (10) at (2.25, -1) {};
		\node [style=none] (11) at (2.25, -2.25) {};
		\node [style=none] (12) at (4.5, -1.5) {$\vdots$};
		\node [style=none] (13) at (2.5, -1.5) {$\vdots$};
		\node [style=none] (14) at (4.75, -1) {};
		\node [style=none] (15) at (4.75, -2.25) {};
		\node [style=umat] (16) at (3.5, 0) {$b+d$};
		\node [style=none] (18) at (0.5, 3) {};
		\node [style=none] (19) at (6.5, 3) {};
		\node [style=none] (20) at (6.5, -2.75) {};
		\node [style=none] (21) at (0.5, -2.75) {};
		\node [style=gather] (22) at (5.75, 1.25) {};
		\node [style=gather] (23) at (5.75, -0.75) {};
		\node [style=divide] (24) at (1.25, 1.25) {};
		\node [style=divide] (25) at (1.25, -0.75) {};
		\node [style=none] (26) at (6.5, 1.25) {};
		\node [style=none] (27) at (6.5, -0.75) {};
		\node [style=none] (28) at (0.5, 1.25) {};
		\node [style=none] (29) at (0.5, -0.75) {};
	\end{pgfonlayer}
	\begin{pgfonlayer}{edgelayer}
		\draw [style=background] (19.center)
			 to (20.center)
			 to (21.center)
			 to (18.center)
			 to cycle;
		\draw [bend left=15] (0) to (3);
		\draw [bend right=15] (0) to (4);
		\draw (24)
			 to [in=180, out=60] (1.center)
			 to [bend left] (0);
		\draw (25)
			 to [in=-180, out=75] (2.center)
			 to [in=-150, out=0] (0);
		\draw [in=0, out=105] (23) to (8.center);
		\draw [in=120, out=0] (7.center) to (22);
		\draw (9)
			 to [in=-180, out=30] (14.center)
			 to [in=-105, out=0, looseness=0.75] (22);
		\draw (9)
			 to [in=-180, out=-30] (15.center)
			 to [in=-120, out=15, looseness=0.75] (23);
		\draw (9)
			 to [in=360, out=150] (10.center)
			 to [in=-75, out=180, looseness=0.75] (24);
		\draw (9)
			 to [bend left] (11.center)
			 to [in=-75, out=180, looseness=0.75] (25);
		\draw (9) to (0);
		\draw [style=very thick] (25) to (29.center);
		\draw [style=very thick] (28.center) to (24);
		\draw [style=very thick] (26.center) to (22);
		\draw [style=very thick] (27.center) to (23);
	\end{pgfonlayer}
\end{tikzpicture}

    = \interp{\begin{tikzpicture}
	\begin{pgfonlayer}{nodelayer}
		\node [style=gn, label={[gphase]below:$\substack{a+c,\\b+d}$}] (0) at (1.5, 0) {};
		\node [style=none] (1) at (3, 2.25) {};
		\node [style=none] (2) at (3, 0.75) {};
		\node [style=none] (3) at (3, -0.75) {};
		\node [style=none] (4) at (3, -2.25) {};
		\node [style=none] (5) at (0, 2.25) {};
		\node [style=none] (6) at (0, 0.75) {};
		\node [style=none] (7) at (0, -0.75) {};
		\node [style=none] (8) at (0, -2.25) {};
		\node [style=none] (9) at (2.75, 1.75) {$\vdots$};
		\node [style=none] (10) at (2.75, -1.25) {$\vdots$};
		\node [style=none] (11) at (0.25, 1.75) {$\vdots$};
		\node [style=none] (12) at (0.25, -1.25) {$\vdots$};
		\node [style=none] (13) at (0, 2.5) {};
		\node [style=none] (14) at (3, 2.5) {};
		\node [style=none] (15) at (3, -2.5) {};
		\node [style=none] (16) at (0, -2.5) {};
	\end{pgfonlayer}
	\begin{pgfonlayer}{edgelayer}
		\draw [style=background] (14.center)
			 to (15.center)
			 to (16.center)
			 to (13.center)
			 to cycle;
		\draw [in=180, out=30] (0) to (2.center);
		\draw [in=180, out=60] (0) to (1.center);
		\draw [in=180, out=-30] (0) to (3.center);
		\draw [in=180, out=-60] (0) to (4.center);
		\draw [in=0, out=120] (0) to (5.center);
		\draw [in=0, out=150] (0) to (6.center);
		\draw [in=0, out=-150] (0) to (7.center);
		\draw [in=0, out=-120] (0) to (8.center);
	\end{pgfonlayer}
\end{tikzpicture}}_\gaa^\zx
  \end{align}

  The white version of \TextFusion follows from an analagous argument.

  (\TextBigebra)
  \begin{align}\label{lem:bigebra}
    \interp{\input{figures/soundness/bigebra/bigebra_LHS.tikz}}_\gaa^\zx
    &= \input{figures/soundness/bigebra/bigebra_1.tikz} 
    \stackeqmid{\GAAidoneii\\ \minilemref{lem:antipode} } \input{figures/soundness/bigebra/bigebra_2.tikz} \\
    \stackeq{\GAABigebra\\ \GAAidoneii\\ \minilemref{lem:antipode}} \input{figures/soundness/bigebra/bigebra_3.tikz} 
    =\interp{\input{figures/soundness/bigebra/bigebra_RHS.tikz}}_\gaa^\zx
  \end{align}

  (\TextCopy)
  \begin{align}\label{lem:copy}
    \interp{\input{figures/soundness/copy/copy_LHS.tikz}}_\gaa^\zx
    = \input{figures/soundness/copy/copy_1.tikz} 
    \stackeqmid{\GAACopyi \\ \GAACopyii } 
      \input{figures/soundness/copy/copy_4.tikz} 
    = \interp{\input{figures/soundness/copy/copy_RHS.tikz}}_\gaa^\zx
  \end{align}

  (\TextOne)
  \begin{align}\label{lem:one}
    \interp{\input{figures/soundness/one/one_LHS.tikz}}_\gaa^\zx
    &= \input{figures/soundness/one/one_1.tikz} 
    \stackeqmid{\minilemref{lem:affinediscard}\\ \GAAdisciii \\ \GAACap}
      \input{figures/soundness/one/one_2.tikz} 
    \stackeqmid{\GAAidzeroi}
      \input{figures/soundness/one/one_MHS.tikz} 
    = \interp{\input{figures/soundness/one/one_MHS.tikz}}_\gaa^\zx\\
    &= \interp{\input{figures/soundness/one/one_MHS.tikz}}_\gaa^\zx
    = \input{figures/soundness/one/one_MHS.tikz} 
    \stackeqmid{\GAAidzeroii \\ \GAAidzeroiii}
      \input{figures/soundness/one/one_3.tikz} 
    =
      \interp{\input{figures/soundness/one/one_RHS.tikz}}_\gaa^\zx
  \end{align}

  (\TextZero)
  \begin{align}\label{lem:zero}
    \interp{\begin{tikzpicture}
	\begin{pgfonlayer}{nodelayer}
		\node [style=none] (0) at (1.75, -0.25) {};
		\node [style=rn] (1) at (0.5, -0.25) {};
		\node [style=gn, label={[gphase]left:$1,0$}] (2) at (0.5, 0.5) {};
		\node [style=none] (3) at (0, 1) {};
		\node [style=none] (4) at (1.75, 1) {};
		\node [style=none] (5) at (1.75, -0.75) {};
		\node [style=none] (6) at (0, -0.75) {};
	\end{pgfonlayer}
	\begin{pgfonlayer}{edgelayer}
		\draw [style=background] (3.center)
			 to (4.center)
			 to (5.center)
			 to (6.center)
			 to cycle;
		\draw (0.center) to (1);
	\end{pgfonlayer}
\end{tikzpicture}}_\gaa^\zx
    = \begin{tikzpicture}
	\begin{pgfonlayer}{nodelayer}
		\node [style=bn] (1) at (1, 0.25) {};
		\node [style=wn, label={[wphase]below:$1$}] (2) at (1, 1) {};
		\node [style=none] (3) at (1.75, -0.5) {};
		\node [style=wn] (4) at (1, -0.5) {};
		\node [style=bn] (5) at (2, 1) {};
		\node [style=wn] (6) at (1.75, 0.25) {};
		\node [style=none] (7) at (0.5, 2) {};
		\node [style=none] (8) at (3.5, 2) {};
		\node [style=none] (9) at (3.5, -1) {};
		\node [style=none] (10) at (0.5, -1) {};
		\node [style=gather] (11) at (2.75, -0.125) {};
		\node [style=none] (12) at (3.5, -0.125) {};
	\end{pgfonlayer}
	\begin{pgfonlayer}{edgelayer}
		\draw [style=background] (10.center)
			 to (7.center)
			 to (8.center)
			 to (9.center)
			 to cycle;
		\draw (11)
			 to [in=0, out=-135] (3.center)
			 to (4);
		\draw (1) to (6);
		\draw [in=135, out=0] (6) to (11);
		\draw [style=very thick] (12.center) to (11);
	\end{pgfonlayer}
\end{tikzpicture}
 
    \stackeqmid{\GAAEmpty}
      \begin{tikzpicture}
	\begin{pgfonlayer}{nodelayer}
		\node [style=bn] (0) at (1, 0.25) {};
		\node [style=wn, label={[wphase]below:$1$}] (1) at (1, 1) {};
		\node [style=none] (2) at (1.75, -0.5) {};
		\node [style=bn] (3) at (1, -0.5) {};
		\node [style=bn] (4) at (2, 1) {};
		\node [style=wn] (5) at (1.75, 0.25) {};
		\node [style=none] (6) at (0.5, 2) {};
		\node [style=none] (7) at (3.5, 2) {};
		\node [style=none] (8) at (3.5, -1) {};
		\node [style=none] (9) at (0.5, -1) {};
		\node [style=gather] (10) at (2.75, -0.125) {};
		\node [style=none] (11) at (3.5, -0.125) {};
	\end{pgfonlayer}
	\begin{pgfonlayer}{edgelayer}
		\draw [style=background] (9.center)
			 to (6.center)
			 to (7.center)
			 to (8.center)
			 to cycle;
		\draw (10)
			 to [in=0, out=-135] (2.center)
			 to (3);
		\draw (0) to (5);
		\draw [in=135, out=0] (5) to (10);
		\draw [style=very thick] (11.center) to (10);
	\end{pgfonlayer}
\end{tikzpicture}
 
    = \interp{\begin{tikzpicture}
	\begin{pgfonlayer}{nodelayer}
		\node [style=none] (0) at (1.75, -0.25) {};
		\node [style=rn] (1) at (0.5, -0.25) {};
		\node [style=bn, label={[gphase]left:$1,0$}] (2) at (0.5, 0.5) {};
		\node [style=none] (3) at (0, 1) {};
		\node [style=none] (4) at (1.75, 1) {};
		\node [style=none] (5) at (1.75, -0.75) {};
		\node [style=none] (6) at (0, -0.75) {};
	\end{pgfonlayer}
	\begin{pgfonlayer}{edgelayer}
		\draw [style=background] (3.center)
			 to (4.center)
			 to (5.center)
			 to (6.center)
			 to cycle;
		\draw (0.center) to (1);
	\end{pgfonlayer}
\end{tikzpicture}}_\gaa^\zx
  \end{align}

  (\TextId)
  \begin{align}\label{lem:id}
    \interp{\begin{tikzpicture}
	\begin{pgfonlayer}{nodelayer}
		\node [style=none] (0) at (0, 0) {};
		\node [style=none] (1) at (1.5, 0) {};
		\node [style=gn] (2) at (0.75, 0) {};
		\node [style=none] (3) at (0, 0.5) {};
		\node [style=none] (4) at (1.5, 0.5) {};
		\node [style=none] (5) at (1.5, -0.5) {};
		\node [style=none] (6) at (0, -0.5) {};
	\end{pgfonlayer}
	\begin{pgfonlayer}{edgelayer}
		\draw [style=background] (6.center)
			 to (3.center)
			 to (4.center)
			 to (5.center)
			 to cycle;
		\draw (0.center) to (1.center);
	\end{pgfonlayer}
\end{tikzpicture}}_\gaa^\zx
    = \begin{tikzpicture}
	\begin{pgfonlayer}{nodelayer}
		\node [style=wn] (0) at (2, 0.5) {};
		\node [style=none] (1) at (2, -0.5) {};
		\node [style=none] (2) at (2.75, -0.5) {};
		\node [style=bn] (3) at (2.375, -0.5) {};
		\node [style=wn] (4) at (2.75, 0.5) {};
		\node [style=none] (5) at (0.5, 1) {};
		\node [style=none] (6) at (4.25, 1) {};
		\node [style=none] (7) at (4.25, -1) {};
		\node [style=none] (8) at (0.5, -1) {};
		\node [style=divide] (9) at (1.25, 0) {};
		\node [style=gather] (10) at (3.5, 0) {};
		\node [style=none] (11) at (4.25, 0) {};
		\node [style=none] (12) at (0.5, 0) {};
	\end{pgfonlayer}
	\begin{pgfonlayer}{edgelayer}
		\draw [style=background] (5.center)
			 to (6.center)
			 to (7.center)
			 to (8.center)
			 to cycle;
		\draw (9)
			 to [in=180, out=-60] (1.center)
			 to (3.center)
			 to (2.center)
			 to [in=-120, out=0] (10);
		\draw [style=very thick] (12.center) to (9);
		\draw [style=very thick] (10) to (11.center);
		\draw (10)
			 to [in=0, out=120] (4.center)
			 to (0.center)
			 to [in=60, out=-180] (9);
	\end{pgfonlayer}
\end{tikzpicture}
 
    \stackeqmid{\GAAidoneii \\ \GAAidoneiii}
      \begin{tikzpicture}
	\begin{pgfonlayer}{nodelayer}
		\node [style=none] (0) at (2, 0.5) {};
		\node [style=none] (1) at (2, -0.5) {};
		\node [style=none] (2) at (2.75, -0.5) {};
		\node [style=none] (3) at (2.375, -0.5) {};
		\node [style=none] (4) at (2.75, 0.5) {};
		\node [style=none] (5) at (0.5, 1) {};
		\node [style=none] (6) at (4.25, 1) {};
		\node [style=none] (7) at (4.25, -1) {};
		\node [style=none] (8) at (0.5, -1) {};
		\node [style=divide] (9) at (1.25, 0) {};
		\node [style=gather] (10) at (3.5, 0) {};
		\node [style=none] (11) at (4.25, 0) {};
		\node [style=none] (12) at (0.5, 0) {};
	\end{pgfonlayer}
	\begin{pgfonlayer}{edgelayer}
		\draw [style=background] (5.center)
			 to (6.center)
			 to (7.center)
			 to (8.center)
			 to cycle;
		\draw (9)
			 to [in=180, out=-60] (1.center)
			 to (3.center)
			 to (2.center)
			 to [in=-120, out=0] (10);
		\draw [style=very thick] (12.center) to (9);
		\draw [style=very thick] (10) to (11.center);
		\draw (10)
			 to [in=0, out=120] (4.center)
			 to (0.center)
			 to [in=60, out=-180] (9);
	\end{pgfonlayer}
\end{tikzpicture} 
    = \interp{\begin{tikzpicture}
	\begin{pgfonlayer}{nodelayer}
		\node [style=none] (0) at (0, 0) {};
		\node [style=none] (1) at (1.5, 0) {};
		\node [style=none] (2) at (0, 0.5) {};
		\node [style=none] (3) at (1.5, 0.5) {};
		\node [style=none] (4) at (1.5, -0.5) {};
		\node [style=none] (5) at (0, -0.5) {};
	\end{pgfonlayer}
	\begin{pgfonlayer}{edgelayer}
		\draw [style=background] (5.center)
			 to (2.center)
			 to (3.center)
			 to (4.center)
			 to cycle;
		\draw (0.center) to (1.center);
	\end{pgfonlayer}
\end{tikzpicture}}_\gaa^\zx
  \end{align}

  (\TextHBox)
  If \(a \neq 0\),
  \begin{align}\label{lem:box}
    \interp{\input{figures/soundness/box/box_RHS.tikz}}_\gaa^\zx
    &= \input{figures/soundness/box/box_1.tikz} 
    \stackeqmid{\minilemref{lem:affinediscard}\\ \GAAidoneii} \input{figures/soundness/box/box_2.tikz} \\
    \stackeq{\GAACopyiii \\ \GAACopyiv}
      \input{figures/soundness/box/box_3.tikz} \\ 
    \stackeq{\minilemref{lem:inversetranspose}} \input{figures/soundness/box/box_4.tikz} \\
    \stackeq{\minilemref{lem:cnotswap}} \input{figures/soundness/box/box_5.tikz} 
    = \input{figures/soundness/box/box_6.tikz} 
    = \interp{\input{figures/soundness/box/box_LHS.tikz}}_\gaa^\zx
  \end{align}

  (\TextTimes)
  If \(a,b \neq 0\),
  \begin{align}\label{lem:times}
    \interp{\begin{tikzpicture}
	\begin{pgfonlayer}{nodelayer}
		\node [style=none] (0) at (0, 0) {};
		\node [style=had] (1) at (2.75, 0) {$b$};
		\node [style=had] (2) at (0.75, 0) {$a$};
		\node [style=none] (3) at (3.5, 0) {};
		\node [style=had] (4) at (1.75, 0) {};
		\node [style=none] (5) at (0, 0.5) {};
		\node [style=none] (6) at (3.5, 0.5) {};
		\node [style=none] (7) at (3.5, -0.5) {};
		\node [style=none] (8) at (0, -0.5) {};
	\end{pgfonlayer}
	\begin{pgfonlayer}{edgelayer}
		\draw [style=background] (6.center)
			 to (7.center)
			 to (8.center)
			 to (5.center)
			 to cycle;
		\draw (2) to (0.center);
		\draw (2) to (1);
		\draw (3.center) to (1);
	\end{pgfonlayer}
\end{tikzpicture}}_\gaa^\zx
    &= \begin{tikzpicture}
	\begin{pgfonlayer}{nodelayer}
		\node [style=none] (0) at (2.25, -0.75) {};
		\node [style=none] (1) at (2.25, 0.75) {};
		\node [style=none] (2) at (4.25, -0.75) {};
		\node [style=none] (3) at (4.25, 0.75) {};
		\node [style=none] (4) at (5.25, -0.75) {};
		\node [style=none] (5) at (5.25, 0.75) {};
		\node [style=rmat] (6) at (4.75, 0.75) {$\minu a$};
		\node [style=lmat] (7) at (4.75, -0.75) {$a$};
		\node [style=none] (8) at (5.25, -0.75) {};
		\node [style=none] (9) at (5.25, 0.75) {};
		\node [style=none] (10) at (7.25, -0.75) {};
		\node [style=none] (11) at (7.25, 0.75) {};
		\node [style=none] (12) at (7.25, -0.75) {};
		\node [style=none] (13) at (7.25, 0.75) {};
		\node [style=none] (15) at (7.25, 0.75) {};
		\node [style=none] (16) at (9.25, -0.75) {};
		\node [style=none] (17) at (9.25, 0.75) {};
		\node [style=wn] (18) at (7.25, 0.75) {};
		\node [style=rmat] (19) at (10, 0.75) {$\minu b$};
		\node [style=lmat] (20) at (10, -0.75) {$b$};
		\node [style=none] (21) at (0.5, 1.25) {};
		\node [style=none] (22) at (12.5, 1.25) {};
		\node [style=none] (23) at (12.5, -1.25) {};
		\node [style=none] (24) at (0.5, -1.25) {};
		\node [style=gather] (25) at (11.75, 0) {};
		\node [style=divide] (26) at (1.25, 0) {};
		\node [style=none] (27) at (12.5, 0) {};
		\node [style=none] (28) at (0.5, 0) {};
	\end{pgfonlayer}
	\begin{pgfonlayer}{edgelayer}
		\draw [style=background] (24.center)
			 to (21.center)
			 to (22.center)
			 to (23.center)
			 to cycle;
		\draw (26)
			 to [in=180, out=-60] (0.center)
			 to [in=-180, out=0] (3.center)
			 to (6);
		\draw (26)
			 to [in=-180, out=60] (1.center)
			 to [in=-180, out=0] (2.center)
			 to (7);
		\draw (7)
			 to (8.center)
			 to [in=-180, out=0] (11.center);
		\draw (6) to (9.center);
		\draw [in=-180, out=0] (9.center) to (10.center);
		\draw (12.center)
			 to [in=-180, out=0] (17.center)
			 to (19);
		\draw (20) to (16.center);
		\draw [in=0, out=-180] (16.center) to (15.center);
		\draw [in=0, out=135] (25) to (19);
		\draw [in=-135, out=0] (20) to (25);
		\draw [style=very thick] (27.center) to (25);
		\draw [style=very thick] (26) to (28.center);
	\end{pgfonlayer}
\end{tikzpicture} \\ 
    &= \begin{tikzpicture}
	\begin{pgfonlayer}{nodelayer}
		\node [style=none] (2) at (0.75, 0.75) {};
		\node [style=none] (3) at (0.75, -0.75) {};
		\node [style=rmat] (6) at (1.5, -0.75) {$\minu a$};
		\node [style=lmat] (7) at (1.5, 0.75) {$a$};
		\node [style=none] (12) at (3.25, -0.75) {};
		\node [style=none] (13) at (3.25, 0.75) {};
		\node [style=none] (14) at (3.25, -0.75) {};
		\node [style=none] (15) at (3.25, 0.75) {};
		\node [style=none] (16) at (5.25, -0.75) {};
		\node [style=none] (17) at (5.25, 0.75) {};
		\node [style=none] (18) at (6.75, -0.75) {};
		\node [style=none] (19) at (6.75, 0.75) {};
		\node [style=rmat] (20) at (5.75, 0.75) {$\minu b$};
		\node [style=lmat] (21) at (5.75, -0.75) {$b$};
		\node [style=wn] (22) at (2.75, 0.75) {};
		\node [style=none] (23) at (-0.75, 1.25) {};
		\node [style=none] (24) at (8.25, 1.25) {};
		\node [style=none] (25) at (8.25, -1.25) {};
		\node [style=none] (26) at (-0.75, -1.25) {};
		\node [style=divide] (27) at (0, 0) {};
		\node [style=gather] (28) at (7.5, 0) {};
		\node [style=none] (29) at (8.25, 0) {};
		\node [style=none] (30) at (-0.75, 0) {};
	\end{pgfonlayer}
	\begin{pgfonlayer}{edgelayer}
		\draw [style=background] (25.center)
			 to (26.center)
			 to (23.center)
			 to (24.center)
			 to cycle;
		\draw (28)
			 to [in=0, out=-120] (18.center)
			 to (21.center)
			 to (16.center)
			 to [in=0, out=-180] (15.center)
			 to (22.center)
			 to (7.center)
			 to (2.center)
			 to [in=60, out=-180] (27);
		\draw (28)
			 to [in=0, out=120] (19.center)
			 to (20.center)
			 to (17.center)
			 to [in=0, out=-180] (14.center)
			 to (6.center)
			 to (3.center)
			 to [in=-60, out=180] (27);
		\draw [style=very thick] (29.center) to (28);
		\draw [style=very thick] (27) to (30.center);
	\end{pgfonlayer}
\end{tikzpicture} \\
    \stackeq{\minilemref{lem:antipode}\\ \GAAMult} \begin{tikzpicture}
	\begin{pgfonlayer}{nodelayer}
		\node [style=none] (0) at (0.5, -0.75) {};
		\node [style=none] (1) at (0.5, 0.75) {};
		\node [style=none] (2) at (2.5, -0.75) {};
		\node [style=none] (3) at (2.5, 0.75) {};
		\node [style=none] (4) at (4.25, -0.75) {};
		\node [style=none] (5) at (4.25, 0.75) {};
		\node [style=rmat] (6) at (3, 0.75) {$ab$};
		\node [style=lmat] (7) at (3, -0.75) {$\minu ab$};
		\node [style=none] (8) at (-1.25, 1.25) {};
		\node [style=none] (9) at (5.75, 1.25) {};
		\node [style=none] (10) at (5.75, -1.25) {};
		\node [style=none] (11) at (-1.25, -1.25) {};
		\node [style=divide] (12) at (-0.5, 0) {};
		\node [style=gather] (13) at (5, 0) {};
		\node [style=none] (14) at (5.75, 0) {};
		\node [style=none] (15) at (-1.25, 0) {};
	\end{pgfonlayer}
	\begin{pgfonlayer}{edgelayer}
		\draw [style=background] (9.center)
			 to (10.center)
			 to (11.center)
			 to (8.center)
			 to cycle;
		\draw (12)
			 to [in=180, out=-60] (0.center)
			 to [in=-180, out=0] (3.center)
			 to (5.center)
			 to [in=120, out=0] (13);
		\draw (12)
			 to [in=-180, out=60] (1.center)
			 to [in=-180, out=0] (2.center)
			 to (4.center)
			 to [in=-120, out=0, looseness=0.75] (13);
		\draw [style=very thick] (13) to (14.center);
		\draw [style=very thick] (15.center) to (12);
	\end{pgfonlayer}
\end{tikzpicture}
 
    = \interp{\begin{tikzpicture}
	\begin{pgfonlayer}{nodelayer}
		\node [style=had] (0) at (1.475, 0) {$\textrm{-}ab$};
		\node [style=none] (1) at (3, 0) {};
		\node [style=none] (2) at (0, 0) {};
		\node [style=none] (3) at (0, 0.5) {};
		\node [style=none] (4) at (3, 0.5) {};
		\node [style=none] (5) at (3, -0.5) {};
		\node [style=none] (6) at (0, -0.5) {};
	\end{pgfonlayer}
	\begin{pgfonlayer}{edgelayer}
		\draw [style=background] (6.center)
			 to (3.center)
			 to (4.center)
			 to (5.center)
			 to cycle;
		\draw (1.center) to (0);
		\draw (2.center) to (0);
	\end{pgfonlayer}
\end{tikzpicture}}_\gaa^\zx
  \end{align}

  (\TextPlus)
  \begin{align}\label{lem:plus}
    \interp{\begin{tikzpicture}
	\begin{pgfonlayer}{nodelayer}
		\node [style=none] (0) at (0.5, 0) {};
		\node [style=gn] (1) at (1.5, 0) {};
		\node [style=gn] (2) at (4, 0) {};
		\node [style=none] (3) at (5, 0) {};
		\node [style=had] (4) at (2.75, 0.75) {$a$};
		\node [style=had] (5) at (2.75, -0.75) {$b$};
		\node [style=none] (6) at (0.5, 1.25) {};
		\node [style=none] (7) at (5, 1.25) {};
		\node [style=none] (8) at (5, -1.25) {};
		\node [style=none] (9) at (0.5, -1.25) {};
	\end{pgfonlayer}
	\begin{pgfonlayer}{edgelayer}
		\draw [style=background] (8.center)
			 to (9.center)
			 to (6.center)
			 to (7.center)
			 to cycle;
		\draw (1) to (0.center);
		\draw (2) to (3.center);
		\draw [bend left] (1) to (4);
		\draw [bend left] (4) to (2);
		\draw [bend left] (2) to (5);
		\draw [bend left] (5) to (1);
	\end{pgfonlayer}
\end{tikzpicture}}_\gaa^\zx
    &= \begin{tikzpicture}
	\begin{pgfonlayer}{nodelayer}
		\node [style=none] (0) at (2.5, -0.75) {};
		\node [style=bn] (1) at (2.5, -0.75) {};
		\node [style=bn] (2) at (8.25, -0.75) {};
		\node [style=none] (3) at (9, -0.75) {};
		\node [style=none] (4) at (2.5, 1.25) {};
		\node [style=wn] (5) at (2.5, 1.25) {};
		\node [style=wn] (6) at (8.25, 1.25) {};
		\node [style=none] (7) at (4, 0.5) {};
		\node [style=wn] (8) at (4, 2) {};
		\node [style=none] (9) at (6, 0.5) {};
		\node [style=none] (10) at (6, 2) {};
		\node [style=lmat] (14) at (6.5, 0.5) {$a$};
		\node [style=none] (15) at (7, 0.5) {};
		\node [style=none] (16) at (7, 2) {};
		\node [style=none] (17) at (6, -2) {};
		\node [style=none] (18) at (6, -0.5) {};
		\node [style=lmat] (23) at (6.625, -2) {$b$};
		\node [style=rmat] (24) at (6.625, -0.5) {$\minu b$};
		\node [style=none] (25) at (4, -2) {};
		\node [style=wn] (26) at (4, -0.5) {};
		\node [style=wn] (27) at (9, 1.25) {};
		\node [style=none] (28) at (7.25, -0.5) {};
		\node [style=none] (29) at (7.25, -2) {};
		\node [style=none] (30) at (11, -2.75) {};
		\node [style=none] (31) at (0.5, 2.75) {};
		\node [style=none] (32) at (11, 2.75) {};
		\node [style=none] (33) at (0.5, -2.75) {};
		\node [style=gather] (34) at (10.25, 0.25) {};
		\node [style=divide] (35) at (1.25, 0.25) {};
		\node [style=none] (36) at (11, 0.25) {};
		\node [style=none] (37) at (0.5, 0.25) {};
		\node [style=rmat] (38) at (6.5, 2) {$\minu a$};
	\end{pgfonlayer}
	\begin{pgfonlayer}{edgelayer}
		\draw [style=background] (32.center)
			 to (30.center)
			 to (33.center)
			 to (31.center)
			 to cycle;
		\draw (2)
			 to (3.center)
			 to [in=-120, out=0] (34);
		\draw (6)
			 to [in=0, out=150] (16.center)
			 to (38.center)
			 to (10.center)
			 to [in=0, out=-180] (7.center)
			 to [in=45, out=180] (1);
		\draw (2)
			 to [in=0, out=120] (15.center)
			 to (14.center)
			 to (9.center)
			 to [in=0, out=-180] (8.center)
			 to [in=30, out=180] (5);
		\draw (1)
			 to [in=180, out=-45] (25.center)
			 to [in=-180, out=0] (18.center)
			 to (24.center)
			 to (28.center)
			 to [in=240, out=0, looseness=0.75] (6);
		\draw (5)
			 to [in=180, out=-60] (26.center)
			 to [in=180, out=0] (17.center)
			 to (23.center)
			 to (29.center)
			 to [in=225, out=0, looseness=0.75] (2);
		\draw [style=very thick] (37.center) to (35);
		\draw [style=very thick] (36.center) to (34);
		\draw [in=0, out=120] (34) to (27);
		\draw (27) to (6);
		\draw [in=-180, out=75] (35) to (4.center);
		\draw [in=180, out=-75] (35) to (0.center);
	\end{pgfonlayer}
\end{tikzpicture}\\
    \stackeq{\GAAFusionW \\ \GAAidoneii} \begin{tikzpicture}
	\begin{pgfonlayer}{nodelayer}
		\node [style=none] (0) at (2.5, -0.75) {};
		\node [style=bn] (1) at (2.5, -0.75) {};
		\node [style=bn] (2) at (8.25, -0.75) {};
		\node [style=none] (3) at (9, -0.75) {};
		\node [style=none] (4) at (2.5, 1.25) {};
		\node [style=wn] (5) at (2.5, 1.25) {};
		\node [style=wn] (6) at (8.25, 1.25) {};
		\node [style=none] (7) at (4, 0.5) {};
		\node [style=none] (8) at (4, 2) {};
		\node [style=none] (9) at (6, 0.5) {};
		\node [style=none] (10) at (6, 2) {};
		\node [style=lmat] (11) at (6.5, 0.5) {$a$};
		\node [style=none] (12) at (7, 0.5) {};
		\node [style=none] (13) at (7, 2) {};
		\node [style=none] (14) at (6, -2) {};
		\node [style=none] (15) at (6, -0.5) {};
		\node [style=lmat] (16) at (6.625, -2) {$b$};
		\node [style=rmat] (17) at (6.625, -0.5) {$\minu b$};
		\node [style=none] (18) at (4, -2) {};
		\node [style=none] (19) at (4, -0.5) {};
		\node [style=wn] (20) at (9, 1.25) {};
		\node [style=none] (21) at (7.25, -0.5) {};
		\node [style=none] (22) at (7.25, -2) {};
		\node [style=none] (23) at (11, -2.75) {};
		\node [style=none] (24) at (-0.25, 2.75) {};
		\node [style=none] (25) at (11, 2.75) {};
		\node [style=none] (26) at (-0.25, -2.75) {};
		\node [style=gather] (27) at (10.25, 0.25) {};
		\node [style=divide] (28) at (0.5, 0.25) {};
		\node [style=none] (29) at (11, 0.25) {};
		\node [style=none] (30) at (-0.25, 0.25) {};
		\node [style=rmat] (31) at (6.5, 2) {$\minu a$};
		\node [style=wn] (32) at (1.75, 1.25) {};
		\node [style=none] (33) at (1.75, -0.75) {};
	\end{pgfonlayer}
	\begin{pgfonlayer}{edgelayer}
		\draw [style=background] (25.center)
			 to (23.center)
			 to (26.center)
			 to (24.center)
			 to cycle;
		\draw (2)
			 to (3.center)
			 to [in=-120, out=0] (27);
		\draw (6)
			 to [in=0, out=150] (13.center)
			 to (31.center)
			 to (10.center)
			 to [in=0, out=-180] (7.center)
			 to [in=45, out=180] (1);
		\draw (2)
			 to [in=0, out=120] (12.center)
			 to (11.center)
			 to (9.center)
			 to [in=0, out=-180] (8.center)
			 to [in=30, out=180] (5);
		\draw (1)
			 to [in=180, out=-45] (18.center)
			 to [in=-180, out=0] (15.center)
			 to (17.center)
			 to (21.center)
			 to [in=240, out=0, looseness=0.75] (6);
		\draw (5)
			 to [in=180, out=-60] (19.center)
			 to [in=180, out=0] (14.center)
			 to (16.center)
			 to (22.center)
			 to [in=225, out=0, looseness=0.75] (2);
		\draw [style=very thick] (30.center) to (28);
		\draw [style=very thick] (29.center) to (27);
		\draw [in=0, out=120] (27) to (20);
		\draw (20) to (6);
		\draw [in=-180, out=60] (28) to (32);
		\draw (32) to (5);
		\draw (28)
			 to [in=180, out=-60] (33.center)
			 to (1);
	\end{pgfonlayer}
\end{tikzpicture} \\ 
    \stackeq{\GAAAdd} \begin{tikzpicture}
	\begin{pgfonlayer}{nodelayer}
		\node [style=none] (0) at (0.5, -0.75) {};
		\node [style=none] (1) at (0.5, 0.75) {};
		\node [style=none] (2) at (2.5, -0.75) {};
		\node [style=none] (3) at (2.5, 0.75) {};
		\node [style=none] (4) at (6.5, -0.75) {};
		\node [style=none] (5) at (6.5, 0.75) {};
		\node [style=rmat] (6) at (4.5, 0.75) {$\minu(a+b)$};
		\node [style=lmat] (7) at (4.5, -0.75) {$a+b$};
		\node [style=none] (8) at (-1.25, 1.5) {};
		\node [style=none] (9) at (8.25, 1.5) {};
		\node [style=none] (10) at (8.25, -1.5) {};
		\node [style=none] (11) at (-1.25, -1.5) {};
		\node [style=divide] (12) at (-0.5, 0) {};
		\node [style=gather] (13) at (7.5, 0) {};
		\node [style=none] (14) at (8.25, 0) {};
		\node [style=none] (15) at (-1.25, 0) {};
	\end{pgfonlayer}
	\begin{pgfonlayer}{edgelayer}
		\draw [style=background] (8.center)
			 to (9.center)
			 to (10.center)
			 to (11.center)
			 to cycle;
		\draw (13)
			 to [in=0, out=120] (5.center)
			 to (3.center)
			 to [in=0, out=-180] (0.center)
			 to [in=-60, out=180] (12);
		\draw (12)
			 to [in=-180, out=60] (1.center)
			 to [in=-180, out=0] (2.center)
			 to (4.center)
			 to [in=-120, out=0] (13);
		\draw [style=very thick] (15.center) to (12);
		\draw [style=very thick] (14.center) to (13);
	\end{pgfonlayer}
\end{tikzpicture} 
    = \interp{\begin{tikzpicture}
	\begin{pgfonlayer}{nodelayer}
		\node [style=none] (0) at (0, 0) {};
		\node [style=had] (1) at (1.475, 0) {$a+b$};
		\node [style=none] (2) at (3, 0) {};
		\node [style=none] (3) at (0, 0.75) {};
		\node [style=none] (4) at (3, 0.75) {};
		\node [style=none] (5) at (3, -0.75) {};
		\node [style=none] (6) at (0, -0.75) {};
	\end{pgfonlayer}
	\begin{pgfonlayer}{edgelayer}
		\draw [style=background] (5.center)
			 to (6.center)
			 to (3.center)
			 to (4.center)
			 to cycle;
		\draw (0.center) to (1);
		\draw (1) to (2.center);
	\end{pgfonlayer}
\end{tikzpicture}}_\gaa^\zx
  \end{align}

  \begingroup\allowdisplaybreaks
  (\TextColour)
  \begin{align}\label{lem:colour}
    &\interp{\input{figures/soundness/colour/colour_LHS.tikz}}_\gaa^\zx
    \stackeqmid{} \input{figures/soundness/colour/colour_1.tikz} \\
    \stackeq{} \input{figures/soundness/colour/colour_1_1.tikz}
    \stackeqmid{\GAAFusionW} \input{figures/soundness/colour/colour_1_2.tikz}\\
    \stackeq{\minilemref{lem:inversetranspose} \\ \GAACopyiii \\ \GAACopyiv} \input{figures/soundness/colour/colour_1_3.tikz}
    \stackeqmid{\minilemref{lem:antipode}} \input{figures/soundness/colour/colour_1_4.tikz}\\
    \stackeq{\minilemref{lem:inversetranspose}} \input{figures/soundness/colour/colour_1_5.tikz}
    \stackeqmid{\GAAMult} \input{figures/soundness/colour/colour_1_6.tikz}\\
    \stackeq{\GAAMult} \input{figures/soundness/colour/colour_1_7.tikz}
    \stackeqmid{\minilemref{lem:antipode}} \input{figures/soundness/colour/colour_1_8.tikz}\\
    \stackeq{\GAAidonei} \input{figures/soundness/colour/colour_1_9.tikz}
    \stackeqmid{\GAAdiscii} \input{figures/soundness/colour/colour_1_10.tikz} \\
    \stackeq{\GAAFusionW} \input{figures/soundness/colour/colour_3.tikz}
    \stackeqmid{ } \interp{\input{figures/soundness/colour/colour_RHS.tikz}}_\gaa^\zx
  \end{align}
  \endgroup
\end{proof}

%
%\begin{proof}[Proof of theorem~\ref{thm:soundness} (Essential surjectivity)]
%  Essential surjectivity is immediate upon recalling that every symplectic
%  \(\K\)-linear space admits a symplectic (or Darboux) basis of cardinality
%  \(2m\) for some \(m \in \N\), which gives an isomorphism with some
%  \(\K^{2m}\).
%\end{proof}
%
%\begin{proof}[Proof of theorem~\ref{thm:soundness} (Fullness)]
%  Fullness follows from the same argument as \cite{comfort_graphical_2021}:
%  the interpretations of the diagrams
%  \begin{equation}
%    \tikzfig{figures/AffLagRel/generators}
%  \end{equation}
%  give a set of generators for \(\ALR\).
%\end{proof}

\subsection{Proofs of section~\ref{ssec:completeness}}
\label{app:completeness}

\begin{lemma}
  \label{lem:box_product}
  Products of boxes are antipodes: for any \(z \in \K^*\)
  \begin{equation*}
    \input{figures/completeness/box_product.tikz}
  \end{equation*}
\end{lemma}
\begin{proof}
  \begin{equation}
    \input{figures/completeness/box_product_proof.tikz}
  \end{equation}
\end{proof}

\begin{lemma}
  \label{lem:box_identity}
  For any \(z \in \K^*\)
  \begin{equation*}
    \input{figures/completeness/box_identity.tikz}
  \end{equation*}
\end{lemma}
\begin{proof}
  \begin{equation}
    \input{figures/completeness/box_identity_proof.tikz}
  \end{equation}
\end{proof}

\begin{lemma}
  \label{lem:box_inverse}
  \begin{equation*}
    \input{figures/completeness/box_inverse.tikz}
  \end{equation*}
\end{lemma}
\begin{proof}
  \begin{equation}
    \input{figures/completeness/box_inverse_proof.tikz}
  \end{equation}
\end{proof}

\begin{lemma}
  \label{lem:box_antipode}
  Boxes and antipodes commute: for any \(z \in \K^*\),
  \begin{equation*}
    \input{figures/completeness/box_antipode.tikz}
  \end{equation*}
\end{lemma}
\begin{proof}
  \begin{equation}
    \input{figures/completeness/box_antipode_proof.tikz}
  \end{equation}
\end{proof}

\begin{lemma}
  \label{lem:box_opposite_inverse}
  The opposite-inverse box can be written as follows: for any \(z \in \K^*\),
  \begin{equation*}
    \input{figures/completeness/box_opposite_inverse.tikz}
  \end{equation*}
\end{lemma}
\begin{proof}
  \begin{equation}
    \input{figures/completeness/box_opposite_inverse_proof.tikz}
  \end{equation}
\end{proof}

\begin{lemma}
  \label{lem:colour_inverted}
  The ``inverse'' colour change rule is provable from the axioms: for any
  \(z \in \K^*\),
  \begin{equation*}
    \input{figures/completeness/colour_inverted.tikz}
  \end{equation*}
\end{lemma}
\begin{proof}
  \begin{equation}
    \input{figures/completeness/colour_inverted_proof.tikz}
  \end{equation}
\end{proof}

\begin{lemma}
  \label{lem:antipode_spider}
  Spiders absorb antipodes:
  \begin{equation*}
    \input{figures/completeness/antipode_spider.tikz}
  \end{equation*}
\end{lemma}
\begin{proof}
  These follow straightforwardly from lemma~\ref{lem:box_product}, \TextColour
  and lemma~\ref{lem:colour_inverted}.
\end{proof}

\begin{lemma}
  \label{lem:box_zero}
  The ``cut'' rule for the zero box is derivable:
  \begin{equation*}
    \input{figures/completeness/box_zero.tikz}
  \end{equation*}
\end{lemma}
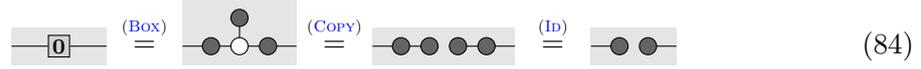
\begin{proof}
  \begin{equation}
    \input{figures/completeness/box_zero_proof.tikz}
  \end{equation}
\end{proof}

\begin{lemma}
  \label{lem:copy_swapped}
  A colour-swapped version of \TextCopy is derivable: for any \(a \in \K\),
  \begin{equation*}
    \input{figures/completeness/copy_swapped.tikz}
  \end{equation*}
\end{lemma}
\begin{proof}
  \begin{equation}
    \input{figures/completeness/copy_swapped_proof.tikz}
  \end{equation}
\end{proof}

\begin{lemma}
  \label{lem:hopf}
  The Hopf identity is derivable:
  \begin{equation*}
    \input{figures/completeness/hopf.tikz}
  \end{equation*}
\end{lemma}
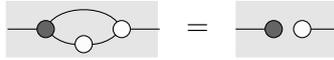
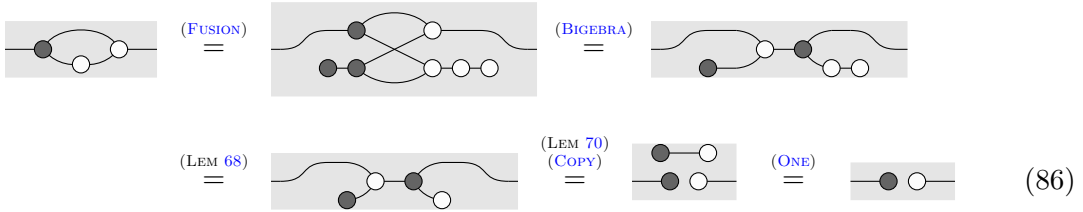
\begin{proof}
  \begin{equation}
    \input{figures/completeness/hopf_proof.tikz}
  \end{equation}
\end{proof}

\begin{lemma}
  \label{lem:push_pauli_state}
  Strictly-affine white states absorb arbitrary phases and vice-versa:
  \begin{equation*}
    \input{figures/completeness/pauli_state_push.tikz}
  \end{equation*}
\end{lemma}
\begin{proof}
  \begin{align}
    &\input{figures/completeness/pauli_state_push_proof.tikz} \\
    &\input{figures/completeness/pauli_state_push_proof_green.tikz}
  \end{align}
\end{proof}

\begin{lemma}
  \label{lem:phase_inverses}
  For any \(a,b \in \K\),
  \begin{equation*}
    \input{figures/completeness/phase_inverses.tikz}
  \end{equation*}
\end{lemma}
\begin{proof}
  \begin{equation}
    \input{figures/completeness/phase_inverses_proof.tikz}
  \end{equation}
  \begin{equation}
    \input{figures/completeness/phase_inverses_proof_2.tikz}
  \end{equation}
\end{proof}

\begin{lemma}
  \label{lem:box_swapped}
  There are four (more or less equivalent) versions of \TextHBox: for any
  \(z \in \K^*\),
  \begin{equation*}
    \input{figures/completeness/box_swapped.tikz}
  \end{equation*}
  The first two equations are also valid for \(z = 0\).
\end{lemma}
\begin{proof}
  The first equation is just \TextHBox.
  \begin{equation}
    \input{figures/completeness/box_swapped_proof_1.tikz}
  \end{equation}
  \vspace{5mm}
  \begin{equation}
    \input{figures/completeness/box_swapped_proof_2.tikz}
  \end{equation}
  \vspace{5mm}
  \begin{equation}
    \input{figures/completeness/box_swapped_proof_3.tikz}
  \end{equation}
\end{proof}

\begin{lemma}
  \label{lem:symplectic_states}
  States with non-zero symplectic part can all be represented using both
  grey and white spiders: for any \(a \in \K\) and \(b \in \K^*\),
  \begin{equation*}
    \input{figures/completeness/symplectic_states.tikz}
  \end{equation*}
\end{lemma}
\begin{proof}
  \begin{equation}
    \input{figures/completeness/symplectic_states_proof_green.tikz}
  \end{equation}
  The second equation follows once again using \TextColour and
  lemma~\ref{lem:colour_inverted}.
\end{proof}

\begin{lemma}
  \label{lem:arrow_transpose}
  For any \(m,n \in \N\) and \(A \in \Matrices\),
  \begin{equation*}
    \input{figures/completeness/scalable/arrow_transpose.tikz}
  \end{equation*}
\end{lemma}
\begin{proof}
  This follows from a straightforward induction on the type \(m \times n\)
  of the matrix \(A\).
%  The base case is given by:
%  \begin{equation}
%    \tikzfig{figures/completeness/scalable/arrow_transpose_proof}
%  \end{equation}
\end{proof}

% \begin{lemma}
%   \label{lem:matrix_arrow_inverse}
%   For any \(A \in \Mmn\),
%   \begin{equation*}
%     \tikzfig{figures/completeness/scalable/matrix_arrow_inverse}
  %   \end{equation*}
% \end{lemma}

\begin{lemma}
  \label{lem:scalable_colour}
  For any \(\vec{x} \in \K^n\) and \(X \in \Sym\),
  \begin{equation*}
    \input{figures/completeness/scalable/colour.tikz}
  \end{equation*}
\end{lemma}
\begin{proof}
  This follows straightforwardly by induction on the size \(n\) of the thick
  wires, using lemma~\ref{lem:arrow_transpose}.
\end{proof}

\begin{lemma}
  \label{lem:scalable_fusion}
  For any \(\vec{x},\vec{y} \in \K^n\) and \(X,Y \in \Sym\),
  \begin{equation*}
    \input{figures/completeness/scalable/fusion.tikz}
  \end{equation*}
\end{lemma}
\begin{proof}
  We prove thickened grey fusion by induction on the type \(n\in\N\) of the thick wires. The base case is trivial. Assume that the inductive claim holds for \(n=k\), and denote this assumption by  \(\textsc{Fusion}_k\) .
  Suppose we have two black spiders as above, of type \(k+1\). Then:
  \begin{equation}
    \input{figures/completeness/scalable/fusion_proof.tikz}
  \end{equation}
  where we assume without loss of generality that
  \begin{equation}
    X = \begin{bmatrix} X_{1,1} & X_{1,\bullet}^\trans \\ X_{1,\bullet} & X' \end{bmatrix}
    \qand
    Y = \begin{bmatrix} Y_{1,1} & Y_{1,\bullet}^\trans \\ Y_{1,\bullet} & Y' \end{bmatrix}
  \end{equation}

The white spider fusion follows from essentially the same argument.
\end{proof}

\begin{lemma}
  \label{lem:scalable_box_loop}
  For any \(n \in \N\) and \(A \in \Matrices[n][n]\),
  \begin{equation*}
    \input{figures/completeness/scalable/box_loop.tikz}
  \end{equation*}
\end{lemma}
\begin{proof}
  We proceed again by induction on the type \(n\) of the thick wires,
  equivalently, the dimension of the matrix \(A\).  The base case is trivial.  Assume that the inductive claim holds for \(n=k\), and denote this by \((\textsc{Loop}_k)\),
  Note that for any \(x \in \K\):
  \begin{equation}
    \label{eq:box_loop_thin}
    \input{figures/completeness/scalable/box_loop_proof_1.tikz}
  \end{equation}
  Now, let \(A \in\Matrices[k+1][k+1]\), and assume without loss of
  generality that it takes the following form for some \(a \in \K, B \in \Matrices[k][1], C \in
     \Matrices[1][k] \qand D \in  \Matrices[k][k]\):
  \begin{equation}
    A =
    \begin{bmatrix}
      a & C \\ B & D 
    \end{bmatrix}
  \end{equation}
Then
  \begin{equation}
    \input{figures/completeness/scalable/box_loop_proof_2.tikz}
  \end{equation}
  where in the last step we have identified
  \begin{equation}
    A + A^\trans = \begin{bmatrix}
      2a & C + B^\trans
      \\ B + C^\trans & D + D^\trans
    \end{bmatrix}
  \end{equation}
\end{proof}

% \begin{lemma}
%   \label{lem:local_complementation_triangle}
%   For any \(x,y,\gamma \in \Z_p^*\),
%   \begin{equation*}
%     \tikzfig{figures/completeness/local_complementation_triangle}
%   \end{equation*}
% \end{lemma}

% \begin{lemma}
%   \label{lem:local_complementation_tree}
%   Let \(A \in M_{1 \times N}(\K)\), then,
%   \begin{equation*}
%     \tikzfig{figures/completeness/local_complementation_tree}
%   \end{equation*}
% \end{lemma}
% \begin{proof}
%   \begin{equation*}
%     \tikzfig{figures/completeness/local_complementation_triangle_proof}
%   \end{equation*}
% \end{proof}

\begin{lemma}
  \label{lem:scalable_arrow_phase}
  For any \(m,n \in \N\), any \(A \in \Matrices\), any \(\vec x \in \K^m\)
  and any \(X \in \Sym[m]\),
  \begin{equation*}
    \input{figures/completeness/scalable/arrow_phase.tikz}
  \end{equation*}
\end{lemma}
\begin{proof}
  The proof is by (double) induction on the shape \(m,n \in \N\) of the
  matrix \(A\). The base case is trivial.

%The base case is given by lemma~\ref{lem:not_sure} and:
%  \begin{equation}
%    \tikzfig{figures/completeness/scalable/arrow_phase_proof_1}
%  \end{equation}

  Now, fix \(m,n\) and assume the lemma is true for any matrix in
  \(\Matrices[j][k]\) where either \(j \leqslant m\) or \(k
  \leqslant n\), and let \(A \in  \Matrices[m+1][n]\). Then,
  denoting the inductive step by \(\textsc{Ind}_{m,n}\),
  \begin{equation}
    \input{figures/completeness/scalable/arrow_phase_proof_2.tikz}
  \end{equation}
  where we have identified
  \begin{equation}
    \vec x = \begin{bmatrix} x_1 \\ \vec  x_\bullet \end{bmatrix}
    \quad\quad\quad
    X = \begin{bmatrix} X_{1,1} & X_{1,\bullet}^\trans \\ X_{1,\bullet} & X_{\bullet,\bullet} \end{bmatrix}
    \quad\quad\quad
    A = \begin{bmatrix} a \\ A_\bullet \end{bmatrix}
  \end{equation}
  so that \(Ax = ax_1 + A_\bullet \vec x_\bullet\) and \(A^\trans XA =
  a^2 X_{1,1} + A_\bullet^\trans X_{\bullet,\bullet}A_\bullet +
  aA_\bullet^\trans X_{1,\bullet} + aX_{1,\bullet}^\trans A_\bullet\).

  Secondly, fix \(m \in \N\) and assume the lemma is true for all matrices
  in \( \Matrices[j][k]\) such that \(j \leqslant m\). Pick
  some \(A \in  \Matrices[1][n+1]\), then if \(A\) is zero, then
  \begin{equation}
    \input{figures/completeness/scalable/arrow_phase_proof_3.tikz}
  \end{equation}
  where the scalable version of \TextId is straightforwardly derivable by
  induction.

  Otherwise, some element of \(A\) is non-zero. For the sake of clarity,
  we assume that this is the first element of \(A\), but since we can
  write any permutation of the elements using the swap, there is no loss of
  generality. Then,
  \begin{equation}
    \input{figures/completeness/scalable/arrow_phase_proof_4.tikz}
  \end{equation}
  where we have identified
  \begin{equation}
    \begin{bmatrix}
      ya^2 & ya A_\bullet^\trans \\
      ya A_\bullet & y A_\bullet^\trans A_\bullet
    \end{bmatrix} = yA^\trans A
  \end{equation}
  This completes the induction.
\end{proof}

\begin{lemma}
  \label{lem:scalable_arrow_label_black}
  For any \(m,n \in \N\), any injective \(A \in  \Matrices\), any \(\vec x \in \K^m\)
  and any \(X \in \Sym\),
  \begin{equation*}
    \input{figures/completeness/scalable/arrow_label_black.tikz}
  \end{equation*}
\end{lemma}
\begin{proof}
  \begin{equation}
    \input{figures/completeness/scalable/arrow_label_black_proof.tikz}
  \end{equation}
\end{proof}

\begin{lemma}
  \label{lem:scalable_arrow_label_white}
  For any \(m,n \in \N\), any surjective \(A \in \Matrices\), any \(\vec x \in \K^n\)
  and any \(X \in \Sym\),
  \begin{equation*}
    \input{figures/completeness/scalable/arrow_label_white.tikz}
  \end{equation*}
\end{lemma}
\begin{proof}
  This follows from lemmas~\ref{lem:scalable_arrow_label_black},
  \ref{lem:arrow_transpose} and \ref{lem:scalable_colour}, and the fact that
  the transpose of a surjective matrix is injective.
\end{proof}

\begin{lemma}
\label{lem:scalable_colour_actual}

\[\input{figures/completeness/scalable/colour.tikz}\]

\end{lemma}
\begin{proof} 
This follows immediately by combining lemmas~\ref{lem:scalable_colour} and~\ref{lem:scalable_arrow_label_white}.
\end{proof}

\begin{lemma}
  \label{lem:box_scalable}
  For any \(X \in \Sym\),
  \begin{equation*}
    \input{figures/completeness/scalable/box.tikz}
  \end{equation*}
  If furthermore, \(X\) is invertible,
  \begin{equation*}
    \input{figures/completeness/scalable/box_invertible.tikz}
  \end{equation*}
\end{lemma}
\begin{proof}
  We start by proving the first equation. The proof is by induction on the
  type \(n \in \N\) of the \(n \times n\) symmetric matrix \(X\). 

The base case is trivial.
%  case is \TextHBox.

 Assume the lemma is true for all \(j < n\) for some \(n
  \in \N\), and pick \(X \in \Sym[n+1]\), which we assume takes the form for some \(X_{1,1} \in \K, X_{1,\bullet} \in \Matrices[n][1], \qand X_{\bullet,\bullet} \in \Matrices[n][n]\):
  \begin{equation}
    X = \begin{bmatrix} X_{1,1} & X_{1,\bullet}^\trans \\ X_{1,\bullet} & X_{\bullet,\bullet} \end{bmatrix}
  \end{equation}
  Now, if \(X_{1,\bullet} = 0\), we have
  \begin{equation}
    \input{figures/completeness/scalable/box_proof_1.tikz}
  \end{equation}
  Otherwise, \(X_{1,\bullet}\) is injective, so that, denoting the inductive
  step by \(\textsc{Ind}_{m}\), 
  \begingroup 
  \allowdisplaybreaks
  \begin{align}
    \input{figures/completeness/scalable/box_proof_2/0.tikz}
      &\stackeqmid{\TMat}
    \input{figures/completeness/scalable/box_proof_2/1.tikz}
      \stackeqmid{\minieqref{eq:scalable_box}}
    \input{figures/completeness/scalable/box_proof_2/2.tikz}\\
      \stackeq{\HBox \\ \text{\tiny{$(\textsc{Ind}_n)$} }}
    \input{figures/completeness/scalable/box_proof_2/3.tikz}
      \stackeqmid{\minilemref{lem:scalable_fusion}}
    \input{figures/completeness/scalable/box_proof_2/4.tikz}\\
      \stackeq{\mini{\Bigebra}}
    \input{figures/completeness/scalable/box_proof_2/5.tikz}
      \stackeqmid{\minilemref{lem:scalable_colour}\\ \minilemref{lem:arrow_transpose}\\ \Fusion \\ \minilemref{lem:scalable_fusion}}
    \input{figures/completeness/scalable/box_proof_2/6.tikz}\\
      \stackeq{\minieqref{eq:thick_spider_def}}
    \input{figures/completeness/scalable/box_proof_2/7.tikz}
      \stackeqmid{\Bigebra}
    \input{figures/completeness/scalable/box_proof_2/8.tikz}\\
      \stackeq{\Bigebra}
    \input{figures/completeness/scalable/box_proof_2/9.tikz}
      \stackeqmid{\minilemref{lem:scalable_fusion}}
		\input{figures/completeness/scalable/box_proof_2/10.tikz}\\
      \stackeq{\tikzeqref{prop:matprop}}
		\input{figures/completeness/scalable/box_proof_2/11.tikz}
      \stackeqmid{\tikzeqref{prop:matprop}}
		\input{figures/completeness/scalable/box_proof_2/12.tikz}\\
     \stackeq{\GAAFusionB\\\GAAFusionW\\ \GAAAdd \\ \GAAZero\\ \GAAidonei\\ \GAAidoneii}
    \input{figures/completeness/scalable/box_proof_2/13.tikz}
        \stackeqmid{\minieqref{eq:thick_spider_def}}
    \input{figures/completeness/scalable/box_proof_2/14.tikz}
  \end{align}
  \endgroup
  This completes the induction.

  The remaining equations follow from manipulations analogous to the proof
  of lemma~\ref{lem:box_swapped}.
\end{proof}

\begin{lemma}
  \label{lem:scalable_symplectic_states}
  For any \(n \in \N\), \(\vec x \in \K^n\) and invertible \(X \in \Sym\),
  \begin{equation*}
    \input{figures/completeness/scalable/symplectic_states.tikz}
  \end{equation*}
\end{lemma}
\begin{proof}
  The proof is analogous to the proof of lemma~\ref{lem:symplectic_states}.
\end{proof}

\begin{proof}[Proof of proposition~\ref{prop:graph_form}]
  Transform the diagram as follows:
  \begin{enumerate}
  \item Use \TextColour to change every white spider into a grey one,
    surrounded by boxes. The resulting diagram consists of only grey
    spiders and boxes.
  \item Use \TextId to add a grey spider between any two subsequent
    boxes. The diagram now consists of only grey spiders, connected either
    by plain edges or boxes.
  \item Use \TextFusion to fuse any two spiders connected by plain edges,
    to eliminate any loops, and use lemma~\ref{lem:scalable_box_loop} to eliminate
    any ``loops'' carrying boxes. The resulting diagram contains no loops, and
    furthermore spiders are connected only by edges carring boxes.
  \item Use \TextPlus to fuse all edges between two spiders into
    a single weighted box or no edge (after a simple application of
    lemma~\ref{lem:box_zero} and \TextFusion).
  \item Use \TextId and lemma~\ref{lem:box_inverse} to obtain
    \begin{equation}
      \input{figures/AffLagRel/spider_graph_split.tikz}\quad,
    \end{equation}
    which allows one to split any spider connected to more than one output into
    several grey spiders, each connected to exactly one output.
  \end{enumerate}
  The resulting diagram is in graph-like form.
\end{proof}

\begin{proof}[Proof of proposition~\ref{prop:symplectic_elimination}]
  \begin{equation}
    \input{figures/completeness/scalable/symplectic_elimination_proof.tikz}
  \end{equation}
\end{proof}

\begin{proof}[Proof of proposition~\ref{prop:local_complementation}]
  \begin{equation}
    \input{figures/completeness/local_complementation_old_proof.tikz}
  \end{equation}
  where \(E \in \Matrices[k][1]\) has components \(E_j\) and \(Z \in \Sym[k]\)
  has components \(Z_{j,k}\). In the last diagram, we have
  \(\gamma_i = z_i + E_i x z^{\minu 1}\),
  \(\delta_i = Z_{i,i} - z^{\minu 1} E_i^2\),
  and \(g_{i,j} = Z_{i j} - z^{\minu 1} E_i E_j\) as required.
\end{proof}

\begin{proof}[Proof of proposition~\ref{prop:pivot}]
  \begin{equation}
    \input{figures/completeness/pivot_proof.tikz}
  \end{equation}
  where we have applied proposiiton~\ref{prop:symplectic_elimination} to the invertible submatrix
  \begin{equation}
    \begin{bmatrix} 0 & \epsilon \\ \epsilon & 0 \end{bmatrix}
    \qq{whose inverse is}
    \begin{bmatrix} 0 & \epsilon^{-1} \\ \epsilon^{-1} & 0 \end{bmatrix}.
  \end{equation}
  The matrix \(E \in \Matrices[k][2]\) has components \(E_{j,k}\) and \(Z
  \in \Sym[k]\) has components \(Z_{j,k}\).  Here, we have
  \(\gamma_i = z_i + \epsilon^{\minu 1} (a E_{2,i} + b E_{1,i})\),
  \(\delta_i = Z_{i,i} - 2 \epsilon^{\minu 1} E_{1,i} E_{2,i}\),
  and \(g_{i,j} = - \epsilon^{\minu 1} (E_{1,i} E_{2,j} + E_{1,j} E_{2,i})\) as requied.
\end{proof}

\begin{proof}[Proof of proposition~\ref{prop:ap_form}]
  Given a graph-like diagram, keep appling
  propositions~\ref{prop:local_complementation} and \ref{prop:pivot} to
  the internal vertices of the diagram as long as it is possible. Every
  application strictly reduces the number of internal vertices hence
  the algorithm must halt. The result is a diagram in AP-form since
  proposition~\ref{prop:local_complementation} eliminates all internal
  vertices with non-zero symplectic phase, and proposition~\ref{prop:pivot}
  eliminates all pairs of connected internal vertices.
\end{proof}

\begin{lemma}
  \label{lem:pauli_push}
  White spiders with arbitrary phases copy affine grey spiders and vice-versa:
  \begin{equation*}
    \input{figures/completeness/extended-copy.tikz}
    \qquad \qquad
    \input{figures/completeness/extended-copy-compl.tikz}
  \end{equation*}
\end{lemma}
\begin{proof}
  First of all,
  \begin{equation}
    % \label{eq:pauli-commute-pf-1}
    \input{figures/completeness/extended-copy-proof-1.tikz}
  \end{equation}
  Then, we separate the equation into two cases based on whether the grey spider
  is affine or not.
  In case $d = 0$, the grey spider is affine and therefore:
  \begin{equation}
    % \label{eq:pauli-commute-pf-2}
    \input{figures/completeness/extended-copy-proof-2.tikz}
  \end{equation}
  Note that if \(d = 0\), then \(ad - c = -c\) and so the lemma holds.
  Otherwise, \(d \neq 0\) and therefore \(d^{\minu 1}\) exists, so we can apply
  the state-change lemma:
  \begin{equation}
    % \label{eq:pauli-commute-pf-3}
    \input{figures/completeness/extended-copy-proof-3.tikz}
  \end{equation}
  Note that the phases after the application of the second state-change
  follow from:
  \begin{equation}
    % \label{eq:pauli-commute-part}
    -(a - c d^{\minu 1})(\minu d^{\minu 1})^{\minu 1},
    \,  \minu(\minu d^{\minu 1})^{\minu 1}
    = -(a - c d^{\minu 1}) (\minu d), \,  d
    = ad - c, \, d
  \end{equation}
  We can prove the second equation of the lemma using Hadamard-boxes as
  follows:
  \begin{equation}
    % \label{eq:pauli-commute-compl-pf}
    \input{figures/completeness/extended-copy-compl-proof.tikz}
  \end{equation}
\end{proof}

\begin{lemma}
  \label{lem:affine_symplectomorphisms}
  For any \(x,y \in \K^n\) and \(X \in \Sym[m]\),
  \begin{equation*}
    \input{figures/completeness/scalable/affine_symplectomorphisms.tikz}
  \end{equation*}
\end{lemma}
\begin{proof}
  We prove the claim by induction on the thickness of scalable spiders.  The base case is trivial.  For the inductive step:
  \begin{equation}
    \input{figures/completeness/scalable/affine_symplectomorphisms_proof.tikz}
  \end{equation}
  %  Assume firstly that \(X\) is invertible. Then,
  %  \begin{equation}
  %    \label{eq:invertible_affine_symplectomorphisms}
  %    \tikzfig{figures/completeness/scalable/affine_symplectomorphisms_proof_1}
  %  \end{equation}
  %  Otherwise, by lemma~\ref{lem:invertible_sum} there is a decomposition \(X =
  %  A + B\) where \(A\) and \(B\) are both symmetric and invertible. Then,
  %  \begin{equation}
  %    \tikzfig{figures/completeness/scalable/affine_symplectomorphisms_proof_2}
  %  \end{equation}
\end{proof}

We can perform row operations on the matrix \(E\) of a
\diagram in AP-form without changing the semantics.  However, applying these row operations can change the affine phases
\(\vec{x}\) of the internal spiders:

\begin{lemma}
  \label{lem:gaussianeliminationi}
  Given a diagram in AP-form and an \emph{invertible} matrix \(A \in
  \Matrices[m][m][\K]\),
  \[
    \input{figures/AffLagRel/scalable/AP-form_GaussianElimination.tikz}
  \]
\end{lemma}
\begin{proof}
  \begin{equation}
    \input{figures/completeness/scalable/AP-form_GaussianElimination_proof.tikz}
  \end{equation}
\end{proof}

\begin{proof}[Proof of proposition~\ref{prop:reduced_ap_form}]
  Consider a diagram in AP-form described by the quadruplet
  \((E,Y,\vec{x},\vec{y})\). Then, by lemma~\ref{lem:gaussianeliminationi}
  we can act on \(E\) on the left by an invertible matrix \(A\). Thus,
  we can run a Gaussian elimination on \(E\). During this procedure,
  if at any point the resulting matrix \(AE\) has a zero row,
  this means that one of the internal vertices has no neighbours in
  the graph. This disconnected internal vertex then takes the form
  \input{figures/completeness/scalable/AP-form_reduction_scalar.tikz}. If
  \(x \neq 0\)  then the diagram has empty semantics and we can use
  lemma~\ref{prop:zero_normal_forms} to rewrite the diagram to normal
  form. Otherwise, \(x=0\) and we can eliminate the disconnected vertex
  using \TextOne.

 The resulting row reduced echelon form of \(E\), with zero rows removed,
 is of the form \(\begin{bmatrix} 1_m & F \end{bmatrix}\), up to a permutation
 of its columns. Denoting this permutation by \(\varsigma\), we have
 \begin{equation}
   \input{figures/completeness/scalable/AP-form_reduction_proof_1.tikz}
 \end{equation}
 for matrices \(A,B,C,F\) and vectors \(\vec{a},\vec{b}\) of suitable
 dimensions. Let's ignore the permutation for now, then following
 equation~\eqref{eq:scalable-AP-form-derivation} we have
 \begin{equation}
   \input{figures/completeness/scalable/AP-form_reduction_proof_2.tikz}
 \end{equation}
 where we have set \(\vec{s} \coloneqq F^\mathsf{T}\vec{a} +
 \vec{b} -  (A+B) \vec{x}\) and \(S \coloneqq F^\mathsf{T}AF + B -
 BF - F^\mathsf{T}B^\mathsf{T}\).  Reintroducing the permutation,
 we see that this final diagram is exactly the form described in
 remark~\ref{rem:scalable_reduced_AP-form}, which is easily seen to verify the
 conditions of definition~\ref{def:reduced_ap_form}. The diagram is therefore
 in reduced AP-form.
\end{proof}

\begin{lemma}
  \label{lem:push_pauli_state_thick}
  The thick scalable version of lemma~\ref{lem:push_pauli_state} holds:
  \begin{equation*}
    \input{figures/completeness/scalable/pauli_state_push.tikz}
  \end{equation*}
\end{lemma}
\begin{proof}
The proof is essentially the same as that of lemma~\ref{lem:push_pauli_state}
\end{proof}

\begin{proof}[Proof of proposition~\ref{prop:ap-unique}]
  Up to permutation the proof of proposition~\ref{prop:reduced_ap_form}, entails that the the AP-form can be rewritten to:
  \begin{equation}
    \input{figures/completeness/scalable/reduced_ap_form_uniqueness_1.tikz}
  \end{equation}
  Note that the number \(m\) of internal vertices is an invariant of the reduced AP-form, as is the permutation of the output wires \(\varsigma\).
  The affine subspace \((F,\vec x)\) is also invariant as:
  \begin{equation}
    \input{figures/completeness/scalable/reduced_ap_form_uniqueness_2.tikz}
  \end{equation}
   As is the affine subspace \((S,\vec s)\):
  \begin{equation}
    \input{figures/completeness/scalable/reduced_ap_form_uniqueness_3.tikz}
  \end{equation}
\end{proof}

\subsection{Proofs of section~\ref{sec:AffCoisoRel}}
\begin{proof}[Proof of Lemma \ref{lem:decomposition}]
  Given a diagram \(L:n\to m\) in \(\ZX\), by proposition~\ref{prop:graph_form}, it can be rewritten into graph-like form.  Moreover, by propositions~\ref{prop:local_complementation} and~\ref{prop:pivot}, we can remove all of the internal edges to obtain a diagram of the following form, for some  \(a\in\K^n\), \(A \in \Sym[n][n]\), \(b\in\K^m\), \(B \in \Sym[m][m]\), and \(E \in \Matrices[m][n]\):
  \begin{equation}
    \input{figures/AffCoisoRel/bipartite1.tikz}
  \end{equation}
  Moreover, any such matrix \(E\) can be factorized into a surjection \(E_s\in \Matrices[k][n]\) followed by an injection \(E_i\in\Matrices[m][k]\) where \(k=|\im (E) |\):
  \begin{equation}
    \input{figures/AffCoisoRel/epimono.tikz}
  \end{equation}
  Where the surjection \(E_s\) can moreover decomposed into an isomorphism \(S\in\Matrices[n][n]\) followed by a canonical projection onto \(k\); and similarly, the injection \(E_i\) can be decomposed into a canonical injection into   \(k\) followed by an isomorphism  \(T\in\Matrices[m][m]\):
  \begin{equation}
    \input{figures/AffCoisoRel/epimonoi.tikz}
  \end{equation}
  Therefore:
  \begin{equation}
    \input{figures/AffCoisoRel/bipartite2.tikz}
  \end{equation}
  Where we remark that \(\im(L)\cong \im(E)\), thus \(|\im(L)|= |\im(E)|=k\), as desired.
\end{proof}

\begin{proof}[Proof of lemma~\ref{lem:enough_iso}]
Take two arrows \(L:n\to a\) and \(K:n\to b\) in  \(\ZX\) such that \(L^\dag L = K^\dag K\) and \(a\leq b\).
By lemma~\ref{lem:decomposition}, there are \(\ell,k \in \N \) and affine symplectomorphisms \(S,X :n\to n\), \(T:a\to a\), and \(T:b\to b\) such that:

  \begin{equation}
    \label{eq:lhepimono}
    \input{figures/AffCoisoRel/supposition.tikz}
  \end{equation}

  Because \(L^\dag L = K^\dag K\), by lemma~\ref{lem:coisotropicimage}:

  \begin{equation}
    \ell=\left|\im(L^\dag)\right|=\left|\im( K^\dag)\right|=k
  \end{equation}

  Moreover: because $L^\dagger L = K^\dagger K$, we have:
  \begin{equation}
  \label{eq:positivemaps}
    \input{figures/AffCoisoRel/enoughisoproof3.tikz}
  \end{equation}

  Furthermore, remark that:
  
  \begin{equation}
  \label{eq:changeofbasis}
    \input{figures/AffCoisoRel/changeofbasisproof.tikz}
  \end{equation}
  
  Therefore, the following map is unitary (equivalently, it is an affine symplectomorphism):

  \begin{equation}
    \input{figures/AffCoisoRel/changeofbasis.tikz}
  \end{equation}

  And thus, the following map is unitary as well:

  \begin{equation}
    \label{eq:h}
    \input{figures/AffCoisoRel/h.tikz}
  \end{equation}

   Therefore:

  \begin{equation}
    \input{figures/AffCoisoRel/proof.tikz}
  \end{equation}
\end{proof}

\end{document}